\documentclass[aps,twocolumn,raggedbottom,superscriptaddress,tightenlines,pra,10pt]{revtex4-1}
\usepackage[utf8]{inputenc}
\usepackage{amsthm}
\usepackage{amsmath}
\usepackage{amssymb}
\usepackage{multirow}
\usepackage{mathtools}
\usepackage{graphicx}
\usepackage[font={small}]{caption}
\usepackage{subcaption}
\usepackage{epstopdf}
\usepackage{url}
\usepackage{float}
\usepackage{bm}
\usepackage{ifthen}
\usepackage[usenames,dvipsnames]{color}
\usepackage{mathrsfs}
\usepackage{blkarray}
\usepackage{natbib}
\usepackage[colorlinks=true,citecolor=blue,urlcolor=blue]{hyperref}
\usepackage{epstopdf}
\usepackage{braket}
\usepackage{array}
\newcommand{\leakEC}{\text{leak}_{\text{EC}}}

\newcommand{\epscor}{\varepsilon_{\text{cor}}}
\newcommand{\epshash}{\varepsilon_{\text{hash}}}
\newcommand{\epssec}{\varepsilon_{\text{sec}}}
\newcommand{\epsqkd}{\varepsilon_{\text{qkd}}}
\newcommand{\epspa}{\varepsilon_{\text{PA}}}
\newcommand{\epspe}{\varepsilon_{\text{PE}}}
\newcommand{\epsec}{\varepsilon_{\text{EC}}}
\newcommand{\epscoh}{\varepsilon_{\text{sec}}^{\text{coh}}}

\newcommand{\fEC}{f_{\text{EC}}}
\newcommand{\ppass}{p_{\text{pass}}}
\newcommand{\ratioPE}{\alpha_{\text{PE}}}
\newcommand{\nPE}{n_{\text{PE}}}

\newcommand*{\rom}[1]{\uppercase\expandafter{\romannumeral #1\relax}}

\renewcommand\bra[1]{{\langle{#1}|}}
\renewcommand\ket[1]{{|{#1}\rangle}}

\newtheorem{theorem}{Theorem}

\newtheorem{lemma}[theorem]{Lemma}

\theoremstyle{definition}
\newtheorem*{primal}{Primal problem}
\newtheorem*{dual}{Dual problem}

\numberwithin{theorem}{section}
\numberwithin{corollary}{section}

\newcommand{\herm}{\text{Herm}}
\DeclareMathOperator{\Tr}{Tr}

\DeclareMathOperator{\LB}{\text{LB}}
\DeclareMathOperator{\UB}{\text{UB}}
\DeclareMathOperator{\tin}{\text{in}}
\DeclareMathOperator{\tout}{\text{out}}
\DeclareMathOperator{\id}{\mathbb{I}}

\DeclarePairedDelimiter{\ceil}{\lceil}{\rceil}

\DeclarePairedDelimiter\abs{\lvert}{\rvert}
\DeclarePairedDelimiterX{\infdivx}[2]{(}{)}{%
  #1\;\delimsize\|\;#2%
}

\newcommand{\infdiv}{D\infdivx}

\newcommand{\kb}[1]{\ket{#1}\bra{#1}}
\newcommand{\matr}[1]{\mathbf{#1}} % undergraduate algebra version
\begin{document}

\newcommand{\RLE}{\affiliation{Research Laboratory of Electronics, Massachusetts Institute of Technology, Cambridge, Massachusetts 02139, USA}}
\newcommand{\BBN}{\affiliation{Raytheon BBN Technologies, 10 Moulton St., Cambridge, Massachuseetts 02138, USA}}

\title{Numerical finite-key analysis of quantum key distribution}
\author{Darius Bunandar}
\email{dariusb@mit.edu}
\RLE
\author{Luke C.~G.~Govia}
\BBN
\author{Hari Krovi}
\BBN
\author{Dirk Englund}
\RLE
\date{July 2019}

\begin{abstract}
    Quantum key distribution (QKD) allows for secure communications safe against attacks by quantum computers. QKD protocols are performed by sending a sizeable, \emph{but finite}, number of quantum signals between the distant parties involved. Many QKD experiments however predict their achievable key rates using asymptotic formulas, which assume the transmission of an infinite number of signals, partly because QKD proofs with finite transmissions (and finite key lengths) can be difficult. Here we develop a robust numerical approach for calculating the key rates for QKD protocols in the finite-key regime in terms of two novel semi-definite programs (SDPs). The first uses the relation between smooth min-entropy and quantum relative entropy, and the second uses the relation between the smooth min-entropy and quantum fidelity. We then solve these SDPs using convex optimization solvers and obtain some of the first numerical calculations of finite key rates for several different protocols, such as BB84, B92, and twin-field QKD. Our numerical approach democratizes the composable security proofs for QKD protocols where the derived keys can be used as an input to another cryptosystem.
\end{abstract}
\maketitle

\section*{Introduction}

Quantum key distribution (QKD), until today, remains the only quantum-resistant method of sharing secret keys and transmitting future-proof secret information at a distance~\citep{pir2019advances}. Even after more than 30 years of development, QKD still does not see widespread adoption, primarily due to practical and theoretical difficulties. 

Practically, QKD as a means of encryption requires a dramatic change to the existing classical fiber-optical communication infrastructure. QKD systems typically require the use of specialized quantum optical devices. For example, some QKD protocols need single photon detectors and dark fiber-optical channels without any classical repeater device, e.g. erbium-doped fiber amplifiers. In other words, the need for significant change to the current telecommunication infrastructure presents a challenge to QKD's wide-use today. 

Theoretically, security proofs are typically complicated, and the key rate derived can be loose due to limited availability of analytical proof techniques. Validating a published security proof is an equally complicated task, and it is likely impractical to expect QKD users to be capable of verifying the security of a protocol.

The security of any QKD protocol is guaranteed when a detailed security analysis certifies that the protocol produces a non-zero secret key rate (in terms of either bits per second or bits per transmission). So far, development in key rate calculations have relied on analytical tools that can be limited in scope to specific protocols. In particular, oftentimes to simplify the analysis, the calculations invoke a high degree of symmetry. Indeed, for some protocols, such as the Bennett-Brassard 1984 (BB84) protocol~\citep{BB84} or the six-state protocol~\citep{Bruss1998}, analytical formulas for the key rates are known. However, in practical implementations of QKD, a lack of symmetry is the norm rather than the exception as experimental imperfections tend to break these symmetries~\citep{Gottesman2004a}. This motivates the need to develop a new method of analyzing the security of QKD protocols that may lack structure.

Recently, \cite{Coles2016} and \cite{Winick2017} proposed two numerical techniques to obtain reliable secret key rate bounds for an arbitrary unstructured QKD protocol. The original technique, described in~\citep{Coles2016}, formulates the problem of calculating the secret key rate in terms of a mathematical optimization problem. Unfortunately, this original formulation resulted in a non-convex problem. The method was improved in~\citep{Winick2017}, which formulates the key rate problem in terms of convex optimization. Commercially available convex optimization tools, such as \texttt{Mosek}~\citep{mosek}, \texttt{SeDuMi}~\citep{sedumi}, or \texttt{SDPT3}~\citep{sdpt3}, can therefore be used to reliably solve the problem. Nevertheless, the problems formulated so far still assumed that Alice and Bob have exchanged an infinite number of signals (and an infinite key length), which is practically impossible. In order to quantify the security of realistic QKD protocols, a new problem that includes the finite-key statistics of the QKD operations must be formulated.

Here, we formulate the key rate problem in terms of a novel semidefinite program (SDP) that considers the practical case of only a finite number of transmitted signals. The program takes as inputs the measured statistics from the parameter estimation step and outputs the key rate as a function of the security parameter of the protocol: $\epsqkd$. The SDP computes a reliable, achievable lower bound on the actual value of the secret key rate. As SDP is a convex optimization problem, we can solve the problem using commercial solvers that often are able to find the global optima. Our problem formulation is reliable, such that even if the solver fails to find the global optimum, the SDP is guaranteed to output an achievable secret key rate. Lastly, since the problem takes into consideration the finite number of signals exchanged, the secret key guaranteed by the method is composable, i.e. can be used as an input to another cryptosystem.

\section*{Results}

\subsection*{The key rate problem in the non-asymptotic regime}
\label{sec:the_key_rate_problem}
We describe the main steps of a typical QKD protocol and take note of the relevant security parameters at each step.
\begin{enumerate}
    \item \emph{Transmission.} A QKD protocol starts with the transmission of quantum signals. Let us assume that $N$ signals are successfully distributed from Alice to Bob. After this step, in the entanglement picture, they will share $N$ entangled quantum states, whose joint state can be described as $\rho_{\matr{A} \matr{B}}$. They will then apply measurements to their respective quantum states to obtain classical data. 
    
    \item \emph{Parameter estimation.} Next, Alice and Bob perform parameter estimation where they reveal a random sample of $m$ signals through the public classical communication channel to estimate the statistics of their data. Sifting is often (but not always) also performed, in which Alice and Bob discard those data where they have chosen a different measurement basis. At this step, they are left with $n \leq N - m$ number of signals, called the raw keys, from which they can eventually generate secret keys. Let us denote their raw keys by $\matr{Z_A}$ and $\matr{Z_B}$, both of which have the length $\abs{\matr{Z_A}} = \abs{\matr{Z_B}} = n$. There is a small probability $\epspe$ that the raw keys obtained are not compatible with the estimated parameters, and it is related to the number of signals $m$ that are used to estimate the relevant statistics of the overall data. (\emph{Relevant security parameter:} $\epspe$.)
    
    \item \emph{Reconciliation.} Alice and Bob then perform key reconciliation (sometimes also referred to as error correction). The key reconciliation step, in which they correct for any possible error between their raw keys, reveals $\leakEC$ number of bits. This error correction step is performed with a certain failure probability $\epsec$, which is the probability that one party computes the wrong guess of the other party's raw keys. (\emph{Relevant security parameter:} $\epsec$.)
    
    \item \emph{Error verification.} To ensure they have identical raw keys, they apply a two-universal hash function and publish $\ceil{\log_2 1/\epscor}$ bits of information. Here, $\epscor = \epshash$, which is the probability two non-identical raw secret keys generate the same hash value. (\emph{Relevant security parameter:} $\epscor$.)
    
    \item \emph{Privacy amplification.}  Next, they apply another two-universal hash function (of different resulting hash length than the previous one in the error verification step) to extract a shorter secret key pair $(\matr{S_A}, \matr{S_B})$ of length $\abs{\matr{S_A}} = \abs{\matr{S_B}} = \ell$. $\epspa$ measures how close the output of the hash function, i.e. the secret keys $\matr{S_A}$ and $\matr{S_B}$ are, from a uniform random bit string conditioned on the eavesdropper's, Eve's, knowledge. (\emph{Relevant security parameter:} $\epspa$.)

\end{enumerate}

In this nonasymptotic regime, we use a generalization of the von Neumann entropy called smooth min-entropy which was developed by~\cite{Renner2005a}. The main significance of smooth min-entropy comes from the fact that it characterizes the number of uniform bits that can be extracted in the privacy amplification step of a QKD protocol. Now, let us take $\matr{E'}$ to be the information that Eve gathered about Alice's raw key $\matr{Z_A}$ up to and including the error correction and verification steps. When Alice and Bob apply a two-universal hash function in the privacy amplification step, they can then extract a $\epssec$-secret key of length:
\begin{equation}
	\ell = H^{\bar{\varepsilon}}_{\min} (\matr{Z}_A| \matr{E'}) - 2 \log_2 \frac{1}{2\epspa},
\end{equation}
for $\bar{\varepsilon} + \epspa \leq \epssec$ (Proof is by the Quantum leftover hash lemma in Ref.~\cite{Tomamichel2011a}). $H^{\bar{\varepsilon}}_{\min} (\matr{Z}_A | \matr{E'})$ is the conditional smoothed min-entropy that quantifies the average probability that Eve guesses $\matr{Z}_A$ correctly using her optimal strategy based on her knowledge of the information $\matr{E'}$.

During the error correction step, a maximum of $\leakEC$ bits of information are revealed about $\matr{Z}_A$. Alice has to send a syndrome bit string of length $\leakEC$ to Bob over the public channel, so that Bob can correct his raw key to match Alice's. Furthermore, during the error verification step, $\ceil{\log_2 (1/\epscor)} \leq \log_2 (2/\epscor)$ bits of information are revealed. If we let $\matr{E}$ be the remaining quantum information Eve has on $\matr{Z}_A$, then
\begin{equation}
	 H^{\bar{\varepsilon}}_{\min} (\matr{Z}_A | \matr{E'}) \geq H^{\bar{\varepsilon}}_{\min} (\matr{Z}_A | \matr{E}) - \leakEC - \log_2 \frac{2}{\epscor}.
\end{equation}
The QKD protocol is said to be $\epsqkd \geq \epscor+\epssec$ secure if it is correct with a probability higher than $1-\epscor$ and is secret with a probability higher than $1-\epssec$.

The quantity $H^{\bar{\varepsilon}}_{\min} (\matr{Z}_A | \matr{E})$ can be simplified in the case of collective attacks, in which Alice and Bob share the state of the form $\rho_{\matr{A} \matr{B}} = \left( \rho_{AB} \right)^{\otimes N}$. In this case, we can also assume that $\rho_{\matr{Z}_A \matr{E}} = \left( \rho_{Z_A E} \right)^{\otimes n}$ since all purifications of $\rho_{\matr{A} \matr{B}}$ are equivalent under a local unitary operation by Eve, and there exists a purification with this property~\citep{Scarani2008}. In other words, after being presented with the tensor product states $\left( \rho_{AB} \right)^{\otimes N}$, Eve is free to choose how to purify this state. (She wants to purify this state because this gives her the most information.) One obvious choice is to purify each transmission such that she has $\ket{\Psi}_{\matr{A}\matr{B}\matr{E}} = \ket{\psi}_{A_1B_1E_1} \otimes \ket{\psi}_{A_2B_2E_2} \otimes \dots \otimes \ket{\psi}_{A_NB_NE_N}$. It is from such a purification, we obtain the tensor product structure of $\rho_{\matr{Z}_A \matr{E}} = \left( \rho_{Z_A E} \right)^{\otimes n}$. The observed statistics of relative detection frequencies, however, only gives some knowledge about the state $\rho_{Z_A E}$. Given that the state $\rho_{Z_A E}$ is contained within a set that contains all $\rho_{Z_A E}$ compatible with the observed statistics, except with probability $\epspe$, we have:
\begin{equation}
	H^{\bar{\varepsilon}}_{\min} (\matr{Z}_A | \matr{E}) \geq H^{\varepsilon}_{\min}(\rho_{Z_A E}^{\otimes n} | \rho_{E}^{\otimes n}),
\end{equation}
where $\bar{\varepsilon} = \varepsilon + \epspe$.

When the error correction step is performed with a failure probability of $\epsec$, i.e. the probability that Bob computes the wrong guess for $\matr{Z}_A$, we can bound the quantity $\leakEC$ with Corrollary 6.3.5 of~\citep{Renner2005a}:
\begin{equation}
	 \frac{1}{n} \leakEC \leq \fEC H(Z_A|Z_B) + \log_2(d + 3) \sqrt{\frac{3\log_2 (2/\epsec)}{n}},
\end{equation}
where $d$ is the number of possible symbols in $Z_A$, and $\fEC \geq 1$ characterizes the error correction (in)efficiency. Commonly, $\fEC$ is chosen to be $\sim 1.2$, which is based on the performance of real codes~\citep{Scarani2008}.

To compute the key rate under the assumption of collective attacks, we therefore have to minimize the quantity:
\begin{equation}
	\min_{\rho_{AB} \in \mathcal{C}_{\epspe}} H^{\varepsilon}_{\min}(\rho_{Z_A E }^{\otimes n} | \rho_E^{\otimes n})
\end{equation}
Here, the set $\mathcal{C}_{\epspe}$ is the set of all density operators $\rho_{AB}$ that are consistent with the statistics measured from the parameter estimation step, except with a probability $\epspe$. Let $\Gamma_i$ be the Hermitian observables for these measurements, then the average values of these operators are within the bounds: $\gamma_i^{LB} \leq \Tr(\rho_{AB} \Gamma_i) \leq \gamma_i^{\UB} $, for $i = 1, \dots, \nPE$, except with probability $\epspe$. Along with the constraint that $\rho_{AB}$ is a valid normalized density operator, i.e. $\rho_{AB} \succeq 0$ and $\Tr(\rho_{AB}) = 1$, then $\rho_{AB}$ is constrained to be in the set:
\begin{equation}
\begin{aligned}
	\mathcal{C}_{\epspe} \equiv \{ \rho_{AB}: \: & \rho_{AB} \succeq 0, \; \Tr(\rho_{AB}) = 1, \; \text{ and} \\ & \gamma_i^{\LB} \leq \Tr(\rho_{AB} \Gamma_i) \leq \gamma_i^{UB} \\ & \text{ for } i=1,\dots,\nPE  \},
\end{aligned}
\end{equation}
except with probability $\epspe$.

To understand how one can obtain the bounds on the average values $\gamma_i \equiv \Tr(\rho_{AB} \Gamma_i)$, consider the parameter estimation step in a typical QKD protocol. Alice and Bob perform the measurements using the POVMs $\set{M^a_A}$ and $\set{M^b_B}$ (in the entanglement-based picture) and use a fraction of their measurements to obtain $\gamma_{(a,b)} = \Tr(\rho_{AB} M^a_A \otimes M^b_B)$. Then, we can make the identification $\Gamma_i \equiv \Gamma_{(a,b)} = M^a_A \otimes M^b_B$ with $\gamma_i \equiv \gamma_{(a,b)}$. To find the relevant bounds, suppose that a total of $m_{i}$ signals have been used to estimate $\gamma_i$, then the deviation of the estimate $\gamma^{m_i}_i$ from the ideal estimate $\gamma^{\infty}_i$ can be quantified using the law of large numbers~\citep{Scarani2008,Cai2009}:
\begin{equation}
	\label{eq:large_numbers}
	\abs{\gamma^{m_i}_i - \gamma^{\infty}_i} \leq \Delta(m_i, d) = \frac{1}{2} \sqrt{\frac{2 \ln (1/\epspe^i) + d \ln(m_i+1)}{m_i}},
\end{equation}
except with a failure probability of $\epspe^i$. Here, $d$ is the number of outcomes of the POVM $\Gamma_i$ needed to estimate it (for error rates, $d=2$ since the outcomes are either Alice $=$ Bob or Alice $\neq$ Bob). The overall parameter estimation step fails with a probability of $\epspe = \sum_i \epspe^i$. We can then obtain the upper and lower bounds:
\begin{equation}
\begin{aligned}
	\gamma_i^{\UB} &= \min(\gamma_i + \Delta(m_i), 1), \\
	\gamma_i^{\LB} &= \max(\gamma_i - \Delta(m_i), 0),
\end{aligned}
\end{equation}
as $\gamma_i$ is a probability and must have values between 0 and 1. We note that the inequality~\eqref{eq:large_numbers} is not the only law of large numbers that can be used to find these bounds. Tighter (asymmetric) bounds can be achieved by applying both the Chernoff bound and the Hoeffding's inequality~\citep{Curty2014,Lim2014}. Recently, even tighter bounds were obtained with clever usage of the Chernoff bound alone~\citep{Zhang2017}.

The definition of $\Gamma_i$ and $\gamma_i$ above may be too fine-grained for a QKD protocol such that each individual $\epspe^i$ may be too small for a given value of $\epspe$. The security of a QKD protocol typically can be defined with only a few parameters; for example, the security of BB84 relies on only the bit error rates when both parties choose the $Z$-basis and the $X$-basis. Coarse-graining the constraints can be achieved by merging the constraints $\Gamma_i$ together, e.g. by summing a subset of or by taking an average value of the constraints and the observed statistics. Coarse-graining, from an optimization perspective, loosens a constraint such that the guaranteed key rate can be lower than the optimal value of the calculations with fine-grained constraints. However, coarse-graining can provide tighter bounds on $\gamma_i$'s for the same value of $\epspe$ that can result in a higher secret key rate.

We now use two methods to evaluate a reliable numerical lower bound on the quantity $H^{\varepsilon}_{\min}(\rho_{Z_A E}^{\otimes n} | \rho_{E}^{\otimes n})$ that will allow us to eventually quantify the key length $\ell$---hence the key rate $r$.

\subsubsection*{Key rate estimation using von Neumann entropy}
\label{sub:von_neumann_entropy}

The smooth min-entropy of an independent and identically distributed product state $\rho_{Z_A E}^{\otimes n}$ converges to the von Neumann entropy in the limit of large $n$:
\begin{equation}
	\lim_{n \rightarrow \infty} \left[ \frac{1}{n} H_{\min}^{\varepsilon} (\rho_{Z_A E}^{\otimes n} | \rho_E^{\otimes n} ) \right] = H(\rho_{Z_A E}|\rho_E) =  H(Z_A|E).
\end{equation}
For the case of finite number of signals $n$, these two entropic quantities are related via a correction factor obtained in Corrollary 3.3.7 of Ref.~\cite{Renner2005a}, i.e.
\begin{equation}
\label{eq:min_entropy_lower_bounded_by_vnm_entropy}
    \frac{1}{n} H_{\min}^{\varepsilon} (\rho_{Z_A E}^{\otimes n} | \rho_E^{\otimes n} ) \geq H(Z_A|E) - \delta(n, \varepsilon),
\end{equation}
where, $\delta(n,\epsilon) = (2d+3)  \sqrt{\log_2(2/\varepsilon)/n}$
is the correction factor. It is worth pointing out that the right hand side of Eq.~\eqref{eq:min_entropy_lower_bounded_by_vnm_entropy} is an achievable secure lower bound to estimate the key rate.
We apply this result to obtain an $\epsqkd$-secure finite-key QKD protocol that is $\epscor$-correct and $\epssec$-secret (with $\epscor + \epssec \leq \epsqkd$) at a secret key rate per transmission of:
\begin{equation}
\label{eq:keyrate_eq_vnm}
\begin{aligned}
    r_1 = \frac{\ell}{N} &= \frac{n}{N} \left\lfloor H(Z_A|E) - \delta(n, \varepsilon) - \frac{1}{n} \leakEC \right. \\ &\qquad \left. - \frac{2}{n} \log_2 \frac{1}{2 \epspa} - \frac{1}{n} \log_2 \frac{2}{\epscor} \right\rfloor,
\end{aligned}
\end{equation}
which is in terms of the von Neumann entropy instead of the smooth min-entropy. The protocol is secret up to a failure probability of $\epssec \geq \varepsilon + \epspa + \epspe + \epsec.$

In light of Eq.~\eqref{eq:keyrate_eq_vnm}, the optimization problem that we have to solve is $\min_{\rho_{AB} \in \mathcal{C}_{\epspe}} H(Z_A|E)$. Ref.~\cite{Coles2016} shows how to recast this as an optimization problem with the quantum relative entropy, rather than the von Neumann entropy, as the objective function. Ref.~\cite{Winick2017} further developed a two-step method to obtain a secure lower bound. Simply put, this two-step method consists of finding an approximate minimum (step one), and then solving a linearized version of the SDP around this approximate minimum to obtain a secure lower bound (step two). In this work we use the semidefinite approximation of the matrix logarithm and quantum relative entropy from Ref.~\cite{Fawzi2017} to peform step one, and then follow step two directly as described in Ref.~\cite{Winick2017}. The Methods section contains a review of the SDP developed in Ref.~\cite{Winick2017}, and a more detailed description of our approach to step one.

\subsubsection*{Key rate estimation using min-entropy}
\label{sub:min_entropy}

To compute the key rate via the min-entropy, we use the fact that the smooth min-entropy is a maximization of the min-entropy and is equal to the min-entropy when the smoothing parameter $\varepsilon = 0$, i.e.
	\begin{equation}
	    H^{\varepsilon}_{\min} (\rho_{Z_A E}^{\otimes n}| \rho_E^{\otimes n}) \xrightarrow{\varepsilon \rightarrow 0} H_{\min}(\rho_{Z_A E}^{\otimes n}| \rho_E^{\otimes n}).
	\end{equation}
 Then, using the additivity of min-entropy (derived in Lemma 3.1.6 of Ref.~\cite{Renner2005a}) we have that:
	\begin{equation}
	\begin{aligned}
		H^{\varepsilon}_{\min} (\rho_{Z_A E}^{\otimes n}| \rho_E^{\otimes n}) &\xrightarrow{\varepsilon \rightarrow 0} H_{\min}(\rho_{Z_A E}^{\otimes n}| \rho_E^{\otimes n}) \\&= n H_{\min}(\rho_{Z_A E}|\rho_E) \\&= n H_{\min}(Z_A|E),
	\end{aligned}
	\end{equation}
which gives a lower bound on the smooth min-entropy in terms of the single-transmission min-entropy. (The same result using different bounds of the smooth min-entropy is found in Ref.~\cite{bratzikMinentropyQuantumKey2011}.)

This approach guarantees an $\epsqkd$-secure QKD protocol that is $\epscor$-correct and $\epssec$-secret at a secret key rate per transmission of:
	\begin{equation}
	\label{eq:keyrate_eq_min_entr}
	\begin{aligned}
		r_2 = \frac{\ell}{N} &= \frac{n}{N} \left\lfloor H_{\min}(Z_A|E) - \frac{1}{n} \leakEC \right. \\ & \left.- \frac{2}{n} \log_2 \frac{1}{2 \epspa} - \frac{1}{n} \log_2 \frac{2}{\epscor}\right\rfloor.
	\end{aligned}
	\end{equation}
The secrecy of the protocol is found by composing the error terms $\epssec \geq \epspa + \epspe + \epsec$ (since $\varepsilon = 0$).

The optimization problem to be solved in this formulation is therefore $\min_{\rho_{AB} \in \mathcal{C}_{\epspe}} [H_{\min}(Z_A|E)]$. To solve this problem, we must show how the objective function $H_{\min}(Z_A|E)$ can be expressed in terms of an optimization problem that does not include Eve's state. We obtain the following relation by following a similar approach to Ref.~\cite{Coles2012} (further detailed in Methods):
\begin{equation}
    H_{\min}(Z_A|E) = - \log_2 \max_{\sigma_{AB}} F(\rho_{A B}, \sum_j Z_A^j \sigma_{AB} Z_A^j),
\end{equation}
where 
\begin{equation}
\label{eq:fidelity}
	F(\rho, \sigma) = \left( \Tr \sqrt{\sqrt{\rho} \sigma \sqrt{\rho}} \right)^2
\end{equation}
is the fidelity function and $\sigma_{AB}$ is a valid density matrix. Finally, using the linear SDP developed in Ref.~\cite{Watrous2009} for the fidelity, we obtain a SDP that can be solved for a secure lower bound to the min-entropy, and therefore for the finite key rate (see Methods for further details).

The lower bound to the key rate using the single-transmission min-entropy is typically not as tight as that using the von Neumann entropy. However, this formulation is computationally less expensive than the von Neumann approach, and is therefore useful for protocols with signal states that have a large Hilbert space.

\subsection*{Examples}
\label{sec:examples}

We now illustrate our numerical approach for obtaining reliable lower bounds on the QKD secret key rate by applying it to some well-known protocols. We consider the BB84 protocol~\cite{BB84}, the B92 protocol~\cite{Bennett1992}, and the novel Twin-Field QKD protocol~\cite{Lucamarini2018} (that is able to beat the fundamental capacity for direct quantum communication without any repeater~\citep{Pirandola2017}). We use BB84 as a benchmark, showing that our numerics exactly reproduce the known theoretical results based on analytical solutions to the key rate problem. For B92 and Twin-Field QKD, where analytical solutions are not known in general, we present novel results in the finite key regime using the approaches derived in the previous section. 

We supplement the finite key results for B92 and Twin-Field QKD with asymptotic results using the numerical approach of Ref.~\citep{Winick2017}, and find improved secret key rate lower bounds over those previously known. In the Supplementary Note~\ref{sec:other_examples}, we study protocols that lack symmetry, which have previously been analyzed numerically in the asymptotic regime~\citep{Winick2017}. In particular, we look at variations of BB84: one with detector-efficiency mismatch, and one with Trojan-horse attacks. For all results presented we used the \texttt{Mosek}~\citep{mosek} SDP solver, with the SDPs programmed within a disciplined convex programming framework: \texttt{cvxpy}~\citep{cvxpy,cvxpy_rewriting} in \texttt{Python} or \texttt{CVX}~\citep{cvx,cvx-2} in \texttt{MATLAB}.

\subsubsection*{BB84}
\label{sub:bb84}

Let us start by considering the idealized entanglement-based version of the BB84 protocol~\citep{BB84}. Alice and Bob each receive a qubit from a maximally entangled state $\ket{\Phi^+} = (\ket{00} + \ket{11})/\sqrt{2}$ and measure their qubit with probability $p_Z$ in the $Z$-basis $ = \{\ket{0}, \ket{1}\}$ or with probability $p_X = 1-p_Z$ in the $X$-basis $= \{\ket{+}, \ket{-}\}$. The $X$-basis states are defined as $\ket{\pm} = (\ket{0} \pm \ket{1})/\sqrt{2}$. 

Alice and Bob postselect for the cases when they both measure their qubits in the same basis, discarding the outcomes when they measure in different bases (see Supplementary Note~\ref{app:general_framework_for_postselection} for the postselection framework). They then generate a secret key using the results when they both measured in either the $Z$-basis or $X$-basis.

The maximally entangled state $\ket{\Phi^{+}}$ is generated in Alice's laboratory so that only one part of the state is transmitted through the channel to Bob. To model this transmission through the quantum channel, we consider the depolarizing channel with a depolarizing probability $p$ on Bob's qubit~\citep{nielsen2000quantum}:
\begin{equation}
	\mathcal{E}^{\text{dep}}(\rho) = (1-p) \rho + p \frac{\id}{2}. 
\end{equation}
Therefore, for this protocol, we consider the statistics given by the state:
\begin{equation}
\label{eq:depolarized_rho_AB_bb84}
	\rho'_{AB} = (\id_A \otimes \mathcal{E}^{\text{dep}}_B(p)) \left(\ket{\Phi^+}\bra{\Phi^+}_{AB} \right).
\end{equation}

Typically in QKD experiments the key rates are determined by the quantum bit error rates (QBERs). Therefore, we use these error rates to define coarse-grained constraints for the key rate SDPs. The error operators corresponding to the QBERs in the $Z$ and $X$ bases are:
\begin{equation}
\begin{aligned}
E_Z &= \ket{0}\bra{0}_A \otimes \ket{1}\bra{1}_B + \ket{1}\bra{1}_A \otimes \ket{0}\bra{0}_B,  \\
E_X &= \ket{+}\bra{+}_A \otimes \ket{-}\bra{-}_B + \ket{-}\bra{-}_A \otimes \ket{+}\bra{+}_B,
\end{aligned}
\end{equation}
whose expectation values (i.e.~the QBERs) are $Q_Z = \braket{E_Z} = \Tr(\rho'_{AB} E_Z)$ and $Q_X = \braket{E_X} = \Tr(\rho'_{AB} E_X)$. For the state defined in~\eqref{eq:depolarized_rho_AB_bb84}, one can show analytically that $Q_Z = Q_X = Q = 2p$.

The analytical solution for the von Neumann entropy optimization is
\begin{equation}
\label{eq:vnm_bb84}
	\min_{\rho_{AB} \in C_{\epspe}} \ppass H(Z_A|E) = (p_Z^2 + p_X^2) \left[1- h_2(Q^{\UB}) \right],
\end{equation}
where $h_2(x) \equiv -x \log_2(x) - (1-x) \log_2(1-x)$ is the binary entropy function. The factor $(p_Z^2 + p_X^2)$ is the probability Alice and Bob postselect for the same basis. Similarly, the solution for the min-entropy optimization is~\cite{bratzikMinentropyQuantumKey2011}:
\begin{equation}
\label{eq:min_bb84}
\begin{aligned}
	& \min_{\rho_{AB} \in C_{\epspe}} \ppass H_{\min}(Z_A|E) \\
	& = -(p_Z^2 + p_X^2) \left[ 1-\log_2\left(1+2\sqrt{Q^{\UB}(1-Q^{\UB})}\right) \right],
\end{aligned}
\end{equation}
where $Q^{\UB} = Q + \Delta(\epspe/2, m_i)$ with $\Delta > 0$ being the deviation that can be quantified using the law of large numbers Eq.~\eqref{eq:large_numbers}. We can compare the key rate predicted by the SDP with these analytical formulas (in the asymptotic regime) and apply the same estimation methods to both the numerical and the analytical key rates to compare them in the non-asymptotic regime.

To simulate a realistic QKD system, we assume that Alice and Bob uses $\ratioPE = 10\%$ of the signals (after postselection) for parameter estimation. We take the protocol to be correct up to $\epscor = 10^{-15}$ and to be secret up to $\epssec = 10^{-10}$. For simplicity, we assume equal security parameters of $\varepsilon'$ for $\epspa$, $\epsec$, and $\varepsilon$. For parameter estimation, we assume each constraint is estimated with a failure probability up to $\epspe^i = 2 \varepsilon'$. The values of the security parameters are tabulated in Tab.~\ref{tab:security_params}. Therefore, we have $\varepsilon' = \epssec/7$ for the calculation with von Neumann entropy (Eq.~\eqref{eq:keyrate_eq_vnm} with SDPs~\eqref{eq:key_rate_nonasymptotic_appx_primal} and~\eqref{eq:lower_bound_nonasymptotic}) and $\varepsilon' = \epssec/6$ for the calculation with min-entropy (Eq.~\eqref{eq:keyrate_eq_min_entr} with SDP~\eqref{eq:keyrate_Hmin_dual_problem}).

\begin{figure}[H]
	\centering
	\begin{subfigure}[b]{\columnwidth}
	    \includegraphics[width=\columnwidth]{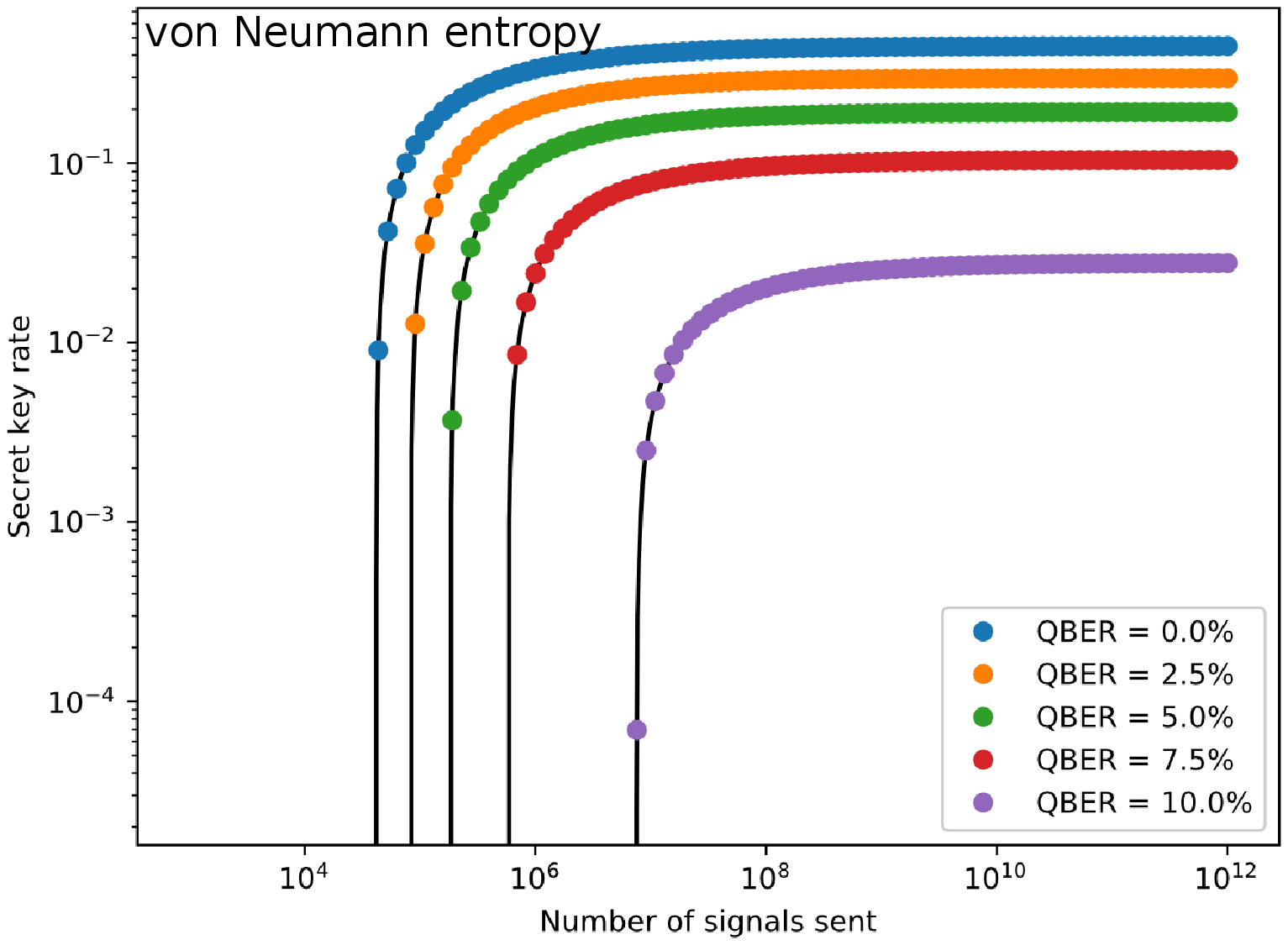}
	\end{subfigure}
	\begin{subfigure}[b]{\columnwidth}
	    \includegraphics[width=\columnwidth]{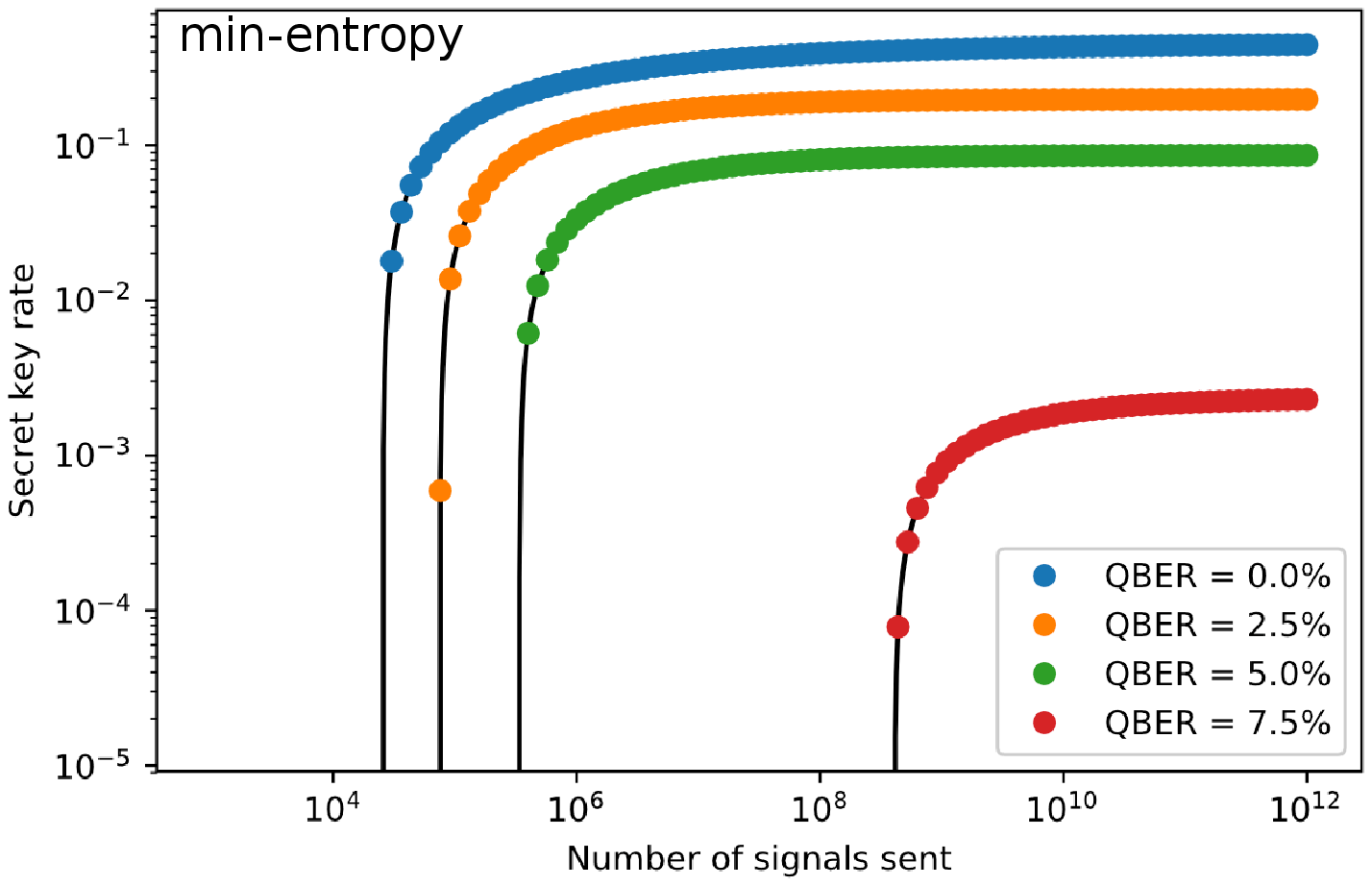}
	\end{subfigure}
	\caption{Nonasymptotic secret key rate per pulse for the BB84 protocol calculated using the von Neumann entropy (top) and the min-entropy (bottom) for different values of QBER. The line is a known theoretical curve calculated from Eqs.~\eqref{eq:vnm_bb84} and~\eqref{eq:min_bb84}, and the dots are reliable lower bounds on the key rate calculated numerically.}
    \label{fig:bb84_finite}
\end{figure}

Fig.~\ref{fig:bb84_finite} show the secret key rate per transmitted pulse as a function of the number of transmissions, $N$, for both numerical methods. We observe that for a QBER of 0\%, the bound from solving the min-entropy SDP is tighter and rises to a significant value at a lower number of transmitted signals. At higher error rates, solving the von Neumann entropy SDP provides a better bound. In fact, for a QBER $\gtrsim 7.58\%$, the bound from min-entropy predicts a zero key rate at any number of transmissions.

\begin{table}[H]
\centering
\begin{tabular}{c | c c}
	\hline
	 Parameter & von Neumann entropy & min-entropy \\
	  \hline
	  $\epsec$ & $\varepsilon'$ & $\varepsilon'$ \\
	  $\epspa$ & $\varepsilon'$ & $\varepsilon'$ \\
	  $\varepsilon$ & $\varepsilon'$ & $0$ \\
	  \hline
	  $\ratioPE$ & 10\% & 10\% \\
	  $\epspe^i$ & $2 \varepsilon'$ & $2 \varepsilon'$ \\
	  $\epspe$ & $2 \nPE \varepsilon'$ & $2 \nPE \varepsilon'$ \\
	  \hline
	  $\epssec$ & $10^{-10}$ & $10^{-10}$ \\
	  $\epscor$ & $10^{-15}$ & $10^{-15}$ \\
   	\hline
\end{tabular}
\caption{Values of security parameters and other relevant quantities for parameter estimation, assuming equal security parameters of $\varepsilon'$ for the two numerical bounds: one bound is calculated using von Neumann entropy and another bound using min-entropy. The parameters listed here are: $\epsec$: error-correction failure probability; $\epspa$: privacy-amplification failure probability; $\varepsilon$: smoothing parameter for smooth min-entropy; $\ratioPE$: fraction of signals used for parameter estimation; $\epspe^i$: failure probability of estimating parameter described by constraint $\Gamma_i$; $\epspe$: parameter-estimation total failure probability; $\nPE$: number of constraints to be quantified from the parameter estimation step; $\epssec$: secrecy failure probability; $\epscor$: probability that Alice and Bob's secret keys are not identical.}
\label{tab:security_params}
\end{table}

\subsubsection*{B92}
\label{sub:B92}

The B92 protocol~\citep{Bennett1992} is a simple QKD protocol that is highly asymmetric. In this protocol, Alice prepares one of two non-orthogonal states $\{\ket{\phi_0}, \ket{\phi_1}\}$, where $\braket{\phi_0 | \phi_1} = \cos(\theta/2)$---each with a probability $1/2$, which she will send to Bob. Alice therefore prepares the state:
\begin{equation}
	\ket{\psi}_{AB} = \sqrt{\frac{1}{2}} \ket{0}_A \ket{\phi_0}_B + \sqrt{\frac{1}{2}} \ket{1}_A \ket{\phi_1}_B.
\end{equation}

Upon receiving the signals from Alice, Bob randomly measures either in the $B_0 = \{\ket{\phi_0}, \ket{\bar{\phi_0}} \}$ or in the $B_1 = \{\ket{\phi_1}, \ket{\bar{\phi_1}}\}$ basis, where $\braket{\phi_0 | \bar{\phi_0}} = \braket{\phi_1 | \bar{\phi_1}} = 0$. Bob postselects for those instances where he measures $\ket{\bar{\phi_0}}$ or $\ket{\bar{\phi_1}}$, and publicly announces `pass'. He then assigns a bit value 1 or 0, respectively, to his key. If he measures $\ket{\phi_0}$ or $\ket{\phi_1}$, he will announce `failure', and both parties discard these signals. 

As constraints for the problem we use operators that describe (on the postselected state) both a successful outcome, where Bob's measurement bit makes Alice's prepared bit, and an unsuccessful one, where they do not agree:
\begin{equation}
\begin{aligned}
    \Gamma^{(=)} &= \kb{0}_A \otimes \kb{\bar{\phi_1}}_B + \kb{1}_A \otimes \kb{\bar{\phi_0}}_B, \\
	\Gamma^{(\neq)} &= \kb{0}_A \otimes \kb{\bar{\phi_0}}_B + \kb{1}_A \otimes \kb{\bar{\phi_1}}_B.
\end{aligned}
\end{equation}
In addition, since this is a prepare-and-measure protocol, we must add the constraints related to Alice's knowledge of $\rho_A$, i.e. we add the constraints $\Omega_A^j \otimes \id_B = \kb{\psi_j}_A \otimes \id_B$ obtained from the spectral decomposition of $\rho_A = \sum_j p_j \kb{\psi_j}_A$. 

To simulate the channel, we again consider that the signal undergoes a depolarizing channel with probability $p$ as it travels from Alice to Bob.
% \begin{equation}
% 	\rho'_{AB} = (\id_A \otimes \mathcal{E}^{\text{dep}}_B (p))(\kb{\psi}_{AB}).
% \end{equation}
Figure~\ref{fig:b92_asymptotic} shows the results of our numerical method in the asymptotic limit. We plot the secret key rate per pulse against the angle $\theta$ and against the depolarizing probability $p$ after optimizing for the parameter $\theta$. The asymptotic formula is obtained by taking the von Neumann formulation with $N\rightarrow\infty$ (and $n\rightarrow\infty$), and replacing the constraints $\gamma_k^{\LB} \leq \Tr(\rho_{AB} \Gamma_k) \leq \gamma_k^{\LB}$ with tight constraints $\gamma_k = \Tr(\rho_{AB} \Gamma_k)$. This is a direct application of the formalism developed in Ref.~\cite{Winick2017}, and is not a development of this manuscript, though the results we show have not been reported elsewhere, and are a useful demonstration of the power of numerical QKD calculations.

Our results guarantee a non-zero key rate even up to $p = 0.15$ (with $r_1 = 0.00574$ at $\theta = {64.8}^\circ$), while previous analytical results predict a non-zero key rate only for %$p \leq 0.034$~\citep{Tamaki2004}, $p \leq 0.048$~\citep{Renner2005a}, and 
$p \leq 0.065$~\citep{Tamaki2004,Renner2005a,Sasaki2015}. Furthermore, this numerical approach guarantees a higher-secret key rate when compared to a previous numerical QKD approach described in Ref.~\citep{Coles2016}, which predicts a non-zero key rate for $p \leq 0.053$. For noise levels where all methods guarantee finite key rates, our results show tighter secure lower bounds than previous approaches. For example, for a depolarizing noise of $p= 0.01$, the method of Ref.~\cite{Winick2017} predicts $r_1 = 0.248$, while the previous method of~\citep{Coles2016} only obtains $r_1 \approx 0.21$ per pulse. %This shows that the new numerical approach is tighter than that of~\citep{Coles2016}, which was loosened due to the use of Golden-Thompson inequality to obtain an approachable dual problem.

\begin{figure}[H]
    % \centering
    \includegraphics[width=\columnwidth]{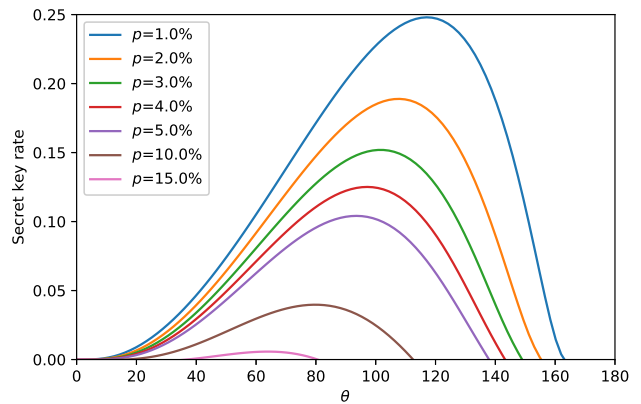}
    \caption{Secret key rate per pulse for the B92 protocol for different depolarizing probability $p$ in the asymptotic regime. The rate is plotted against the Bloch-sphere angle between the two signal states $\ket{\phi_0}$ and $\ket{\phi_1}$.}
    \label{fig:b92_asymptotic}
\end{figure}

Now, we consider the security of the B92 protocol in the nonasymptotic regime, using the finite key SDPs developed in the previous section. Fig.~\ref{fig:b92_finite} shows the secret key generation rate per pulse in terms of the number of signals that Alice has sent, for different values of $p$. For each curve, we choose the value of $\theta$ that maximizes the secret key rate. We consider the security parameters tabulated in Tab.~\ref{tab:security_params} and assume that the protocol is $\epssec = 10^{-10}$-secret and $\epscor = 10^{-15}$-correct. Our analysis shows that the B92 protocol is a simple way of exchanging random secret keys with composable security.

\begin{figure}[H]
    \centering
    \begin{subfigure}[b]{\columnwidth}
        \includegraphics[width=\columnwidth]{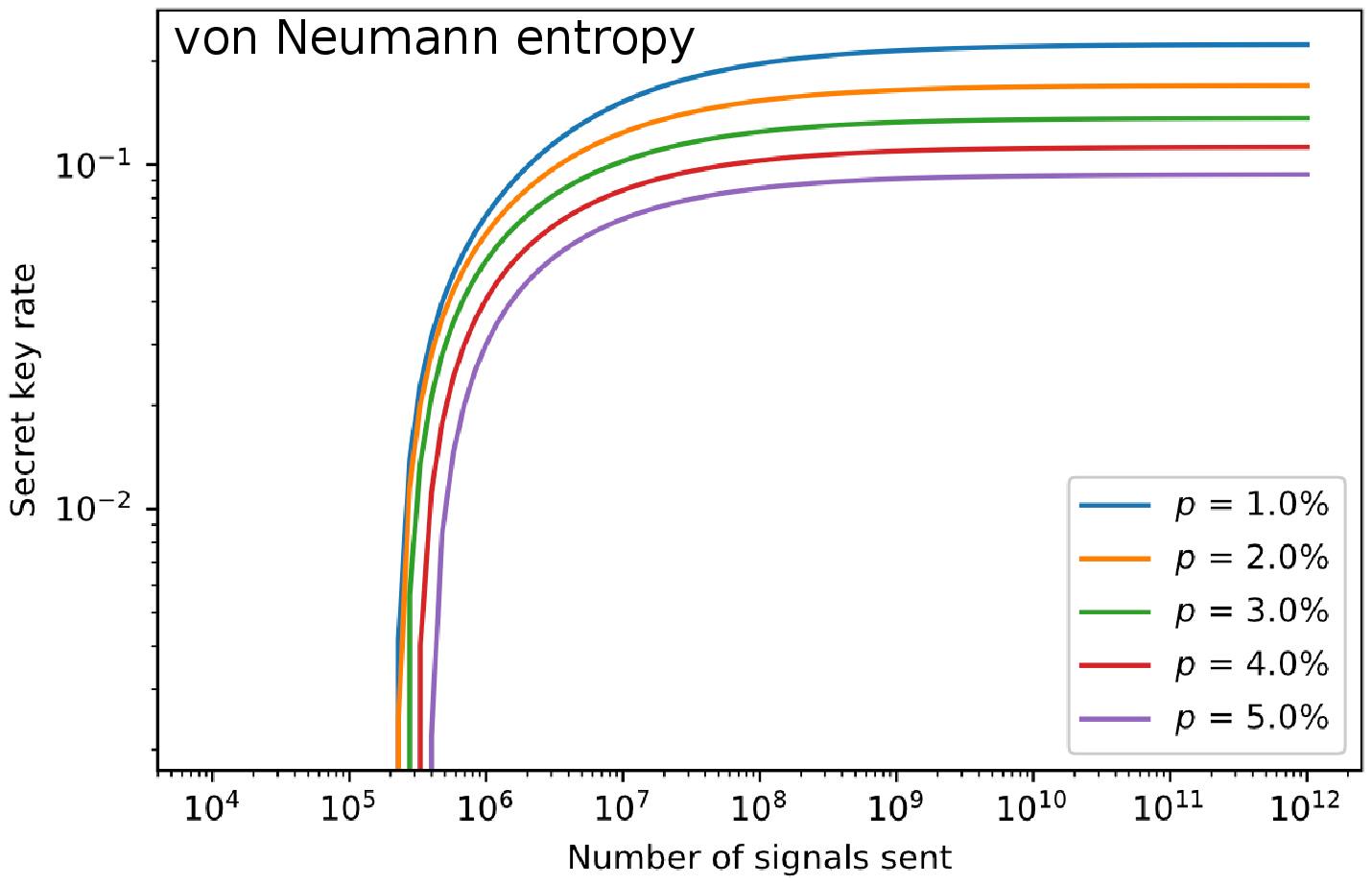}
    \end{subfigure}
    \begin{subfigure}[b]{\columnwidth}
        \includegraphics[width=\columnwidth]{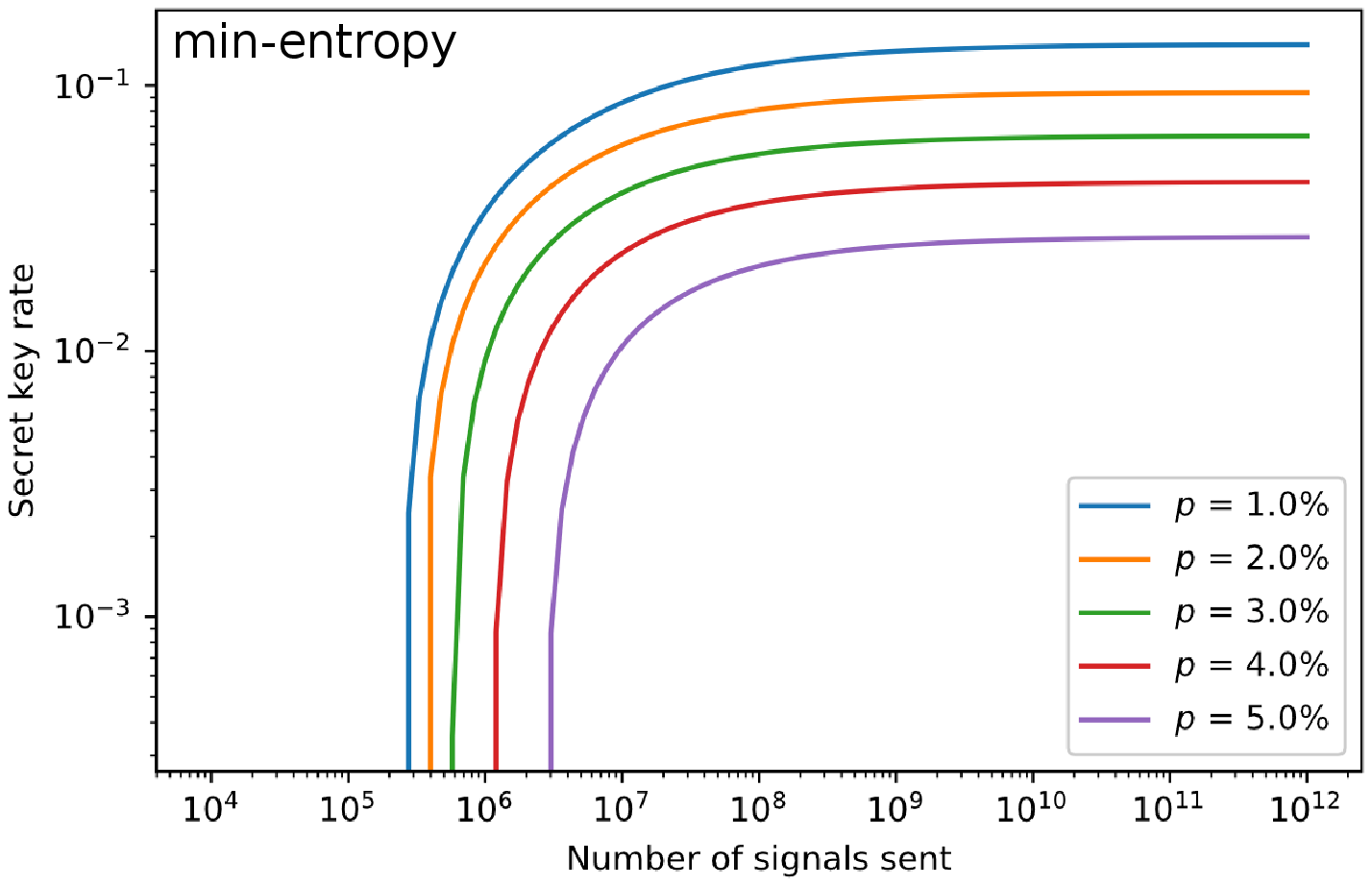}
    \end{subfigure}
    \caption{Secret key rate per pulse for the B92 protocol calculated using the von Neumann entropy (top) and the min-entropy (bottom) for different values of depolarizing probability $p$.}
    \label{fig:b92_finite}
\end{figure}

\subsubsection*{Twin-Field QKD}
\label{sub:twin_field}

Twin-Field QKD is a novel variation of the measurement-device-independent (MDI)-QKD protocol that enables two parties to communicate through an intermediate untrusted node~\cite{Lo2012}. The main difference is that TF-QKD uses single-photon interference (instead of two-photon interference in MDI-QKD), and achieves a key rate that is expected to beat the fundamental information capacity for the repeaterless quantum communication rate, typically at long distances~\citep{Lucamarini2018}. Because the protocol promises higher secret key rates than previous QKD protocols at long distances, it is a subject of extensive research both theoretically and experimentally. Theoretically, several security proofs of the protocol (and its variations) with asymptotic~\cite{tamakiInformationTheoreticSecurity2018,curtySimpleSecurityProof2019,maPhaseMatchingQuantumKey2018} and non-asymptotic~\cite{jiangUnconditionalSecuritySending2019a,he2019finitekey,lorenzo2019tightfinitekey} key rates have been proposed. Experimentally, the protocol has been demonstrated within the laboratory setting~\cite{minderExperimentalQuantumKey2019,wangBeatingFundamentalRateDistance2019,zhongProofofPrincipleExperimentalDemonstration2019,liuExperimentalTwinFieldQuantum2019}, and future field tests of the protocol are to be expected.

In the entanglement-based description of Twin-Field QKD presented in Ref.~\cite{curtySimpleSecurityProof2019}, Alice and Bob each prepare the entangled state $\ket{\Phi_q} = \sqrt{q} \ket{00} + \sqrt{1-q} \ket{11}$ for $0 \leq q \leq 1$, where $\ket{0}$ is the vacuum state (as opposed to the logical state 0) and $\ket{1}$ is the single photon state~\footnote{In Twin-Field QKD, the logical basis and the photon-number or Fock basis coincide.}. They then randomly choose to measure their qubits in the standard $Z = \{\ket{0}, \ket{1} \}$ basis with probability $p_Z$ or in the $X = \{\ket{+}, \ket{-}\}$ basis with probability $p_X = 1 - p_Z$. 

Alice and Bob send one part of their quantum signals ($A'$ for Alice and $B'$ for Bob) through optical channels with transmittance $\sqrt{\eta}$ to Charlie. The total optical transmittance from Alice to Bob is $\eta$. Charlie then performs a Bell-state measurement on the combined state $A'B'$ that he has received. One way to do so is to use a 50:50 beamsplitter to mix Alice and Bob's signals, and then route the outputs of this beamsplitter to two single-photon detectors. Charlie announces which of his two detectors fires, and Alice and Bob postselect for those events where one (and only one) detector fires. This is equivalent to postselecting for Charlie's state being one of the Bell states $\ket{\Psi^{\pm}}$. We can postselect to $\ket{\Psi^+}$ and $\ket{\Psi^-}$ independently, and we focus our discussion henceforth on postselection to the single state $\ket{\Psi^-}$ and the same arguments apply for postselection to the state $\ket{\Psi^+}$.

The measurement POVMs for analyzing the Twin-Field QKD protocol are
\begin{center}
\begin{tabular}{c c c}
	\hline
	 Alice's POVM & Basis & Bit-value \\
	 \hline	
	 $p_Z \ket{0}\bra{0}_A$ & $Z$ & 0 \\
	 $p_Z \ket{1}\bra{1}_A$ & $Z$ & 1 \\
	 $p_X \ket{+}\bra{+}_A$ & $X$ & 0 \\
	 $p_X \ket{-}\bra{-}_A$ & $X$ & 1 \\
   	\hline
\end{tabular}
\end{center}
for Alice, and
\begin{center}
\begin{tabular}{c c c}
	\hline
	 Bob's POVM & Basis & Bit-value \\
	 \hline	
	 $p_Z \ket{0}\bra{0}_B$ & $Z$ & 0 \\
	 $p_Z \ket{1}\bra{1}_B$ & $Z$ & 1 \\
	 $p_X \ket{+}\bra{+}_B$ & $X$ & 0 \\
	 $p_X \ket{-}\bra{-}_B$ & $X$ & 1 \\
   	\hline
\end{tabular}
\end{center}
for Bob. In the simulation, the channel that the transmitted signal goes through is a pure-loss channel (or an amplitude damping channel for the single photon case~\citep{nielsen2000quantum}). We can describe the pure-loss channel $\mathcal{E}^{\text{loss}}(\sqrt{\eta})$ using a beam-splitter transformation with the help of an additional Hilbert space $A_0$ starting in the vacuum state. For example, the photon creation operator for Alice's transmitted signal $\hat{a}_{A'}^{\dag}$ undergoes the following transformation for a channel with transmittance $\sqrt{\eta}$:
\begin{equation}
	\begin{pmatrix}
		\hat{a}_{A'}^{\dag} \\
		\hat{a}_{A_0}^{\dag}
	\end{pmatrix}
	\rightarrow
	\begin{pmatrix}
		\sqrt{\eta^{1/2}} & \sqrt{1-\eta^{1/2}} \\
		-\sqrt{1-\eta^{1/2}} & \sqrt{\eta^{1/2}}
	\end{pmatrix}
	\begin{pmatrix}
		\hat{a}_{A'}^{\dag} \\
		\hat{a}_{A_0}^{\dag}
	\end{pmatrix}.
\end{equation}
Here, $\hat{a}_{A_0}^{\dag}$ is the creation operator for the additional Hilbert space. In summary, to describe the transmitted state, Alice generates the entangled state:
\begin{equation}
	\ket{\psi}_{AA'A_0} = \ket{\Phi_q}_{AA'} \otimes \ket{0}_{A_0},
\end{equation}
which then undergoes a pure loss channel before being measured by Charlie. The state after the transmission is
\begin{equation}
	\rho'_{AA'} = \Tr_{A_0} \left[ (\id_A \otimes \mathcal{E}^{\text{loss}}_{A'A_0}(\sqrt{\eta})) \left(\kb{\psi}_{AA'A_0} \right) \right].
\end{equation}
We can also define a similar state for Bob:
\begin{equation}
	\rho'_{BB'} = \Tr_{B_0} \left[ (\id_B \otimes \mathcal{E}^{\text{loss}}_{B'B_0}(\sqrt{\eta})) \left(\kb{\psi}_{BB'B_0} \right) \right],
\end{equation}
and the overall state after both Alice and Bob's transmissions is: $\rho'_{AA'BB'} = \rho'_{AA'} \otimes \rho'_{BB'}$.

Charlie, equipped with only threshold detectors, cannot distinguish between the click due to a only single photon arriving at his first detector and the click due to two photons arriving. Therefore, he is projecting the signals he receives to
\begin{equation}
	\kb{\widetilde{\Psi}^-}_{A'B'} = \frac{1}{2} \left( \kb{\Psi^-}_{A'B'} + \kb{11}_{A'B'} \right).
\end{equation}
We further assume that Charlie's detectors have small, but non-negligible dark counts. Let $p_d$ be the dark count probability for each clock cycle. We can modify Charlie's projection operator above into the following POVM:
\begin{equation}
\begin{aligned}
    \kb{\widetilde{\Psi}^-_{\text{dark}}}_{A'B'} &= (1-p_d)^2 \kb{\widetilde{\Psi}^-}_{A'B'} \\
    & + p_d (1-p_d) \frac{\id_{A'B'}}{4}.
\end{aligned}
\end{equation}
Alice and Bob postselect those cases, where Alice and Bob measure in the same basis and Charlie successfully measures $\kb{\widetilde{\Psi}^-_{\text{dark}}}$.

The performance of the protocol in the asymptotic limit as a function of the overall loss between Alice and Bob is plotted in Fig.~\ref{fig:tf_asymptotic}. At each loss value, we optimize for the value of $q$ that gives the best key rate using the Brent's method~\citep{numerical_recipes,brent2013algorithms}. As shown in Fig.~\ref{fig:tf_q_asymptotic}, the value of $q$ increases monotonically to about $\sim 0.93$ at 40~dB loss, and saturates at this value for higher losses. The high value of $q$ suggests that a weakly pumped photon pair source, which uses spontaneous parametric down conversion or spontaneous four-wave mixing, would be ideal to generate the initial entangled states.

From Fig.~\ref{fig:tf_asymptotic}, it is clear that the Twin-Field QKD protocol---at sufficiently high losses (above $\sim40$~dB)---can perform better than the capacity of direct repeaterless quantum communications, which we dub as the PLOB bound after the original authors~\citep{Pirandola2017}. For a channel with a transmittance $\eta$, the bound which is an achievable rate is $-\log_2(1-\eta)$ and scales linearly as $\sim \eta$ at low transmittance. As the dark count rate increases, the region of losses at which the Twin-Field QKD protocol can beat the PLOB bound is reduced.

\begin{figure}[H]
    \centering
	\includegraphics[width=\columnwidth]{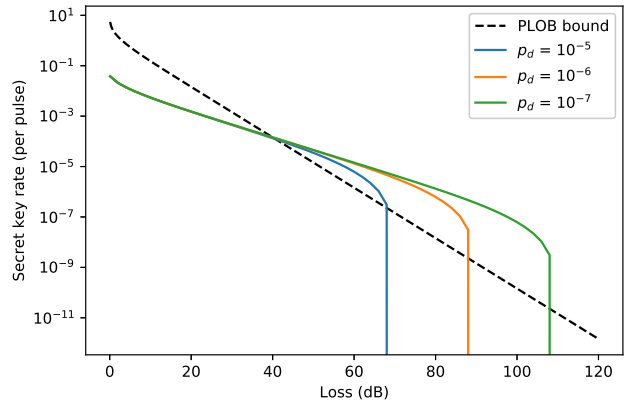}
	\caption{Secret key rate per pulse for the Twin-Field QKD protocol as a function of the overall loss between Alice and Bob. The different lines are for QKD operations with different dark count probability $p_d$. The black dashed line corresponds to PLOB bound: the fundamental bound for direct repeaterless communications, calculated with $\eta = 10^{-(\text{Loss in dB})/10}$.}
	\label{fig:tf_asymptotic}
\end{figure}

\begin{figure}[H]
	\includegraphics[width=\columnwidth]{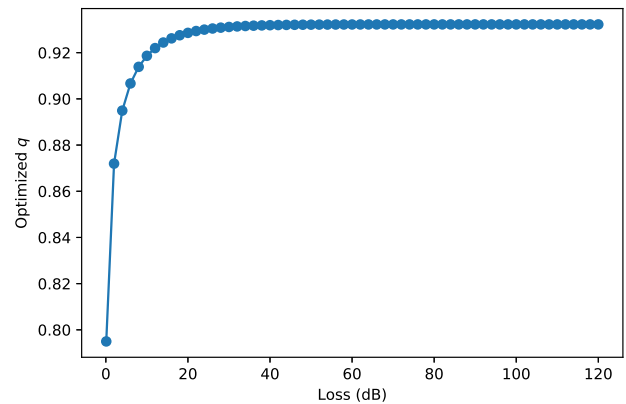}
	\caption{Optimized value of $q$ that gives the highest secret key rate at each loss value.}
    \label{fig:tf_q_asymptotic}
\end{figure}

For the nonasymptotic regime, we evaluate the security of a protocol that is $\epssec$-secret and $\epscor$-correct with $\epssec = 10^{-10}$ and $\epscor = 10^{-15}$. We consider the values for the security parameters as in Tab.~\ref{tab:security_params}. 
%These values give us $\epssec = 7 \varepsilon'$ for the solution of the key rate problem via von Neumann entropy (Eq.~\eqref{eq:keyrate_eq_vnm} with SDPs~\eqref{eq:key_rate_nonasymptotic_appx_primal} and~\eqref{eq:lower_bound_nonasymptotic}), and $\epssec = 6 \varepsilon'$ for the solution via min-entropy (Eq.~\eqref{eq:keyrate_eq_min_entr} with SDP~\eqref{eq:keyrate_Hmin_dual_problem}).

Fig.~\ref{fig:tf_finite} shows that the nonasymptotic bounds from von Neumann entropy can obtain better secret key rates than the PLOB bounds---even with the presence of dark counts. The plots also show that to faithfully demonstrate a better rate than the PLOB bound in a Twin-Field QKD experiment, both Alice and Bob must send a large number of transmissions to Charlie. For example, at 60~dB overall channel loss, they must send $N\sim 10^{10}$ transmissions which are $\sim 10^5$ higher than the number of transmissions needed to obtain a substantial secret key in a BB84 QKD protocol (see Fig.~\ref{fig:bb84_finite}). Interestingly, the bounds from min-entropy are unable to beat the PLOB bound, but they do guarantee a substantial secret key rate even at orders of magnitude fewer transmissions.

\begin{figure}[H]
    \begin{subfigure}[b]{\columnwidth}
	\includegraphics[width=\columnwidth]{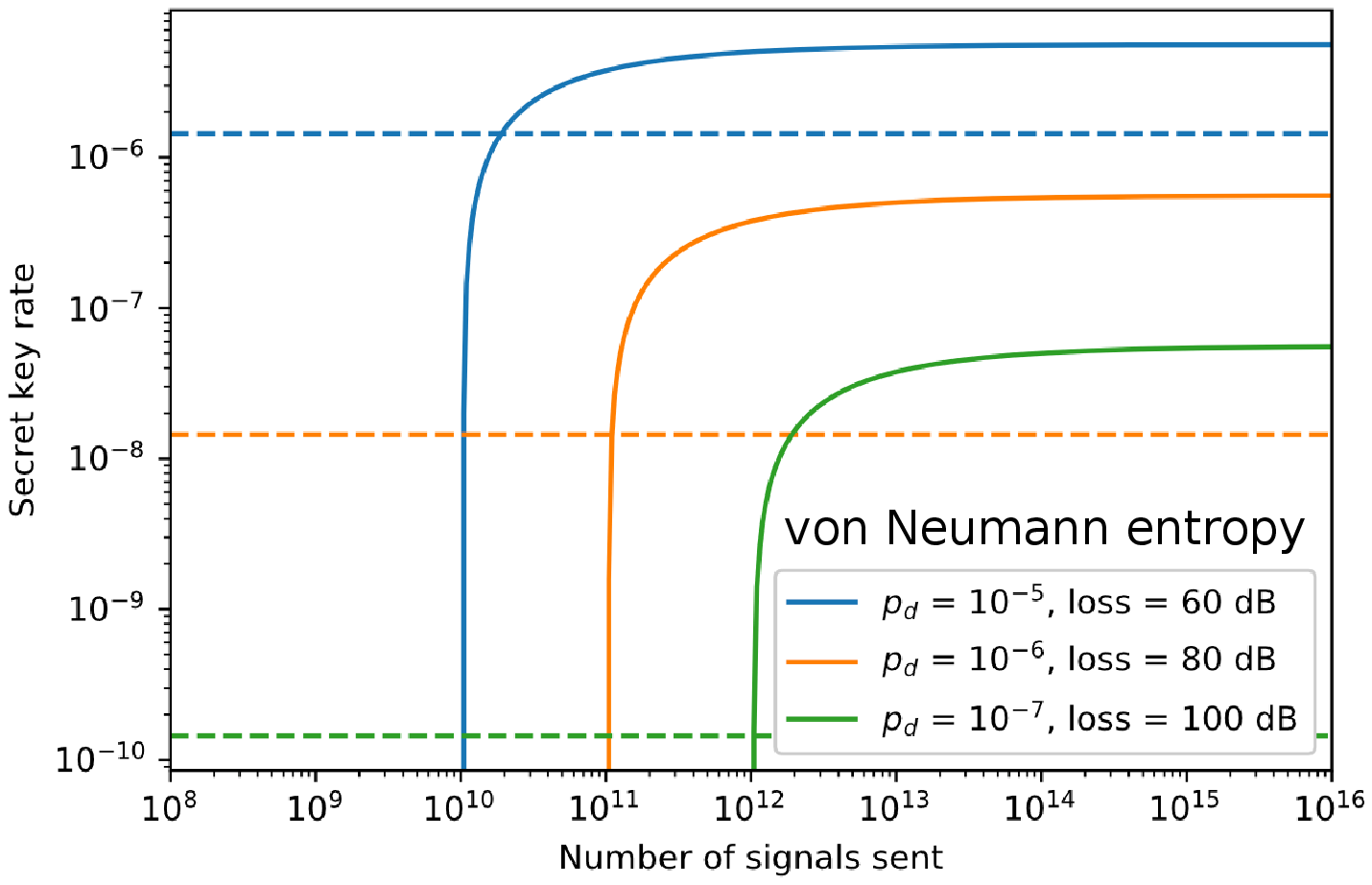}
	\end{subfigure}
	\begin{subfigure}[b]{\columnwidth}
	\includegraphics[width=\columnwidth]{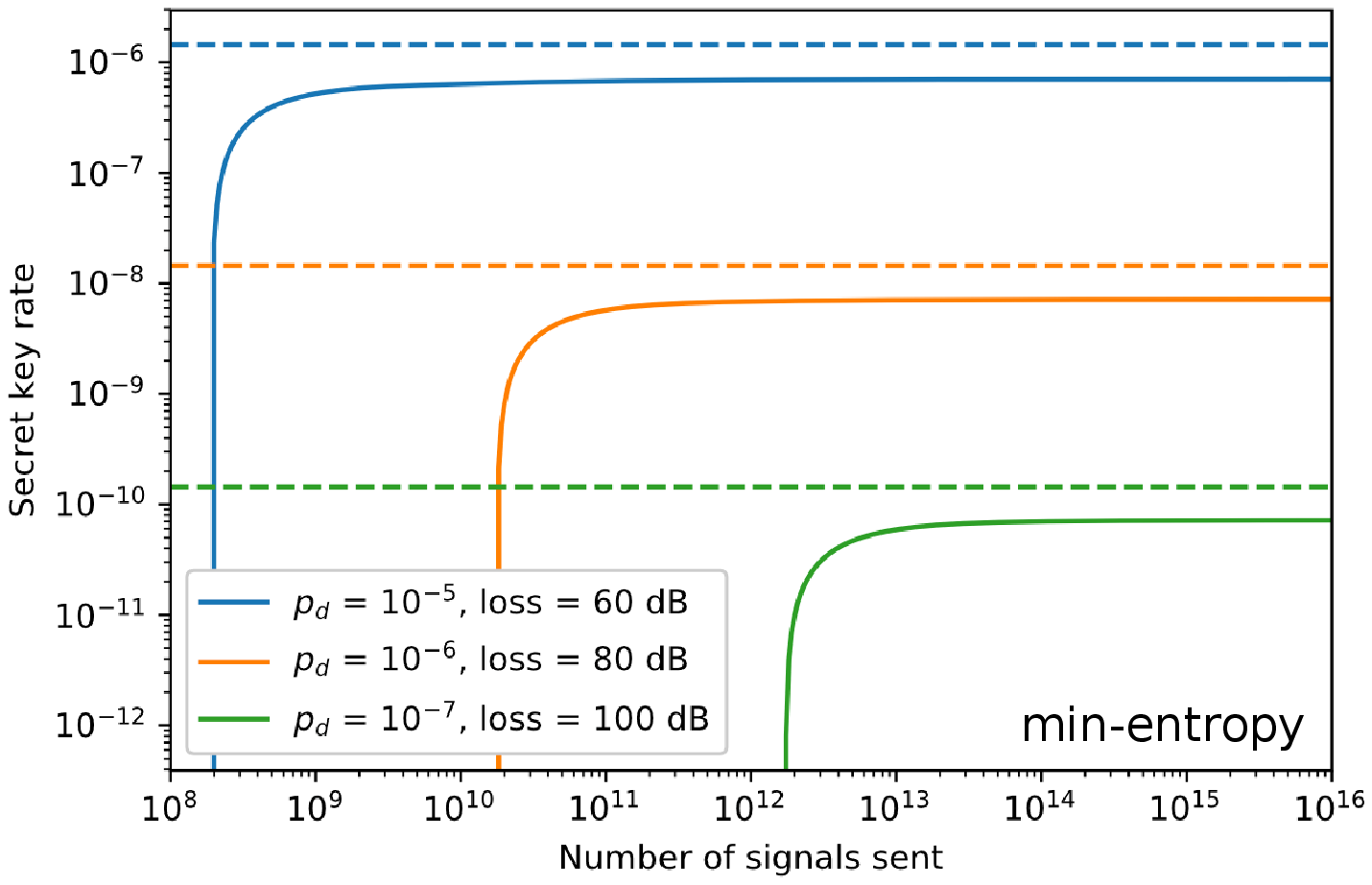}
	\end{subfigure}
    \caption{Secret key rate per pulse for the Twin-Field QKD protocol calculated at different values of dark count $p_d$ and overall channel loss. The key rates are obtained from von Neumann entropy (top) and from min-entropy (bottom). The dashed lines with the same color are the PLOB bound at the different loss values, i.e. blue: 60~dB loss, orange: 80~dB loss, green: 100~dB loss.}
    \label{fig:tf_finite}
\end{figure}

\section*{Discussion}
\label{sec:conclusions_and_future_work}

We have developed semidefinite programs for finding reliable lower bounds on the secret key rate of an arbitrary QKD protocol in the nonasymptotic regime. We presented two methods of calculating such bounds, one via an SDP for von Neumann entropy and one via an SDP for min-entropy. For some of the protocols we have considered, the bound from min-entropy provides a better key rate than the bound for von Neumann entropy at lower error rates and at lower numbers of transmissions. The computational advantage for solving the SDP for min-entropy is also clear since the problem is more tractable than that for von Neumann entropy. For a problem involving a density matrix between Alice and Bob of size $n \times n$, the SDP~\eqref{eq:keyrate_Hmin_dual_problem} for min-entropy only requires us to solve for $\mathcal{O}(n^2)$ parameters while the SDP~\eqref{eq:key_rate_nonasymptotic_appx_primal} for von Neumann entropy requires us to solve for $\mathcal{O}(n^4)$ parameters. Nevertheless, the nonasymptotic bound from von Neumann entropy guarantees a better secret key rate at higher numbers of transmission and, unlike the bound from min-entropy, can approach the asymptotic key rate. The supremum between these two methods should be considered as the tightest lower bound that our numerical approach offers.

So far, we have only considered security against collective attacks. Some protocols with high-symmetry have been found to have the same secret key rates under collective attacks and under the more general coherent attacks. Examples of these protocols include popular protocols such as BB84.

General methods for bounding the possible information advantage of coherent attacks over collective attacks has been outlined in multiple approaches. The first such approach uses the \emph{exponential de Finetti theorem}~\citep{Renner2007}, but the overhead obtained by this theorem turns out to be heavy making the finite-key bounds unrealistically pessimistic. The de Finetti theorem is tight if one compares the attacks signal-by-signal. Ref.~\cite{Christandl2009} found that it suffices to only consider the entire collection of states. This method, known as the \emph{postselection technique}, compares the distance between two maps: the map between the ideal protocol and an actual protocol under collective attacks \emph{and} the map between the ideal protocol and an actual protocol under coherent attacks.

Using the postselection technique, we can define a new secrecy parameter under a coherent attack $\epscoh$, which quantifies the probability the QKD protocol passes but is not secret to an eavesdropper with coherent-attack capabilities. $\epscoh$ is related to the secrecy parameter under collective attack $\epssec$ in the following manner:
\begin{equation}
	\epssec = \epscoh (N+1)^{-(d^4-1)}.
\end{equation}
For the above value of secrecy, the key rate under coherent attack $r^{\text{coh}}$ is related to the key rate under collective attack $r$ by the following relation:
\begin{equation}
	r^{\text{coh}} = r - 2(d^4-1)\frac{\log_2(N+1)}{N}.
\end{equation}

Our numerical method is reliable and robust for calculating key rates involving single photon transmissions. Most practical implementations of QKD however have relied on the use of weak coherent states made by highly attenuated laser pulses. We hope to eventually evaluate such protocols numerically in the future. However, two main issues must be addressed when doing so. First, the probability of multiphoton emissions from a highly attenuated coherent light source, although small, is not negligible. Multiphoton signals are inherently insecure due to a class of attacks called the photon number splitting attack. One solution to combat the photon number splitting attack is to implement the decoy state protocol. In the decoy state protocol, Alice prepares an additional set of states---the decoy states---that are used to detect the presence of eavesdropping~\citep{Wang2005,Lo2005}. Therefore, we plan to incorporate decoy state analysis to the numerical method in the future.

Second, Alice's coherent state transmission uses an infinite-dimensional Hilbert space. The calculation on this infinite-dimensional space is extremely challenging. For simple QKD protocols, there exist squashing maps that provide direct correlations between measurements in the infinite-dimensional optical implementation and measurements in the abstract low-dimensional protocol~\citep{Beaudry2009,Tsurumaru2008,Gittsovich2014}. Therefore, the numerical method (or the user) must also be able to determine the appropriate squashing map to reduce the size of the problem.

To conclude, our results extend the earlier numerical QKD approaches by presenting a general robust framework for calculating QKD key rates in nonasymptotic regimes. The numerical methods presented here will be useful for democratizing the QKD security proofs that are needed to estimate the amount of secret key generated in any QKD operation.

\emph{Note added}---We recently became aware of another proposal for calculating the finite-key rate of a QKD protocol numerically~\cite{george2019qcrypt}.

\section*{Methods}
\label{sec:methods}

\subsection*{Proof of the relationship between min-entropy and the fidelity function}
\label{app:min_entropy_with_fidelity}
We include the proof from Ref.~\cite{Coles2012} here for completeness. In this proof, we consider the pure state shared between Alice, Bob, and Eve: $\rho_{ABE}$, and will use the max-entropy, defined as:
\begin{equation}
	H_{\max} (X|Y)_{\rho} \equiv \max_{\sigma_Y} \log_2 F(\rho_{XY}, \id_X \otimes \sigma_Y).
\end{equation}
The max-entropy is dual to the min-entropy, i.e. $H_{\min}(X|Y) = -H_{\max}(X|Z)$ for any pure state $\rho_{XYZ}$. 

Now, consider the pure state $\tilde{\rho}_{Z_A ABE} = V_{Z_A} \rho_{ABE} V_{Z_A}^{\dag}$ with the isometry $V_{Z_A} = \sum_j \ket{j}_{Z_A} \otimes Z_A^j$ representing Alice's key map.  
We can derive the following series of equalities:
\begin{align}
	&\nonumber H_{\min}(Z_A |E)_{\rho}  \equiv H_{\min}(Z_A|E)_{\tilde{\rho}} = - H_{\max}(Z_A | A B) \\
	&\nonumber = - \log_2 \max_{\sigma_{AB}} F(\tilde{\rho}_{Z_A A B}, \id_{Z_A} \otimes \sigma_{AB}) \\
	&\nonumber = - \log_2 \max_{\sigma_{AB}} F(\tilde{\rho}_{Z_A A B}, V_{Z_A} V_{Z_A}^{\dag} (\id_{Z_A} \otimes \sigma_{AB}) V_{Z_A} V_{Z_A}^{\dag}) \\
	&\nonumber = - \log_2 \max_{\sigma_{AB}} F(\rho_{A B}, V_{Z_A}^{\dag} (\id_{Z_A} \otimes \sigma_{AB}) V_{Z_A} ) \\
	& = - \log_2 \max_{\sigma_{AB}} F(\rho_{A B}, \sum_j Z_A^j \sigma_{AB} Z_A^j), 
	\end{align}
which relates the min-entropy to a maximization of the fidelity. The third line is true because an isometry must satisfy $V_{Z_A} V_{Z_A}^{\dag} = \id$, and the fourth line uses the fact that fidelity is invariant under isometries.

\subsection*{SDP for quantum relative entropy}
\label{app:QRE_SDP}

Let us express the problem in terms of a convex optimization problem with quantum relative entropy as the objective function:
\begin{equation}
\label{eq:key_rate_nonasymptotic_convex_primal}
\begin{aligned}
		& \text{minimize} 
		& & \infdiv*{\rho_{AB}}{\sum_j Z_A^j \rho_{AB} Z_A^j} \\
		& \text{subject to}
		& & \Tr(\rho_{AB} \Gamma_i) \leq \gamma_i^{\UB} \text{ and } \Tr(\rho_{AB} \Gamma_i) \geq \gamma_i^{\LB}\\ 
		&&& \qquad \text{for } i = 1, \dots, \nPE, \\
		&&& \Tr(\rho_{AB}) = 1, \\
		&&& \rho_{AB} \succeq 0.
\end{aligned}
\end{equation}
Here $Z^j_A$ are projectors onto the signal-state basis of the Hilbert space of $A$.

The SDP for the $(m,k)$-approximation of the quantum relative entropy in this case is \cite{Fawzi2017}:
\begin{equation}
\label{eq:key_rate_nonasymptotic_appx_primal}
\begin{aligned}
	& \text{minimize}
	& & \tau \\
	& \text{subject to}
	& & \begin{pmatrix}
			M_i & M_{i+1} \\
		M_{i+1} & X
	\end{pmatrix} \succeq 0  \text{ for } i=0, \dots, k-1, \\
	&&& \begin{pmatrix}
			\braket{e|X|e} - s_j t_j/w_j & \bra{e} X \\
			X \ket{e} & X + s_j (Z-X)
		\end{pmatrix} \succeq 0  \\
	&&& \qquad \text{ for } j=1, \dots, m, \\
	&&& \sum_j t_j 2^{k} + \tau \geq 0, \\
	&&& Y = M_0, \; Z = M_{k}, \\
	&&& X = \rho_{AB} \otimes \id, \; Y = \id \otimes \sum_{j} Z_A^j \rho_{AB} Z_A^j, \\
	&&& \Tr(\rho_{AB} \Gamma_k) \leq \gamma_k^{\UB} \text{ and } \Tr(\rho_{AB} \Gamma_k) \geq \gamma_k^{\LB}  \\
	&&& \qquad \text{ for } k = 1, \dots, \nPE, \\
	&&& \Tr(\rho_{AB}) = 1, \\
	&&& \rho_{AB} \succeq 0,
\end{aligned}
\end{equation}
where $w_j$ and $s_j$ are the weights and nodes for the $m$-point Gauss-Legendre quadrature on interval $[0,1]$. Here, $\ket{e}$ is the vector obtained by vertically stacking the columns of an identity matrix.

Solving the approximate problem above only gives us a density matrix $\hat{\rho}_{AB}$ that is close to the optimal matrix $\rho^*_{AB}$. However, as it pointed out by Ref.~\cite{Winick2017}, we can use this close-to-optimal density matrix $\hat{\rho}_{AB}$ and find a secure lower bound through linearization of the convex objective function: $f(\rho) \equiv \infdiv{\rho}{Z_A^j \rho Z_A^j}$. Using the fact that this objective function is convex and differentiable, we have:
\begin{equation}
	\label{eq:lower_bound_nonasymptotic}
	\begin{aligned}
		f(\rho^*_{AB}) &\geq f(\hat{\rho}_{AB}) + \Tr\left(\nabla f(\hat{\rho}_{AB})^T (\rho^*_{AB} - \hat{\rho}_{AB}) \right) \\
		& \geq f(\hat{\rho}_{AB}) + \min_{\sigma \in \mathcal{C}} \Tr\left(\nabla f(\hat{\rho}_{AB})^T (\sigma - \hat{\rho}_{AB}) \right) \\
		& = f(\hat{\rho}_{AB}) - \Tr\left( \nabla f(\hat{\rho}_{AB})^T \hat{\rho}_{AB}\right) \\
		& \qquad + \min_{\sigma \in \mathcal{C}_{\epspe}} \Tr\left(\nabla f(\hat{\rho}_{AB})^T \sigma \right),
	\end{aligned}
\end{equation}
where
\begin{equation}
\label{eq:grad_f_rho}
	\nabla f(\rho_{AB})^T = \log_2 \rho_{AB} - \log_2 \left(\sum_j Z_A^j \rho_{AB} Z_A^j \right).
\end{equation}
The primal and dual SDPs for the last term in Eq.~\eqref{eq:lower_bound_nonasymptotic} are:
\begin{primal}
	\label{eq:convex_diff_primal_problem}
	\begin{equation}
	\begin{aligned}
		& \text{minimize}
		& & \Tr( \nabla f(\hat{\rho}_{AB})^T \sigma) \\
		& \text{subject to}
		& & \Tr(\sigma \Gamma_i) \leq \gamma_i^{\UB} \text{ and } \Tr(\sigma \Gamma_i) \geq \gamma_i^{\LB} \\
		&&& \qquad \text{ for } i = 1, \dots, \nPE,\\
		&&& \Tr(\sigma) = 1, \\
		&&& \sigma \succeq 0,
	\end{aligned}
	\end{equation}
\end{primal}
\noindent and
\begin{dual}
	\label{eq:convex_diff_dual_problem}
	\begin{equation}
	\begin{aligned}
		& \text{maximize}
		& & z + \sum_{i=1}^{\nPE} \left( x_i \gamma_i^{\LB} - y_i \gamma_i^{\UB} \right)\\
		& \text{subject to}
		& & z \id + \sum_{i=1}^{\nPE} \left( x_i - y_i \right) \Gamma_i \preceq \nabla f(\hat{\rho}_{AB})^T, \\
		&&& x_i \geq 0, \: y_i \geq 0, \text{ for } i = 1, \dots, \nPE.
	\end{aligned}
	\end{equation}
\end{dual}

Key rate problems for some of the more well-known QKD protocols, e.g. the BB84 protocol, can be solved efficiently (within a second on a personal computer) with Eq.~\eqref{eq:key_rate_nonasymptotic_appx_primal} using a commercial or an open-source SDP solver, e.g. \texttt{Mosek}~\citep{mosek} or \texttt{SeDuMi}~\citep{sedumi}. Some larger problems, such as the prepare-and-measure protocol or the measurement-device-independent QKD protocol, require simplification that makes use of the block diagonal structure of the density operator $\rho_{AB}$ to be efficiently solved, see Supplementary Note~\ref{app:simplification_postselection} and Refs.~\cite{Lin:2019aa,govia2019clifford}.

The main inefficiency in our formulation comes during step one. Suppose that $\rho_{AB}$ is an $n \times n$ matrix, then $X, Y, Z, $ and $M_i$ matrices are of size $n^2 \times n^2$. Therefore, the problem needs to solve a total of $k$ blocks of $2n^2 \times 2n^2$ positive semidefinite matrices along with another $m$ blocks of $(n^2+1) \times (n^2+1)$ positive semidefinite matrices. It is therefore desirable to find another approximation method that requires a smaller number of parameters.

\subsection*{SDP for min-entropy and quantum fidelity}
\label{app:SDP_min}

Ref.~\cite{Watrous2009} shows how the fidelity can be expressed in terms of a simple linear SDP. The primal and dual SDP problems for computing the $\sqrt{F(P,Q)}$ between two operators $P \succeq 0$ and $Q \succeq 0$ are as follows:
\begin{primal}
	\label{eq:fidelity_primal_problem}
	\begin{equation}
	\begin{aligned}
		& \text{minimize}
		& & \Tr(P Y_{11}) + \Tr(Q Y_{22}) \\
		& \text{subject to}
		& & \begin{pmatrix}
			Y_{11} & 0 \\
			0 & Y_{22}
		\end{pmatrix} \succeq \frac{1}{2}
		\begin{pmatrix}
			0 & \id \\
			\id & 0
		\end{pmatrix}. \\
		&&& Y_{11} \succeq 0 , \; Y_{22} \succeq 0.
	\end{aligned}
	\end{equation}
\end{primal}
\noindent and
\begin{dual}
	\label{eq:fidelity_dual_problem}
	\begin{equation}
	\begin{aligned}
		& \text{maximize}
		& & \frac{1}{2} \left(\Tr X_{12} + \Tr X_{12}^{\dag} \right)\\
		& \text{subject to}
		& & \begin{pmatrix}
			X_{11} & X_{12} \\
			X_{12}^{\dag} & X_{22}
		\end{pmatrix} \succeq 0, \\
		&&& X_{11} \preceq P, \; X_{22} \preceq Q.
	\end{aligned}
	\end{equation}
\end{dual}

We can therefore formulate the following optimization problem:
\begin{equation}
\begin{aligned}
	g(\mathcal{C}_{\epspe}) &\equiv \min_{\rho_{AB} \in \mathcal{C}_{\epspe}} H_{\min}(Z_A |E) \\
	&= \min_{\rho_{AB} \in \mathcal{C}_{\epspe}} \left[- \log_2 \max_{\sigma_{AB}} F(\rho_{A B}, \sum_j Z_A^j \sigma_{AB} Z_A^j)\right] \\
	&= - \log_2 \max_{\rho_{AB} \in \mathcal{C}_{\epspe}} \max_{\sigma_{AB}} F(\rho_{A B}, \sum_j Z_A^j \sigma_{AB} Z_A^j) \\
	&= - 2 \log_2 \max_{\rho_{AB} \in \mathcal{C}_{\epspe}} \max_{\sigma_{AB}} \sqrt{F(\rho_{A B}, \sum_j Z_A^j \sigma_{AB} Z_A^j)}.
\end{aligned}
\end{equation}
In particular, we can compute the following quantity:
\begin{equation}
	2^{-g(\mathcal{C}_{\epspe})/2} = \max_{\rho_{AB} \in \mathcal{C}_{\epspe}} \max_{\sigma_{AB}} \sqrt{F(\rho_{A B}, \sum_j Z_A^j \sigma_{AB} Z_A^j)}.
\end{equation}

The primal SDP problem for the quantity $2^{-g(\mathcal{C}_{\epspe})/2}$ is
\begin{primal}
	\begin{equation}
	\label{eq:keyrate_Hmin_primal_problem}
	\begin{aligned}
		& \text{maximize}
		& & \frac{1}{2} \left(\Tr X_{12} + \Tr X_{12}^{\dag} \right)\\
		& \text{subject to}
		& & X_{11} \preceq \rho_{AB}, \; X_{22} \preceq \sum_j Z_A^j \sigma_{AB} Z_A^j, \\
		&&& \Tr(\rho_{AB} \Gamma_i) \leq \gamma_i^{\UB} \text{ and } \Tr(\rho_{AB} \Gamma_i) \geq \gamma_i^{\LB} \\
		&&& \qquad \text{ for } i = 1, \dots, \nPE,\\
		&&& \Tr(\rho_{AB}) = 1, \; \Tr(\sigma_{AB}) = 1, \\
		&&& \begin{pmatrix}
			X_{11} & X_{12} \\
			X_{12}^{\dag} & X_{22}
		\end{pmatrix} \succeq 0, \\
		&&& \rho_{AB} \succeq 0, \; \sigma_{AB} \succeq 0,
	\end{aligned}
	\end{equation}
\end{primal}
\noindent that can be transformed into the following dual problem:
\begin{dual}
	\begin{equation}
	\label{eq:keyrate_Hmin_dual_problem}
	\begin{aligned}
		& \text{minimize}
		& & z + \bar{y} + \sum_{i=1}^{\nPE} \left( y_i \gamma_i^{\UB} - x_i \gamma_i^{\LB} \right)\\
		& \text{subject to}
		& & \sum_{i=1}^{\nPE} (y_i - x_i) \Gamma_i + \bar{y} \id \succeq Y_{11}, \\ 
		&&& z \id \succeq \sum_j Z_A^j Y_{22} Z_A^j, \\
		&&& \begin{pmatrix}
			Y_{11} & 0 \\
			0 & Y_{22}
		\end{pmatrix} \succeq \frac{1}{2}
		\begin{pmatrix}
			0 & \id \\
			\id & 0
		\end{pmatrix}. \\
		&&& Y_{11} \succeq 0 , \; Y_{22} \succeq 0, \; y_i \geq 0, \; x_i \geq 0.
	\end{aligned}
	\end{equation}
\end{dual}
\noindent Solving the dual problem~\eqref{eq:keyrate_Hmin_dual_problem} directly provides a reliable lower bound to the key rate. The min-entropy SDP derived here has a computational advantage over the von Neumann SDP, due to the fact that other than the positive real numbers $x_i$ and $y_i$, only two matrices $Y_{11}$ and $Y_{22}$---both of the same size as the density matrix $\rho_{AB}$---have to be computed.

\bibliography{References/pic,References/qkd,References/reviews,References/post_processing,References/tomo,References/references,References/OtherRefs}

\newcommand{\noopsort}[1]{}
\begin{thebibliography}{66}%
\makeatletter
\providecommand \@ifxundefined [1]{%
 \@ifx{#1\undefined}
}%
\providecommand \@ifnum [1]{%
 \ifnum #1\expandafter \@firstoftwo
 \else \expandafter \@secondoftwo
 \fi
}%
\providecommand \@ifx [1]{%
 \ifx #1\expandafter \@firstoftwo
 \else \expandafter \@secondoftwo
 \fi
}%
\providecommand \natexlab [1]{#1}%
\providecommand \enquote  [1]{``#1''}%
\providecommand \bibnamefont  [1]{#1}%
\providecommand \bibfnamefont [1]{#1}%
\providecommand \citenamefont [1]{#1}%
\providecommand \href@noop [0]{\@secondoftwo}%
\providecommand \href [0]{\begingroup \@sanitize@url \@href}%
\providecommand \@href[1]{\@@startlink{#1}\@@href}%
\providecommand \@@href[1]{\endgroup#1\@@endlink}%
\providecommand \@sanitize@url [0]{\catcode `\\12\catcode `\$12\catcode
  `\&12\catcode `\#12\catcode `\^12\catcode `\_12\catcode `\%12\relax}%
\providecommand \@@startlink[1]{}%
\providecommand \@@endlink[0]{}%
\providecommand \url  [0]{\begingroup\@sanitize@url \@url }%
\providecommand \@url [1]{\endgroup\@href {#1}{\urlprefix }}%
\providecommand \urlprefix  [0]{URL }%
\providecommand \Eprint [0]{\href }%
\providecommand \doibase [0]{http://dx.doi.org/}%
\providecommand \selectlanguage [0]{\@gobble}%
\providecommand \bibinfo  [0]{\@secondoftwo}%
\providecommand \bibfield  [0]{\@secondoftwo}%
\providecommand \translation [1]{[#1]}%
\providecommand \BibitemOpen [0]{}%
\providecommand \bibitemStop [0]{}%
\providecommand \bibitemNoStop [0]{.\EOS\space}%
\providecommand \EOS [0]{\spacefactor3000\relax}%
\providecommand \BibitemShut  [1]{\csname bibitem#1\endcsname}%
\let\auto@bib@innerbib\@empty
%</preamble>
\bibitem [{\citenamefont {Pirandola}\ \emph {et~al.}(2019)\citenamefont
  {Pirandola}, \citenamefont {Andersen}, \citenamefont {Banchi}, \citenamefont
  {Berta}, \citenamefont {Bunandar}, \citenamefont {Colbeck}, \citenamefont
  {Englund}, \citenamefont {Gehring}, \citenamefont {Lupo}, \citenamefont
  {Ottaviani}, \citenamefont {Pereira}, \citenamefont {Razavi}, \citenamefont
  {Shaari}, \citenamefont {Tomamichel}, \citenamefont {Usenko}, \citenamefont
  {Vallone}, \citenamefont {Villoresi},\ and\ \citenamefont
  {Wallden}}]{pir2019advances}%
  \BibitemOpen
  \bibfield  {author} {\bibinfo {author} {\bibfnamefont {S.}~\bibnamefont
  {Pirandola}}, \bibinfo {author} {\bibfnamefont {U.~L.}\ \bibnamefont
  {Andersen}}, \bibinfo {author} {\bibfnamefont {L.}~\bibnamefont {Banchi}},
  \bibinfo {author} {\bibfnamefont {M.}~\bibnamefont {Berta}}, \bibinfo
  {author} {\bibfnamefont {D.}~\bibnamefont {Bunandar}}, \bibinfo {author}
  {\bibfnamefont {R.}~\bibnamefont {Colbeck}}, \bibinfo {author} {\bibfnamefont
  {D.}~\bibnamefont {Englund}}, \bibinfo {author} {\bibfnamefont
  {T.}~\bibnamefont {Gehring}}, \bibinfo {author} {\bibfnamefont
  {C.}~\bibnamefont {Lupo}}, \bibinfo {author} {\bibfnamefont {C.}~\bibnamefont
  {Ottaviani}}, \bibinfo {author} {\bibfnamefont {J.}~\bibnamefont {Pereira}},
  \bibinfo {author} {\bibfnamefont {M.}~\bibnamefont {Razavi}}, \bibinfo
  {author} {\bibfnamefont {J.~S.}\ \bibnamefont {Shaari}}, \bibinfo {author}
  {\bibfnamefont {M.}~\bibnamefont {Tomamichel}}, \bibinfo {author}
  {\bibfnamefont {V.~C.}\ \bibnamefont {Usenko}}, \bibinfo {author}
  {\bibfnamefont {G.}~\bibnamefont {Vallone}}, \bibinfo {author} {\bibfnamefont
  {P.}~\bibnamefont {Villoresi}}, \ and\ \bibinfo {author} {\bibfnamefont
  {P.}~\bibnamefont {Wallden}},\ }\href@noop {} {\enquote {\bibinfo {title}
  {Advances in quantum cryptography},}\ } (\bibinfo {year} {2019}),\ \Eprint
  {http://arxiv.org/abs/1906.01645} {arXiv:1906.01645 [quant-ph]} \BibitemShut
  {NoStop}%
\bibitem [{\citenamefont {Bennett}\ and\ \citenamefont
  {Brassard}(1984)}]{BB84}%
  \BibitemOpen
  \bibfield  {author} {\bibinfo {author} {\bibfnamefont {C.~H.}\ \bibnamefont
  {Bennett}}\ and\ \bibinfo {author} {\bibfnamefont {G.}~\bibnamefont
  {Brassard}},\ }in\ \href@noop {} {\emph {\bibinfo {booktitle} {Proceedings of
  IEEE International Conference on Computers, Systems, and Signal
  Processing}}}\ (\bibinfo  {publisher} {IEEE},\ \bibinfo {year} {1984})\ pp.\
  \bibinfo {pages} {175--179}\BibitemShut {NoStop}%
\bibitem [{\citenamefont {Bru{\ss}}(1998)}]{Bruss1998}%
  \BibitemOpen
  \bibfield  {author} {\bibinfo {author} {\bibfnamefont {D.}~\bibnamefont
  {Bru{\ss}}},\ }\href {\doibase 10.1103/PhysRevLett.81.3018} {\bibfield
  {journal} {\bibinfo  {journal} {Physical Review Letters}\ }\textbf {\bibinfo
  {volume} {81}},\ \bibinfo {pages} {3018} (\bibinfo {year}
  {1998})}\BibitemShut {NoStop}%
\bibitem [{\citenamefont {Gottesman}\ \emph {et~al.}(2004)\citenamefont
  {Gottesman}, \citenamefont {Lo}, \citenamefont {L{\"u}tkenhaus},\ and\
  \citenamefont {Preskill}}]{Gottesman2004a}%
  \BibitemOpen
  \bibfield  {author} {\bibinfo {author} {\bibfnamefont {D.}~\bibnamefont
  {Gottesman}}, \bibinfo {author} {\bibfnamefont {H.-k.}\ \bibnamefont {Lo}},
  \bibinfo {author} {\bibfnamefont {N.}~\bibnamefont {L{\"u}tkenhaus}}, \ and\
  \bibinfo {author} {\bibfnamefont {J.}~\bibnamefont {Preskill}},\ }\href
  {http://arxiv.org/abs/quant-ph/0212066} {\bibfield  {journal} {\bibinfo
  {journal} {Quantum Information \& Computation}\ }\textbf {\bibinfo {volume}
  {4}},\ \bibinfo {pages} {325} (\bibinfo {year} {2004})}\BibitemShut {NoStop}%
\bibitem [{\citenamefont {Coles}\ \emph {et~al.}(2016)\citenamefont {Coles},
  \citenamefont {Metodiev},\ and\ \citenamefont {L{\"u}tkenhaus}}]{Coles2016}%
  \BibitemOpen
  \bibfield  {author} {\bibinfo {author} {\bibfnamefont {P.~J.}\ \bibnamefont
  {Coles}}, \bibinfo {author} {\bibfnamefont {E.~M.}\ \bibnamefont {Metodiev}},
  \ and\ \bibinfo {author} {\bibfnamefont {N.}~\bibnamefont {L{\"u}tkenhaus}},\
  }\href {\doibase 10.1038/ncomms11712} {\bibfield  {journal} {\bibinfo
  {journal} {Nature Communications}\ }\textbf {\bibinfo {volume} {7}},\
  \bibinfo {pages} {11712} (\bibinfo {year} {2016})}\BibitemShut {NoStop}%
\bibitem [{\citenamefont {Winick}\ \emph {et~al.}(2018)\citenamefont {Winick},
  \citenamefont {L{\"u}tkenhaus},\ and\ \citenamefont {Coles}}]{Winick2017}%
  \BibitemOpen
  \bibfield  {author} {\bibinfo {author} {\bibfnamefont {A.}~\bibnamefont
  {Winick}}, \bibinfo {author} {\bibfnamefont {N.}~\bibnamefont
  {L{\"u}tkenhaus}}, \ and\ \bibinfo {author} {\bibfnamefont {P.~J.}\
  \bibnamefont {Coles}},\ }\href {\doibase 10.22331/q-2018-07-26-77} {\bibfield
   {journal} {\bibinfo  {journal} {Quantum}\ }\textbf {\bibinfo {volume} {2}},\
  \bibinfo {pages} {77} (\bibinfo {year} {2018})}\BibitemShut {NoStop}%
\bibitem [{\citenamefont {ApS}(2017)}]{mosek}%
  \BibitemOpen
  \bibfield  {author} {\bibinfo {author} {\bibfnamefont {M.}~\bibnamefont
  {ApS}},\ }\href {http://docs.mosek.com/8.1/toolbox/index.html} {\emph
  {\bibinfo {title} {The MOSEK optimization toolbox for MATLAB.}}} (\bibinfo
  {year} {2017})\BibitemShut {NoStop}%
\bibitem [{\citenamefont {Sturm}(1999)}]{sedumi}%
  \BibitemOpen
  \bibfield  {author} {\bibinfo {author} {\bibfnamefont {J.~F.}\ \bibnamefont
  {Sturm}},\ }\href {\doibase 10.1080/10556789908805766} {\bibfield  {journal}
  {\bibinfo  {journal} {Optim. Method. Softw.}\ }\textbf {\bibinfo {volume}
  {11}},\ \bibinfo {pages} {625} (\bibinfo {year} {1999})}\BibitemShut
  {NoStop}%
\bibitem [{\citenamefont {T{\"u}t{\"u}nc{\"u}}\ \emph
  {et~al.}(2003)\citenamefont {T{\"u}t{\"u}nc{\"u}}, \citenamefont {Toh},\ and\
  \citenamefont {Todd}}]{sdpt3}%
  \BibitemOpen
  \bibfield  {author} {\bibinfo {author} {\bibfnamefont {R.~H.}\ \bibnamefont
  {T{\"u}t{\"u}nc{\"u}}}, \bibinfo {author} {\bibfnamefont {K.~C.}\
  \bibnamefont {Toh}}, \ and\ \bibinfo {author} {\bibfnamefont {M.~J.}\
  \bibnamefont {Todd}},\ }\href {\doibase 10.1007/s10107-002-0347-5} {\bibfield
   {journal} {\bibinfo  {journal} {Mathematical Programming}\ }\textbf
  {\bibinfo {volume} {95}},\ \bibinfo {pages} {189} (\bibinfo {year}
  {2003})}\BibitemShut {NoStop}%
\bibitem [{\citenamefont {Renner}(2008)}]{Renner2005a}%
  \BibitemOpen
  \bibfield  {author} {\bibinfo {author} {\bibfnamefont {R.}~\bibnamefont
  {Renner}},\ }\href {\doibase 10.1142/S0219749908003256} {\bibfield  {journal}
  {\bibinfo  {journal} {Int. J. Quantum Inf.}\ }\textbf {\bibinfo {volume}
  {06}},\ \bibinfo {pages} {1} (\bibinfo {year} {2008})}\BibitemShut {NoStop}%
\bibitem [{\citenamefont {Tomamichel}\ \emph {et~al.}(2011)\citenamefont
  {Tomamichel}, \citenamefont {Schaffner}, \citenamefont {Smith},\ and\
  \citenamefont {Renner}}]{Tomamichel2011a}%
  \BibitemOpen
  \bibfield  {author} {\bibinfo {author} {\bibfnamefont {M.}~\bibnamefont
  {Tomamichel}}, \bibinfo {author} {\bibfnamefont {C.}~\bibnamefont
  {Schaffner}}, \bibinfo {author} {\bibfnamefont {A.}~\bibnamefont {Smith}}, \
  and\ \bibinfo {author} {\bibfnamefont {R.}~\bibnamefont {Renner}},\ }\href
  {\doibase 10.1109/TIT.2011.2158473} {\bibfield  {journal} {\bibinfo
  {journal} {IEEE Trans. Inform. Theory}\ }\textbf {\bibinfo {volume} {57}},\
  \bibinfo {pages} {5524} (\bibinfo {year} {2011})},\ \Eprint
  {http://arxiv.org/abs/1002.2436} {arXiv:1002.2436} \BibitemShut {NoStop}%
\bibitem [{\citenamefont {Scarani}\ and\ \citenamefont
  {Renner}(2008)}]{Scarani2008}%
  \BibitemOpen
  \bibfield  {author} {\bibinfo {author} {\bibfnamefont {V.}~\bibnamefont
  {Scarani}}\ and\ \bibinfo {author} {\bibfnamefont {R.}~\bibnamefont
  {Renner}},\ }\href {\doibase 10.1103/PhysRevLett.100.200501} {\bibfield
  {journal} {\bibinfo  {journal} {Physical Review Letters}\ }\textbf {\bibinfo
  {volume} {100}},\ \bibinfo {pages} {200501} (\bibinfo {year}
  {2008})}\BibitemShut {NoStop}%
\bibitem [{\citenamefont {Cai}\ and\ \citenamefont {Scarani}(2009)}]{Cai2009}%
  \BibitemOpen
  \bibfield  {author} {\bibinfo {author} {\bibfnamefont {R.~Y.~Q.}\
  \bibnamefont {Cai}}\ and\ \bibinfo {author} {\bibfnamefont {V.}~\bibnamefont
  {Scarani}},\ }\href {\doibase 10.1088/1367-2630/11/4/045024} {\bibfield
  {journal} {\bibinfo  {journal} {New Journal of Physics}\ }\textbf {\bibinfo
  {volume} {11}},\ \bibinfo {pages} {045024} (\bibinfo {year}
  {2009})}\BibitemShut {NoStop}%
\bibitem [{\citenamefont {Curty}\ \emph {et~al.}(2014)\citenamefont {Curty},
  \citenamefont {Xu}, \citenamefont {Cui}, \citenamefont {Lim}, \citenamefont
  {Tamaki},\ and\ \citenamefont {Lo}}]{Curty2014}%
  \BibitemOpen
  \bibfield  {author} {\bibinfo {author} {\bibfnamefont {M.}~\bibnamefont
  {Curty}}, \bibinfo {author} {\bibfnamefont {F.}~\bibnamefont {Xu}}, \bibinfo
  {author} {\bibfnamefont {W.}~\bibnamefont {Cui}}, \bibinfo {author}
  {\bibfnamefont {C.~C.~W.}\ \bibnamefont {Lim}}, \bibinfo {author}
  {\bibfnamefont {K.}~\bibnamefont {Tamaki}}, \ and\ \bibinfo {author}
  {\bibfnamefont {H.-K.}\ \bibnamefont {Lo}},\ }\href {\doibase
  10.1038/ncomms4732} {\bibfield  {journal} {\bibinfo  {journal} {Nature
  Communications}\ }\textbf {\bibinfo {volume} {5}},\ \bibinfo {pages} {3732}
  (\bibinfo {year} {2014})}\BibitemShut {NoStop}%
\bibitem [{\citenamefont {Lim}\ \emph {et~al.}(2014)\citenamefont {Lim},
  \citenamefont {Curty}, \citenamefont {Walenta}, \citenamefont {Xu},\ and\
  \citenamefont {Zbinden}}]{Lim2014}%
  \BibitemOpen
  \bibfield  {author} {\bibinfo {author} {\bibfnamefont {C.~C.~W.}\
  \bibnamefont {Lim}}, \bibinfo {author} {\bibfnamefont {M.}~\bibnamefont
  {Curty}}, \bibinfo {author} {\bibfnamefont {N.}~\bibnamefont {Walenta}},
  \bibinfo {author} {\bibfnamefont {F.}~\bibnamefont {Xu}}, \ and\ \bibinfo
  {author} {\bibfnamefont {H.}~\bibnamefont {Zbinden}},\ }\href {\doibase
  10.1103/PhysRevA.89.022307} {\bibfield  {journal} {\bibinfo  {journal}
  {Physical Review A}\ }\textbf {\bibinfo {volume} {89}},\ \bibinfo {pages}
  {022307} (\bibinfo {year} {2014})}\BibitemShut {NoStop}%
\bibitem [{\citenamefont {Zhang}\ \emph {et~al.}(2017)\citenamefont {Zhang},
  \citenamefont {Zhao}, \citenamefont {Razavi},\ and\ \citenamefont
  {Ma}}]{Zhang2017}%
  \BibitemOpen
  \bibfield  {author} {\bibinfo {author} {\bibfnamefont {Z.}~\bibnamefont
  {Zhang}}, \bibinfo {author} {\bibfnamefont {Q.}~\bibnamefont {Zhao}},
  \bibinfo {author} {\bibfnamefont {M.}~\bibnamefont {Razavi}}, \ and\ \bibinfo
  {author} {\bibfnamefont {X.}~\bibnamefont {Ma}},\ }\href {\doibase
  10.1103/PhysRevA.95.012333} {\bibfield  {journal} {\bibinfo  {journal}
  {Physical Review A}\ }\textbf {\bibinfo {volume} {95}},\ \bibinfo {pages}
  {012333} (\bibinfo {year} {2017})}\BibitemShut {NoStop}%
\bibitem [{\citenamefont {Fawzi}\ \emph {et~al.}(2018)\citenamefont {Fawzi},
  \citenamefont {Saunderson},\ and\ \citenamefont {Parrilo}}]{Fawzi2017}%
  \BibitemOpen
  \bibfield  {author} {\bibinfo {author} {\bibfnamefont {H.}~\bibnamefont
  {Fawzi}}, \bibinfo {author} {\bibfnamefont {J.}~\bibnamefont {Saunderson}}, \
  and\ \bibinfo {author} {\bibfnamefont {P.~A.}\ \bibnamefont {Parrilo}},\
  }\href {http://arxiv.org/abs/1705.00812} {\bibfield  {journal} {\bibinfo
  {journal} {Foundations of Computational Mathematics}\ } (\bibinfo {year}
  {2018})}\BibitemShut {NoStop}%
\bibitem [{\citenamefont {Bratzik}\ \emph {et~al.}(2011)\citenamefont
  {Bratzik}, \citenamefont {Mertz}, \citenamefont {Kampermann},\ and\
  \citenamefont {Bru{\ss}}}]{bratzikMinentropyQuantumKey2011}%
  \BibitemOpen
  \bibfield  {author} {\bibinfo {author} {\bibfnamefont {S.}~\bibnamefont
  {Bratzik}}, \bibinfo {author} {\bibfnamefont {M.}~\bibnamefont {Mertz}},
  \bibinfo {author} {\bibfnamefont {H.}~\bibnamefont {Kampermann}}, \ and\
  \bibinfo {author} {\bibfnamefont {D.}~\bibnamefont {Bru{\ss}}},\ }\href
  {\doibase 10.1103/PhysRevA.83.022330} {\bibfield  {journal} {\bibinfo
  {journal} {Physical Review A}\ }\textbf {\bibinfo {volume} {83}},\ \bibinfo
  {pages} {022330} (\bibinfo {year} {2011})}\BibitemShut {NoStop}%
\bibitem [{\citenamefont {Coles}(2012)}]{Coles2012}%
  \BibitemOpen
  \bibfield  {author} {\bibinfo {author} {\bibfnamefont {P.~J.}\ \bibnamefont
  {Coles}},\ }\href {\doibase 10.1103/PhysRevA.85.042103} {\bibfield  {journal}
  {\bibinfo  {journal} {Physical Review A}\ }\textbf {\bibinfo {volume} {85}},\
  \bibinfo {pages} {042103} (\bibinfo {year} {2012})}\BibitemShut {NoStop}%
\bibitem [{\citenamefont {Watrous}(2012)}]{Watrous2009}%
  \BibitemOpen
  \bibfield  {author} {\bibinfo {author} {\bibfnamefont {J.}~\bibnamefont
  {Watrous}},\ }\href {\doibase 10.4086/toc.2009.v005a011} {\bibfield
  {journal} {\bibinfo  {journal} {Theory of Computing}\ }\textbf {\bibinfo
  {volume} {5}},\ \bibinfo {pages} {217} (\bibinfo {year} {2012})}\BibitemShut
  {NoStop}%
\bibitem [{\citenamefont {Bennett}\ \emph {et~al.}(1992)\citenamefont
  {Bennett}, \citenamefont {Bessette}, \citenamefont {Brassard}, \citenamefont
  {Salvail},\ and\ \citenamefont {Smolin}}]{Bennett1992}%
  \BibitemOpen
  \bibfield  {author} {\bibinfo {author} {\bibfnamefont {C.~H.}\ \bibnamefont
  {Bennett}}, \bibinfo {author} {\bibfnamefont {F.}~\bibnamefont {Bessette}},
  \bibinfo {author} {\bibfnamefont {G.}~\bibnamefont {Brassard}}, \bibinfo
  {author} {\bibfnamefont {L.}~\bibnamefont {Salvail}}, \ and\ \bibinfo
  {author} {\bibfnamefont {J.}~\bibnamefont {Smolin}},\ }\href {\doibase
  10.1007/BF00191318} {\bibfield  {journal} {\bibinfo  {journal} {Journal of
  Cryptology}\ }\textbf {\bibinfo {volume} {5}},\ \bibinfo {pages} {3}
  (\bibinfo {year} {1992})}\BibitemShut {NoStop}%
\bibitem [{\citenamefont {Lucamarini}\ \emph {et~al.}(2018)\citenamefont
  {Lucamarini}, \citenamefont {Yuan}, \citenamefont {Dynes},\ and\
  \citenamefont {Shields}}]{Lucamarini2018}%
  \BibitemOpen
  \bibfield  {author} {\bibinfo {author} {\bibfnamefont {M.}~\bibnamefont
  {Lucamarini}}, \bibinfo {author} {\bibfnamefont {Z.~L.}\ \bibnamefont
  {Yuan}}, \bibinfo {author} {\bibfnamefont {J.~F.}\ \bibnamefont {Dynes}}, \
  and\ \bibinfo {author} {\bibfnamefont {A.~J.}\ \bibnamefont {Shields}},\
  }\href {\doibase 10.1038/s41586-018-0066-6} {\bibfield  {journal} {\bibinfo
  {journal} {Nature}\ }\textbf {\bibinfo {volume} {557}},\ \bibinfo {pages}
  {400} (\bibinfo {year} {2018})}\BibitemShut {NoStop}%
\bibitem [{\citenamefont {Pirandola}\ \emph {et~al.}(2017)\citenamefont
  {Pirandola}, \citenamefont {Laurenza}, \citenamefont {Ottaviani},\ and\
  \citenamefont {Banchi}}]{Pirandola2017}%
  \BibitemOpen
  \bibfield  {author} {\bibinfo {author} {\bibfnamefont {S.}~\bibnamefont
  {Pirandola}}, \bibinfo {author} {\bibfnamefont {R.}~\bibnamefont {Laurenza}},
  \bibinfo {author} {\bibfnamefont {C.}~\bibnamefont {Ottaviani}}, \ and\
  \bibinfo {author} {\bibfnamefont {L.}~\bibnamefont {Banchi}},\ }\href
  {\doibase 10.1038/ncomms15043} {\bibfield  {journal} {\bibinfo  {journal}
  {Nature Communications}\ }\textbf {\bibinfo {volume} {8}},\ \bibinfo {pages}
  {15043} (\bibinfo {year} {2017})}\BibitemShut {NoStop}%
\bibitem [{\citenamefont {Diamond}\ and\ \citenamefont {Boyd}(2016)}]{cvxpy}%
  \BibitemOpen
  \bibfield  {author} {\bibinfo {author} {\bibfnamefont {S.}~\bibnamefont
  {Diamond}}\ and\ \bibinfo {author} {\bibfnamefont {S.}~\bibnamefont {Boyd}},\
  }\href@noop {} {\bibfield  {journal} {\bibinfo  {journal} {Journal of Machine
  Learning Research}\ }\textbf {\bibinfo {volume} {17}},\ \bibinfo {pages} {1}
  (\bibinfo {year} {2016})}\BibitemShut {NoStop}%
\bibitem [{\citenamefont {Akshay~Agrawal}\ and\ \citenamefont
  {Boyd}(2018)}]{cvxpy_rewriting}%
  \BibitemOpen
  \bibfield  {author} {\bibinfo {author} {\bibfnamefont {S.~D.}\ \bibnamefont
  {Akshay~Agrawal}, \bibfnamefont {Robin~Verschueren}}\ and\ \bibinfo {author}
  {\bibfnamefont {S.}~\bibnamefont {Boyd}},\ }\href@noop {} {\bibfield
  {journal} {\bibinfo  {journal} {Journal of Control and Decision}\ }\textbf
  {\bibinfo {volume} {5}},\ \bibinfo {pages} {42} (\bibinfo {year}
  {2018})}\BibitemShut {NoStop}%
\bibitem [{\citenamefont {CVX~Research}(2012)}]{cvx}%
  \BibitemOpen
  \bibfield  {author} {\bibinfo {author} {\bibfnamefont {I.}~\bibnamefont
  {CVX~Research}},\ }\href@noop {} {\enquote {\bibinfo {title} {{CVX}: Matlab
  software for disciplined convex programming, version 2.0},}\ }\bibinfo
  {howpublished} {\url{http://cvxr.com/cvx}} (\bibinfo {year}
  {2012})\BibitemShut {NoStop}%
\bibitem [{\citenamefont {Grant}\ and\ \citenamefont {Boyd}(2008)}]{cvx-2}%
  \BibitemOpen
  \bibfield  {author} {\bibinfo {author} {\bibfnamefont {M.}~\bibnamefont
  {Grant}}\ and\ \bibinfo {author} {\bibfnamefont {S.}~\bibnamefont {Boyd}},\
  }in\ \href@noop {} {\emph {\bibinfo {booktitle} {Recent Advances in Learning
  and Control}}},\ \bibinfo {series and number} {Lecture Notes in Control and
  Information Sciences},\ \bibinfo {editor} {edited by\ \bibinfo {editor}
  {\bibfnamefont {V.}~\bibnamefont {Blondel}}, \bibinfo {editor} {\bibfnamefont
  {S.}~\bibnamefont {Boyd}}, \ and\ \bibinfo {editor} {\bibfnamefont
  {H.}~\bibnamefont {Kimura}}}\ (\bibinfo  {publisher} {Springer-Verlag
  Limited},\ \bibinfo {year} {2008})\ pp.\ \bibinfo {pages} {95--110},\
  \bibinfo {note} {\url{http://stanford.edu/~boyd/graph_dcp.html}}\BibitemShut
  {NoStop}%
\bibitem [{\citenamefont {Nielsen}\ and\ \citenamefont
  {Chuang}(2000)}]{nielsen2000quantum}%
  \BibitemOpen
  \bibfield  {author} {\bibinfo {author} {\bibfnamefont {M.}~\bibnamefont
  {Nielsen}}\ and\ \bibinfo {author} {\bibfnamefont {I.}~\bibnamefont
  {Chuang}},\ }\href {https://books.google.com/books?id=65FqEKQOfP8C} {\emph
  {\bibinfo {title} {Quantum Computation and Quantum Information}}},\ Cambridge
  Series on Information and the Natural Sciences\ (\bibinfo  {publisher}
  {Cambridge University Press},\ \bibinfo {year} {2000})\BibitemShut {NoStop}%
\bibitem [{\citenamefont {Tamaki}\ and\ \citenamefont
  {L{\"u}tkenhaus}(2004)}]{Tamaki2004}%
  \BibitemOpen
  \bibfield  {author} {\bibinfo {author} {\bibfnamefont {K.}~\bibnamefont
  {Tamaki}}\ and\ \bibinfo {author} {\bibfnamefont {N.}~\bibnamefont
  {L{\"u}tkenhaus}},\ }\href {\doibase 10.1103/PhysRevA.69.032316} {\bibfield
  {journal} {\bibinfo  {journal} {Physical Review A}\ }\textbf {\bibinfo
  {volume} {69}},\ \bibinfo {pages} {032316} (\bibinfo {year}
  {2004})}\BibitemShut {NoStop}%
\bibitem [{\citenamefont {Sasaki}\ \emph {et~al.}(2015)\citenamefont {Sasaki},
  \citenamefont {Matsumoto},\ and\ \citenamefont {Uyematsu}}]{Sasaki2015}%
  \BibitemOpen
  \bibfield  {author} {\bibinfo {author} {\bibfnamefont {H.}~\bibnamefont
  {Sasaki}}, \bibinfo {author} {\bibfnamefont {R.}~\bibnamefont {Matsumoto}}, \
  and\ \bibinfo {author} {\bibfnamefont {T.}~\bibnamefont {Uyematsu}},\ }in\
  \href {\doibase 10.1109/ISIT.2015.7282544} {\emph {\bibinfo {booktitle} {2015
  {{IEEE International Symposium}} on {{Information Theory}}}}},\ Vol.\
  \bibinfo {volume} {2015-June}\ (\bibinfo  {publisher} {{IEEE}},\ \bibinfo
  {year} {2015})\ pp.\ \bibinfo {pages} {696--699}\BibitemShut {NoStop}%
\bibitem [{\citenamefont {Lo}\ \emph {et~al.}(2012)\citenamefont {Lo},
  \citenamefont {Curty},\ and\ \citenamefont {Qi}}]{Lo2012}%
  \BibitemOpen
  \bibfield  {author} {\bibinfo {author} {\bibfnamefont {H.-K.}\ \bibnamefont
  {Lo}}, \bibinfo {author} {\bibfnamefont {M.}~\bibnamefont {Curty}}, \ and\
  \bibinfo {author} {\bibfnamefont {B.}~\bibnamefont {Qi}},\ }\href {\doibase
  10.1103/PhysRevLett.108.130503} {\bibfield  {journal} {\bibinfo  {journal}
  {Physical Review Letters}\ }\textbf {\bibinfo {volume} {108}},\ \bibinfo
  {pages} {130503} (\bibinfo {year} {2012})}\BibitemShut {NoStop}%
\bibitem [{\citenamefont {Tamaki}\ \emph {et~al.}(2018)\citenamefont {Tamaki},
  \citenamefont {Lo}, \citenamefont {Wang},\ and\ \citenamefont
  {Lucamarini}}]{tamakiInformationTheoreticSecurity2018}%
  \BibitemOpen
  \bibfield  {author} {\bibinfo {author} {\bibfnamefont {K.}~\bibnamefont
  {Tamaki}}, \bibinfo {author} {\bibfnamefont {H.-K.}\ \bibnamefont {Lo}},
  \bibinfo {author} {\bibfnamefont {W.}~\bibnamefont {Wang}}, \ and\ \bibinfo
  {author} {\bibfnamefont {M.}~\bibnamefont {Lucamarini}},\ }\href
  {http://arxiv.org/abs/1805.05511} {\bibfield  {journal} {\bibinfo  {journal}
  {arXiv:1805.05511 [quant-ph]}\ } (\bibinfo {year} {2018})},\ \Eprint
  {http://arxiv.org/abs/1805.05511} {arXiv:1805.05511 [quant-ph]} \BibitemShut
  {NoStop}%
\bibitem [{\citenamefont {Curty}\ \emph {et~al.}(2019)\citenamefont {Curty},
  \citenamefont {Azuma},\ and\ \citenamefont
  {Lo}}]{curtySimpleSecurityProof2019}%
  \BibitemOpen
  \bibfield  {author} {\bibinfo {author} {\bibfnamefont {M.}~\bibnamefont
  {Curty}}, \bibinfo {author} {\bibfnamefont {K.}~\bibnamefont {Azuma}}, \ and\
  \bibinfo {author} {\bibfnamefont {H.-K.}\ \bibnamefont {Lo}},\ }\href
  {\doibase 10.1038/s41534-019-0175-6} {\bibfield  {journal} {\bibinfo
  {journal} {npj Quantum Information}\ }\textbf {\bibinfo {volume} {5}},\
  \bibinfo {pages} {64} (\bibinfo {year} {2019})}\BibitemShut {NoStop}%
\bibitem [{\citenamefont {Ma}\ \emph {et~al.}(2018)\citenamefont {Ma},
  \citenamefont {Zeng},\ and\ \citenamefont
  {Zhou}}]{maPhaseMatchingQuantumKey2018}%
  \BibitemOpen
  \bibfield  {author} {\bibinfo {author} {\bibfnamefont {X.}~\bibnamefont
  {Ma}}, \bibinfo {author} {\bibfnamefont {P.}~\bibnamefont {Zeng}}, \ and\
  \bibinfo {author} {\bibfnamefont {H.}~\bibnamefont {Zhou}},\ }\href {\doibase
  10.1103/PhysRevX.8.031043} {\bibfield  {journal} {\bibinfo  {journal}
  {Physical Review X}\ }\textbf {\bibinfo {volume} {8}},\ \bibinfo {pages}
  {031043} (\bibinfo {year} {2018})}\BibitemShut {NoStop}%
\bibitem [{\citenamefont {Jiang}\ \emph {et~al.}(2019)\citenamefont {Jiang},
  \citenamefont {Yu}, \citenamefont {Hu},\ and\ \citenamefont
  {Wang}}]{jiangUnconditionalSecuritySending2019a}%
  \BibitemOpen
  \bibfield  {author} {\bibinfo {author} {\bibfnamefont {C.}~\bibnamefont
  {Jiang}}, \bibinfo {author} {\bibfnamefont {Z.-W.}\ \bibnamefont {Yu}},
  \bibinfo {author} {\bibfnamefont {X.-L.}\ \bibnamefont {Hu}}, \ and\ \bibinfo
  {author} {\bibfnamefont {X.-B.}\ \bibnamefont {Wang}},\ }\href {\doibase
  10.1103/PhysRevApplied.12.024061} {\bibfield  {journal} {\bibinfo  {journal}
  {Physical Review Applied}\ }\textbf {\bibinfo {volume} {12}},\ \bibinfo
  {pages} {024061} (\bibinfo {year} {2019})}\BibitemShut {NoStop}%
\bibitem [{\citenamefont {He}\ \emph {et~al.}(2019)\citenamefont {He},
  \citenamefont {Wang}, \citenamefont {Li},\ and\ \citenamefont
  {Bao}}]{he2019finitekey}%
  \BibitemOpen
  \bibfield  {author} {\bibinfo {author} {\bibfnamefont {S.-F.}\ \bibnamefont
  {He}}, \bibinfo {author} {\bibfnamefont {Y.}~\bibnamefont {Wang}}, \bibinfo
  {author} {\bibfnamefont {H.-W.}\ \bibnamefont {Li}}, \ and\ \bibinfo {author}
  {\bibfnamefont {W.-S.}\ \bibnamefont {Bao}},\ }\href@noop {} {\enquote
  {\bibinfo {title} {Finite-key analysis for a practical decoy-state twin-field
  quantum key distribution},}\ } (\bibinfo {year} {2019}),\ \Eprint
  {http://arxiv.org/abs/1910.12416} {arXiv:1910.12416 [quant-ph]} \BibitemShut
  {NoStop}%
\bibitem [{\citenamefont {Lorenzo}\ \emph {et~al.}(2019)\citenamefont
  {Lorenzo}, \citenamefont {Navarrete}, \citenamefont {Azuma}, \citenamefont
  {Curty},\ and\ \citenamefont {Razavi}}]{lorenzo2019tightfinitekey}%
  \BibitemOpen
  \bibfield  {author} {\bibinfo {author} {\bibfnamefont {G.~C.}\ \bibnamefont
  {Lorenzo}}, \bibinfo {author} {\bibfnamefont {A.}~\bibnamefont {Navarrete}},
  \bibinfo {author} {\bibfnamefont {K.}~\bibnamefont {Azuma}}, \bibinfo
  {author} {\bibfnamefont {M.}~\bibnamefont {Curty}}, \ and\ \bibinfo {author}
  {\bibfnamefont {M.}~\bibnamefont {Razavi}},\ }\href@noop {} {\enquote
  {\bibinfo {title} {Tight finite-key security for twin-field quantum key
  distribution},}\ } (\bibinfo {year} {2019}),\ \Eprint
  {http://arxiv.org/abs/1910.11407} {arXiv:1910.11407 [quant-ph]} \BibitemShut
  {NoStop}%
\bibitem [{\citenamefont {Minder}\ \emph {et~al.}(2019)\citenamefont {Minder},
  \citenamefont {Pittaluga}, \citenamefont {Roberts}, \citenamefont
  {Lucamarini}, \citenamefont {Dynes}, \citenamefont {Yuan},\ and\
  \citenamefont {Shields}}]{minderExperimentalQuantumKey2019}%
  \BibitemOpen
  \bibfield  {author} {\bibinfo {author} {\bibfnamefont {M.}~\bibnamefont
  {Minder}}, \bibinfo {author} {\bibfnamefont {M.}~\bibnamefont {Pittaluga}},
  \bibinfo {author} {\bibfnamefont {G.~L.}\ \bibnamefont {Roberts}}, \bibinfo
  {author} {\bibfnamefont {M.}~\bibnamefont {Lucamarini}}, \bibinfo {author}
  {\bibfnamefont {J.~F.}\ \bibnamefont {Dynes}}, \bibinfo {author}
  {\bibfnamefont {Z.~L.}\ \bibnamefont {Yuan}}, \ and\ \bibinfo {author}
  {\bibfnamefont {A.~J.}\ \bibnamefont {Shields}},\ }\href {\doibase
  10.1038/s41566-019-0377-7} {\bibfield  {journal} {\bibinfo  {journal} {Nature
  Photonics}\ }\textbf {\bibinfo {volume} {13}},\ \bibinfo {pages} {334}
  (\bibinfo {year} {2019})}\BibitemShut {NoStop}%
\bibitem [{\citenamefont {Wang}\ \emph {et~al.}(2019)\citenamefont {Wang},
  \citenamefont {He}, \citenamefont {Yin}, \citenamefont {Lu}, \citenamefont
  {Cui}, \citenamefont {Chen}, \citenamefont {Zhou}, \citenamefont {Guo},\ and\
  \citenamefont {Han}}]{wangBeatingFundamentalRateDistance2019}%
  \BibitemOpen
  \bibfield  {author} {\bibinfo {author} {\bibfnamefont {S.}~\bibnamefont
  {Wang}}, \bibinfo {author} {\bibfnamefont {D.-Y.}\ \bibnamefont {He}},
  \bibinfo {author} {\bibfnamefont {Z.-Q.}\ \bibnamefont {Yin}}, \bibinfo
  {author} {\bibfnamefont {F.-Y.}\ \bibnamefont {Lu}}, \bibinfo {author}
  {\bibfnamefont {C.-H.}\ \bibnamefont {Cui}}, \bibinfo {author} {\bibfnamefont
  {W.}~\bibnamefont {Chen}}, \bibinfo {author} {\bibfnamefont {Z.}~\bibnamefont
  {Zhou}}, \bibinfo {author} {\bibfnamefont {G.-C.}\ \bibnamefont {Guo}}, \
  and\ \bibinfo {author} {\bibfnamefont {Z.-F.}\ \bibnamefont {Han}},\ }\href
  {\doibase 10.1103/PhysRevX.9.021046} {\bibfield  {journal} {\bibinfo
  {journal} {Physical Review X}\ }\textbf {\bibinfo {volume} {9}},\ \bibinfo
  {pages} {021046} (\bibinfo {year} {2019})}\BibitemShut {NoStop}%
\bibitem [{\citenamefont {Zhong}\ \emph {et~al.}(2019)\citenamefont {Zhong},
  \citenamefont {Hu}, \citenamefont {Curty}, \citenamefont {Qian},\ and\
  \citenamefont {Lo}}]{zhongProofofPrincipleExperimentalDemonstration2019}%
  \BibitemOpen
  \bibfield  {author} {\bibinfo {author} {\bibfnamefont {X.}~\bibnamefont
  {Zhong}}, \bibinfo {author} {\bibfnamefont {J.}~\bibnamefont {Hu}}, \bibinfo
  {author} {\bibfnamefont {M.}~\bibnamefont {Curty}}, \bibinfo {author}
  {\bibfnamefont {L.}~\bibnamefont {Qian}}, \ and\ \bibinfo {author}
  {\bibfnamefont {H.-K.}\ \bibnamefont {Lo}},\ }\href {\doibase
  10.1103/PhysRevLett.123.100506} {\bibfield  {journal} {\bibinfo  {journal}
  {Physical Review Letters}\ }\textbf {\bibinfo {volume} {123}},\ \bibinfo
  {pages} {100506} (\bibinfo {year} {2019})}\BibitemShut {NoStop}%
\bibitem [{\citenamefont {Liu}\ \emph {et~al.}(2019)\citenamefont {Liu},
  \citenamefont {Yu}, \citenamefont {Zhang}, \citenamefont {Guan},
  \citenamefont {Chen}, \citenamefont {Zhang}, \citenamefont {Hu},
  \citenamefont {Li}, \citenamefont {Jiang}, \citenamefont {Lin}, \citenamefont
  {Chen}, \citenamefont {You}, \citenamefont {Wang}, \citenamefont {Wang},
  \citenamefont {Zhang},\ and\ \citenamefont
  {Pan}}]{liuExperimentalTwinFieldQuantum2019}%
  \BibitemOpen
  \bibfield  {author} {\bibinfo {author} {\bibfnamefont {Y.}~\bibnamefont
  {Liu}}, \bibinfo {author} {\bibfnamefont {Z.-W.}\ \bibnamefont {Yu}},
  \bibinfo {author} {\bibfnamefont {W.}~\bibnamefont {Zhang}}, \bibinfo
  {author} {\bibfnamefont {J.-Y.}\ \bibnamefont {Guan}}, \bibinfo {author}
  {\bibfnamefont {J.-P.}\ \bibnamefont {Chen}}, \bibinfo {author}
  {\bibfnamefont {C.}~\bibnamefont {Zhang}}, \bibinfo {author} {\bibfnamefont
  {X.-L.}\ \bibnamefont {Hu}}, \bibinfo {author} {\bibfnamefont
  {H.}~\bibnamefont {Li}}, \bibinfo {author} {\bibfnamefont {C.}~\bibnamefont
  {Jiang}}, \bibinfo {author} {\bibfnamefont {J.}~\bibnamefont {Lin}}, \bibinfo
  {author} {\bibfnamefont {T.-Y.}\ \bibnamefont {Chen}}, \bibinfo {author}
  {\bibfnamefont {L.}~\bibnamefont {You}}, \bibinfo {author} {\bibfnamefont
  {Z.}~\bibnamefont {Wang}}, \bibinfo {author} {\bibfnamefont {X.-B.}\
  \bibnamefont {Wang}}, \bibinfo {author} {\bibfnamefont {Q.}~\bibnamefont
  {Zhang}}, \ and\ \bibinfo {author} {\bibfnamefont {J.-W.}\ \bibnamefont
  {Pan}},\ }\href {\doibase 10.1103/PhysRevLett.123.100505} {\bibfield
  {journal} {\bibinfo  {journal} {Physical Review Letters}\ }\textbf {\bibinfo
  {volume} {123}},\ \bibinfo {pages} {100505} (\bibinfo {year}
  {2019})}\BibitemShut {NoStop}%
\bibitem [{Note1()}]{Note1}%
  \BibitemOpen
  \bibinfo {note} {In Twin-Field QKD, the logical basis and the photon-number
  or Fock basis coincide.}\BibitemShut {Stop}%
\bibitem [{\citenamefont {Press}\ \emph {et~al.}(1993)\citenamefont {Press},
  \citenamefont {Teukolsky}, \citenamefont {Vetterling},\ and\ \citenamefont
  {Flannery}}]{numerical_recipes}%
  \BibitemOpen
  \bibfield  {author} {\bibinfo {author} {\bibfnamefont {W.~H.}\ \bibnamefont
  {Press}}, \bibinfo {author} {\bibfnamefont {S.~A.}\ \bibnamefont
  {Teukolsky}}, \bibinfo {author} {\bibfnamefont {W.~T.}\ \bibnamefont
  {Vetterling}}, \ and\ \bibinfo {author} {\bibfnamefont {B.~P.}\ \bibnamefont
  {Flannery}},\ }\href@noop {} {\emph {\bibinfo {title} {Numerical Recipes in
  FORTRAN; The Art of Scientific Computing}}},\ \bibinfo {edition} {2nd}\ ed.\
  (\bibinfo  {publisher} {Cambridge University Press},\ \bibinfo {address} {New
  York, NY, USA},\ \bibinfo {year} {1993})\BibitemShut {NoStop}%
\bibitem [{\citenamefont {Brent}(2013)}]{brent2013algorithms}%
  \BibitemOpen
  \bibfield  {author} {\bibinfo {author} {\bibfnamefont {R.}~\bibnamefont
  {Brent}},\ }\href {https://books.google.com/books?id=AITCAgAAQBAJ} {\emph
  {\bibinfo {title} {Algorithms for Minimization Without Derivatives}}},\ Dover
  Books on Mathematics\ (\bibinfo  {publisher} {Dover Publications},\ \bibinfo
  {year} {2013})\BibitemShut {NoStop}%
\bibitem [{\citenamefont {Renner}(2007)}]{Renner2007}%
  \BibitemOpen
  \bibfield  {author} {\bibinfo {author} {\bibfnamefont {R.}~\bibnamefont
  {Renner}},\ }\href {\doibase 10.1038/nphys684} {\bibfield  {journal}
  {\bibinfo  {journal} {Nature Physics}\ }\textbf {\bibinfo {volume} {3}},\
  \bibinfo {pages} {645} (\bibinfo {year} {2007})}\BibitemShut {NoStop}%
\bibitem [{\citenamefont {Christandl}\ \emph {et~al.}(2009)\citenamefont
  {Christandl}, \citenamefont {K{\"o}nig},\ and\ \citenamefont
  {Renner}}]{Christandl2009}%
  \BibitemOpen
  \bibfield  {author} {\bibinfo {author} {\bibfnamefont {M.}~\bibnamefont
  {Christandl}}, \bibinfo {author} {\bibfnamefont {R.}~\bibnamefont
  {K{\"o}nig}}, \ and\ \bibinfo {author} {\bibfnamefont {R.}~\bibnamefont
  {Renner}},\ }\href {\doibase 10.1103/PhysRevLett.102.020504} {\bibfield
  {journal} {\bibinfo  {journal} {Physical Review Letters}\ }\textbf {\bibinfo
  {volume} {102}},\ \bibinfo {pages} {020504} (\bibinfo {year}
  {2009})}\BibitemShut {NoStop}%
\bibitem [{\citenamefont {Wang}(2005)}]{Wang2005}%
  \BibitemOpen
  \bibfield  {author} {\bibinfo {author} {\bibfnamefont {X.-B.}\ \bibnamefont
  {Wang}},\ }\href {\doibase 10.1103/PhysRevLett.94.230503} {\bibfield
  {journal} {\bibinfo  {journal} {Physical Review Letters}\ }\textbf {\bibinfo
  {volume} {94}},\ \bibinfo {pages} {230503} (\bibinfo {year}
  {2005})}\BibitemShut {NoStop}%
\bibitem [{\citenamefont {Lo}\ \emph {et~al.}(2005)\citenamefont {Lo},
  \citenamefont {Ma},\ and\ \citenamefont {Chen}}]{Lo2005}%
  \BibitemOpen
  \bibfield  {author} {\bibinfo {author} {\bibfnamefont {H.-K.}\ \bibnamefont
  {Lo}}, \bibinfo {author} {\bibfnamefont {X.}~\bibnamefont {Ma}}, \ and\
  \bibinfo {author} {\bibfnamefont {K.}~\bibnamefont {Chen}},\ }\href {\doibase
  10.1103/PhysRevLett.94.230504} {\bibfield  {journal} {\bibinfo  {journal}
  {Physical Review Letters}\ }\textbf {\bibinfo {volume} {94}},\ \bibinfo
  {pages} {230504} (\bibinfo {year} {2005})}\BibitemShut {NoStop}%
\bibitem [{\citenamefont {Beaudry}\ \emph {et~al.}(2008)\citenamefont
  {Beaudry}, \citenamefont {Moroder},\ and\ \citenamefont
  {L{\"u}tkenhaus}}]{Beaudry2009}%
  \BibitemOpen
  \bibfield  {author} {\bibinfo {author} {\bibfnamefont {N.~J.}\ \bibnamefont
  {Beaudry}}, \bibinfo {author} {\bibfnamefont {T.}~\bibnamefont {Moroder}}, \
  and\ \bibinfo {author} {\bibfnamefont {N.}~\bibnamefont {L{\"u}tkenhaus}},\
  }\href {\doibase 10.1103/PhysRevLett.101.093601} {\bibfield  {journal}
  {\bibinfo  {journal} {Physical Review Letters}\ }\textbf {\bibinfo {volume}
  {101}},\ \bibinfo {pages} {093601} (\bibinfo {year} {2008})}\BibitemShut
  {NoStop}%
\bibitem [{\citenamefont {Tsurumaru}\ and\ \citenamefont
  {Tamaki}(2008)}]{Tsurumaru2008}%
  \BibitemOpen
  \bibfield  {author} {\bibinfo {author} {\bibfnamefont {T.}~\bibnamefont
  {Tsurumaru}}\ and\ \bibinfo {author} {\bibfnamefont {K.}~\bibnamefont
  {Tamaki}},\ }\href {\doibase 10.1103/PhysRevA.78.032302} {\bibfield
  {journal} {\bibinfo  {journal} {Physical Review A}\ }\textbf {\bibinfo
  {volume} {78}},\ \bibinfo {pages} {032302} (\bibinfo {year}
  {2008})}\BibitemShut {NoStop}%
\bibitem [{\citenamefont {Gittsovich}\ \emph {et~al.}(2014)\citenamefont
  {Gittsovich}, \citenamefont {Beaudry}, \citenamefont {Narasimhachar},
  \citenamefont {Alvarez}, \citenamefont {Moroder},\ and\ \citenamefont
  {L{\"u}tkenhaus}}]{Gittsovich2014}%
  \BibitemOpen
  \bibfield  {author} {\bibinfo {author} {\bibfnamefont {O.}~\bibnamefont
  {Gittsovich}}, \bibinfo {author} {\bibfnamefont {N.~J.}\ \bibnamefont
  {Beaudry}}, \bibinfo {author} {\bibfnamefont {V.}~\bibnamefont
  {Narasimhachar}}, \bibinfo {author} {\bibfnamefont {R.~R.}\ \bibnamefont
  {Alvarez}}, \bibinfo {author} {\bibfnamefont {T.}~\bibnamefont {Moroder}}, \
  and\ \bibinfo {author} {\bibfnamefont {N.}~\bibnamefont {L{\"u}tkenhaus}},\
  }\href {\doibase 10.1103/PhysRevA.89.012325} {\bibfield  {journal} {\bibinfo
  {journal} {Physical Review A}\ }\textbf {\bibinfo {volume} {89}},\ \bibinfo
  {pages} {012325} (\bibinfo {year} {2014})}\BibitemShut {NoStop}%
\bibitem [{\citenamefont {George}\ and\ \citenamefont
  {L{\"u}tkenhaus}(2019)}]{george2019qcrypt}%
  \BibitemOpen
  \bibfield  {author} {\bibinfo {author} {\bibfnamefont {I.}~\bibnamefont
  {George}}\ and\ \bibinfo {author} {\bibfnamefont {N.}~\bibnamefont
  {L{\"u}tkenhaus}},\ }in\ \href@noop {} {\emph {\bibinfo {booktitle} {9th
  International Conference on Quantum Cryptography}}}\ (\bibinfo {year}
  {2019})\BibitemShut {NoStop}%
\bibitem [{\citenamefont {Lin}\ \emph {et~al.}(2019)\citenamefont {Lin},
  \citenamefont {Upadhyaya},\ and\ \citenamefont
  {L{\"u}tkenhaus}}]{Lin:2019aa}%
  \BibitemOpen
  \bibfield  {author} {\bibinfo {author} {\bibfnamefont {J.}~\bibnamefont
  {Lin}}, \bibinfo {author} {\bibfnamefont {T.}~\bibnamefont {Upadhyaya}}, \
  and\ \bibinfo {author} {\bibfnamefont {N.}~\bibnamefont {L{\"u}tkenhaus}},\
  }\href@noop {} {\enquote {\bibinfo {title} {Asymptotic security analysis of
  discrete-modulated continuous-variable quantum key distribution},}\ }
  (\bibinfo {year} {2019}),\ \Eprint {http://arxiv.org/abs/arXiv:1905.10896}
  {arXiv:1905.10896} \BibitemShut {NoStop}%
\bibitem [{\citenamefont {Govia}\ \emph {et~al.}(2019)\citenamefont {Govia},
  \citenamefont {Bunandar}, \citenamefont {Lin}, \citenamefont {Englund},
  \citenamefont {Lütkenhaus},\ and\ \citenamefont
  {Krovi}}]{govia2019clifford}%
  \BibitemOpen
  \bibfield  {author} {\bibinfo {author} {\bibfnamefont {L.~C.~G.}\
  \bibnamefont {Govia}}, \bibinfo {author} {\bibfnamefont {D.}~\bibnamefont
  {Bunandar}}, \bibinfo {author} {\bibfnamefont {J.}~\bibnamefont {Lin}},
  \bibinfo {author} {\bibfnamefont {D.}~\bibnamefont {Englund}}, \bibinfo
  {author} {\bibfnamefont {N.}~\bibnamefont {Lütkenhaus}}, \ and\ \bibinfo
  {author} {\bibfnamefont {H.}~\bibnamefont {Krovi}},\ }\href@noop {} {\enquote
  {\bibinfo {title} {Clifford group restricted eavesdroppers in quantum key
  distribution},}\ } (\bibinfo {year} {2019}),\ \Eprint
  {http://arxiv.org/abs/1910.10564} {arXiv:1910.10564 [quant-ph]} \BibitemShut
  {NoStop}%
\bibitem [{\citenamefont {Boyd}\ and\ \citenamefont
  {Vandenberghe}(2004)}]{Boyd:2004}%
  \BibitemOpen
  \bibfield  {author} {\bibinfo {author} {\bibfnamefont {S.}~\bibnamefont
  {Boyd}}\ and\ \bibinfo {author} {\bibfnamefont {L.}~\bibnamefont
  {Vandenberghe}},\ }\href@noop {} {\emph {\bibinfo {title} {Convex
  Optimization}}}\ (\bibinfo  {publisher} {Cambridge University Press},\
  \bibinfo {address} {New York, NY, USA},\ \bibinfo {year} {2004})\BibitemShut
  {NoStop}%
\bibitem [{\citenamefont {Watrous}(2011)}]{Watrous2011a}%
  \BibitemOpen
  \bibfield  {author} {\bibinfo {author} {\bibfnamefont {J.}~\bibnamefont
  {Watrous}},\ }\href@noop {} {\bibfield  {journal} {\bibinfo  {journal} {John
  Watrous Lecture Notes}\ } (\bibinfo {year} {2011})}\BibitemShut {NoStop}%
\bibitem [{Note2()}]{Note2}%
  \BibitemOpen
  \bibinfo {note} {We have changed the labels $a$ and $b$ from the definition
  in the Results section into tuples $(a_b, a_v)$ and $(b_b, b_v)$, where the
  subscript $b$ denotes the basis (or any other possible announcement) of the
  measurement and the subscript $v$ denotes the classical outcome of the
  measurement.}\BibitemShut {Stop}%
\bibitem [{Note3()}]{Note3}%
  \BibitemOpen
  \bibinfo {note} {$p_{\protect \text {pass}}$ is not included in the
  optimization as it can be immediately measured from the
  experiment.}\BibitemShut {Stop}%
\bibitem [{Note4()}]{Note4}%
  \BibitemOpen
  \bibinfo {note} {More accurately, a density matrix $\rho $ is a positive
  semidefinite operator. We can consider the density matrix after a
  depolarizing channel with small depolarizing error/probability $\epsilon \geq
  0$, then the resulting density operator $\rho \succ 0$. Calculating the key
  rate using this density matrix incurs some errors that can be quantified,
  see~\protect \citep {Winick2017} for the complete details.}\BibitemShut
  {Stop}%
\bibitem [{\citenamefont {Fung}\ \emph {et~al.}(2009)\citenamefont {Fung},
  \citenamefont {Tamaki}, \citenamefont {Qi}, \citenamefont {Lo},\ and\
  \citenamefont {Ma}}]{Fung2009}%
  \BibitemOpen
  \bibfield  {author} {\bibinfo {author} {\bibfnamefont {C.~C.-H.~F.}\
  \bibnamefont {Fung}}, \bibinfo {author} {\bibfnamefont {K.}~\bibnamefont
  {Tamaki}}, \bibinfo {author} {\bibfnamefont {B.}~\bibnamefont {Qi}}, \bibinfo
  {author} {\bibfnamefont {H.-K.~H.}\ \bibnamefont {Lo}}, \ and\ \bibinfo
  {author} {\bibfnamefont {X.}~\bibnamefont {Ma}},\ }\href
  {http://dl.acm.org/citation.cfm?id=2021264} {\bibfield  {journal} {\bibinfo
  {journal} {Quantum Inf.\textasciitilde{}Comput.}\ }\textbf {\bibinfo {volume}
  {9}},\ \bibinfo {pages} {0131} (\bibinfo {year} {2009})}\BibitemShut
  {NoStop}%
\bibitem [{\citenamefont {Jain}\ \emph {et~al.}(2014)\citenamefont {Jain},
  \citenamefont {Anisimova}, \citenamefont {Khan}, \citenamefont {Makarov},
  \citenamefont {Marquardt},\ and\ \citenamefont {Leuchs}}]{Jain2014}%
  \BibitemOpen
  \bibfield  {author} {\bibinfo {author} {\bibfnamefont {N.}~\bibnamefont
  {Jain}}, \bibinfo {author} {\bibfnamefont {E.}~\bibnamefont {Anisimova}},
  \bibinfo {author} {\bibfnamefont {I.}~\bibnamefont {Khan}}, \bibinfo {author}
  {\bibfnamefont {V.}~\bibnamefont {Makarov}}, \bibinfo {author} {\bibfnamefont
  {C.}~\bibnamefont {Marquardt}}, \ and\ \bibinfo {author} {\bibfnamefont
  {G.}~\bibnamefont {Leuchs}},\ }\href {\doibase
  10.1088/1367-2630/16/12/123030} {\bibfield  {journal} {\bibinfo  {journal}
  {New Journal of Physics}\ }\textbf {\bibinfo {volume} {16}},\ \bibinfo
  {pages} {123030} (\bibinfo {year} {2014})}\BibitemShut {NoStop}%
\bibitem [{\citenamefont {Gisin}\ \emph {et~al.}(2006)\citenamefont {Gisin},
  \citenamefont {Fasel}, \citenamefont {Kraus}, \citenamefont {Zbinden},\ and\
  \citenamefont {Ribordy}}]{Gisin2006}%
  \BibitemOpen
  \bibfield  {author} {\bibinfo {author} {\bibfnamefont {N.}~\bibnamefont
  {Gisin}}, \bibinfo {author} {\bibfnamefont {S.}~\bibnamefont {Fasel}},
  \bibinfo {author} {\bibfnamefont {B.}~\bibnamefont {Kraus}}, \bibinfo
  {author} {\bibfnamefont {H.}~\bibnamefont {Zbinden}}, \ and\ \bibinfo
  {author} {\bibfnamefont {G.}~\bibnamefont {Ribordy}},\ }\href {\doibase
  10.1103/PhysRevA.73.022320} {\bibfield  {journal} {\bibinfo  {journal}
  {Physical Review A}\ }\textbf {\bibinfo {volume} {73}},\ \bibinfo {pages}
  {022320} (\bibinfo {year} {2006})}\BibitemShut {NoStop}%
\bibitem [{\citenamefont {Zhao}\ \emph {et~al.}(2007)\citenamefont {Zhao},
  \citenamefont {Qi},\ and\ \citenamefont {Lo}}]{Zhao2007}%
  \BibitemOpen
  \bibfield  {author} {\bibinfo {author} {\bibfnamefont {Y.}~\bibnamefont
  {Zhao}}, \bibinfo {author} {\bibfnamefont {B.}~\bibnamefont {Qi}}, \ and\
  \bibinfo {author} {\bibfnamefont {H.-K.}\ \bibnamefont {Lo}},\ }\href
  {\doibase 10.1063/1.2432296} {\bibfield  {journal} {\bibinfo  {journal}
  {Applied Physics Letters}\ }\textbf {\bibinfo {volume} {90}},\ \bibinfo
  {pages} {044106} (\bibinfo {year} {2007})}\BibitemShut {NoStop}%
\bibitem [{\citenamefont {Muller}\ \emph {et~al.}(1997)\citenamefont {Muller},
  \citenamefont {Herzog}, \citenamefont {Huttner}, \citenamefont {Tittel},
  \citenamefont {Zbinden},\ and\ \citenamefont {Gisin}}]{Muller1997}%
  \BibitemOpen
  \bibfield  {author} {\bibinfo {author} {\bibfnamefont {A.}~\bibnamefont
  {Muller}}, \bibinfo {author} {\bibfnamefont {T.}~\bibnamefont {Herzog}},
  \bibinfo {author} {\bibfnamefont {B.}~\bibnamefont {Huttner}}, \bibinfo
  {author} {\bibfnamefont {W.}~\bibnamefont {Tittel}}, \bibinfo {author}
  {\bibfnamefont {H.}~\bibnamefont {Zbinden}}, \ and\ \bibinfo {author}
  {\bibfnamefont {N.}~\bibnamefont {Gisin}},\ }\href {\doibase
  10.1063/1.118224} {\bibfield  {journal} {\bibinfo  {journal} {Applied Physics
  Letters}\ }\textbf {\bibinfo {volume} {70}},\ \bibinfo {pages} {793}
  (\bibinfo {year} {1997})}\BibitemShut {NoStop}%
\bibitem [{\citenamefont {Stucki}\ \emph {et~al.}(2002)\citenamefont {Stucki},
  \citenamefont {Gisin}, \citenamefont {Guinnard}, \citenamefont {Ribordy},\
  and\ \citenamefont {Zbinden}}]{Stucki2002}%
  \BibitemOpen
  \bibfield  {author} {\bibinfo {author} {\bibfnamefont {D.}~\bibnamefont
  {Stucki}}, \bibinfo {author} {\bibfnamefont {N.}~\bibnamefont {Gisin}},
  \bibinfo {author} {\bibfnamefont {O.}~\bibnamefont {Guinnard}}, \bibinfo
  {author} {\bibfnamefont {G.}~\bibnamefont {Ribordy}}, \ and\ \bibinfo
  {author} {\bibfnamefont {H.}~\bibnamefont {Zbinden}},\ }\href {\doibase
  10.1088/1367-2630/4/1/341} {\bibfield  {journal} {\bibinfo  {journal} {New
  Journal of Physics}\ }\textbf {\bibinfo {volume} {4}},\ \bibinfo {pages}
  {341} (\bibinfo {year} {2002})}\BibitemShut {NoStop}%
\bibitem [{\citenamefont {Lucamarini}\ \emph {et~al.}(2015)\citenamefont
  {Lucamarini}, \citenamefont {Choi}, \citenamefont {Ward}, \citenamefont
  {Dynes}, \citenamefont {Yuan},\ and\ \citenamefont
  {Shields}}]{Lucamarini2015b}%
  \BibitemOpen
  \bibfield  {author} {\bibinfo {author} {\bibfnamefont {M.}~\bibnamefont
  {Lucamarini}}, \bibinfo {author} {\bibfnamefont {I.}~\bibnamefont {Choi}},
  \bibinfo {author} {\bibfnamefont {M.~B.}\ \bibnamefont {Ward}}, \bibinfo
  {author} {\bibfnamefont {J.~F.}\ \bibnamefont {Dynes}}, \bibinfo {author}
  {\bibfnamefont {Z.~L.}\ \bibnamefont {Yuan}}, \ and\ \bibinfo {author}
  {\bibfnamefont {A.~J.}\ \bibnamefont {Shields}},\ }\href {\doibase
  10.1103/PhysRevX.5.031030} {\bibfield  {journal} {\bibinfo  {journal}
  {Physical Review X}\ }\textbf {\bibinfo {volume} {5}},\ \bibinfo {pages}
  {031030} (\bibinfo {year} {2015})}\BibitemShut {NoStop}%
\end{thebibliography}%

\section*{Acknowledgments}
\label{sec:acknowledgments}
We would like to thank Norbert L\"{u}tkenhaus (Univ. of Waterloo) and Jie Lin (Univ. of Waterloo) for their helpful feedback and their suggestions to include the analysis against coherent attacks. We acknowledge the support from the Office of Naval Research CONQUEST program N00014-16-C-2069. We thank members of the CONQUEST Team for their helpful inputs: Saikat Guha (Univ. of Arizona), Jeffrey Shapiro (MIT), Mark Wilde (Louisiana State Univ.), Franco Wong (MIT). D.B. and D.E. also acknowledge the support from the Air Force Office of Scientific Research program FA9550-16-1-0391, supervised by Gernot Pomrenke. 

\section*{Author Contributions}
D.B. and D.E. contributed to the initial conception of the ideas. D.B. provided the initial security proofs, and L.G. and H.K. assisted with the simplifications and extensions to the proofs. D.B. and L.G. wrote the source code, and D.B. implemented the code for the different examples. All authors contributed to writing the manuscript.

\newpage
\appendix

\section{Convex optimization}
	\label{app:convex_optimization}

	Problems in quantum information can often be formulated as an optimization problem. In particular, the secret key rate problem can be expressed as a convex optimization problem, specifically a semidefinite program. There is in general no analytical formula for the solution of convex optimization problems, but there are efficient methods for solving them, such as the interior point methods~\citep{Boyd:2004}.

	A convex optimization problem is an optimization problem of the form:
	\begin{primal}
	\begin{equation}
	\label{eq:convex_primal}
	\begin{aligned}
		& \text{minimize} 
		& & f_0 (x) \\
		& \text{subject to}
		& & f_i(x) \leq 0, \: i = 1, \dots, m \\
		&&& a_i^T x = b_i, \: i = 1, \dots, p,
	\end{aligned}
	\end{equation}
	\end{primal}
	\noindent where $f_0, \dots, f_m$ are convex functions, i.e. $f_i (p x_1 + (1-p)x_2) \leq p f_i(x_1) + (1-p) f_i(x_2)$ for any $x_1, x_2$ and $0 \leq p \leq 1$. Let us call the set of $x$ values that satisfies the constraints as the feasible set, denoted as $\mathcal{P}$. We refer to this problem as the primal problem. 

	By rewriting the equality constraint as $h_i(x) = a_i^T x -b_i$ and require $h_i(x) = 0$, we can define the Lagrangian associated with Prob.~\eqref{eq:convex_primal} as
	\begin{equation}
		\mathcal{L}(x, \lambda, \nu) = f_0 (x) + \sum_{i=1}^m \lambda_i f_i(x) + \sum_{i=1}^p \nu_i h_i(x),
	\end{equation}
	where $\lambda_i$ and $\nu_i$ are Lagrange multipliers associated with the problem.

	For each primal problem, there exists an associated dual problem:
	\begin{dual}
	\begin{equation}
	\label{eq:convex_dual}
	\begin{aligned}
		& \text{maximize} 
		& & g(\lambda, \nu) \\
		& \text{subject to}
		& & \nu \geq 0,
	\end{aligned}
	\end{equation}
	\end{dual}
	\noindent where $g(\lambda, \nu) \equiv \inf_{x \in \mathcal{P}} \mathcal{L}(x, \nu, \lambda)$. The significance of this dual problem is as follows. The optimal value of the dual problem (Prob.~\eqref{eq:convex_dual}) $d^{*}$ is, by definition, the best lower bound on the optimal value of the primal problem (Prob.~\eqref{eq:convex_primal}) $p^{*}$. In particular, we have an important relation $d^{*} \leq p^{*}$ called weak duality, which always holds even when the problem is not convex.

	If the gap between $d^{*}$ and $p^{*}$ is 0, then we say that strong duality holds. For convex optimization problems, the strong duality holds if Slater's condition is satisfied: if there exists a point $x \in \mathcal{P}$ such that all the inequality constraints $f_i(x)$ is strictly less than zero and all the equality constraints are satisfied~\citep{Boyd:2004}.

	An important class of convex optimization problems that are often encountered in quantum information processing is the semidefinite program (SDP). Here we define a semidefinite program in the standard form proposed by~\cite{Watrous2009} that is more directly applicable to working with quantum density matrices.

	Let us first define several mathematical terms to help our discussion. Given a complex vector space $\mathcal{X} = \mathbb{C}^n$, we call an $n \times n$ linear operator $X$ Hermitian if $X = X^{\dag}$; let us denote the set of such operators with $\herm(\mathcal{X})$. An operator $X$ is called positive semidefinite if it is Hermitian and all of its eigenvalues are nonnegative. We use the notation $X \succeq 0$ to indicate that $X$ is positive semidefinite. More generally, $Y \succeq X$ indicates that $Y-X \succeq 0$ for Hermitian operators $X$ and $Y$. We also use the notation $X \succ 0$ to indicate that the operator $X$ is positive definite: Hermitian and all its eigenvalues are strictly positive.

	An SDP can be defined using a few parameters:
	\begin{itemize}
		\item $\Phi$ which is a Hermiticity-preserving linear map, and
		\item $A$ and $B$ which are Hermitian operators.
	\end{itemize}

	We define the primal of an SDP problem to be:
	\begin{primal}
	\begin{equation}
	\label{eq:sdp_primal}
	\begin{aligned}
		& \text{minimize}
		& & \braket{A, X} \\
		& \text{subject to}
		& & \Phi(X) \succeq B, \\
		&&& X \succeq 0,
	\end{aligned}
	\end{equation}
	\end{primal}
	\noindent where the inner product $\braket{A, B} \equiv \Tr(A^{\dag} B)$. The Lagrange dual problem to the SDP above is
	\begin{dual}
	\begin{equation}
	\label{eq:sdp_dual}
	\begin{aligned}
		& \text{maximize} 
		& & \braket{B, Y} \\
		& \text{subject to}
		& & \Phi^{\dag}(Y) \preceq A, \\
		&&& Y \succeq 0.
	\end{aligned}
	\end{equation}
	\end{dual}
	\noindent The mapping $\Phi^{\dag}$ is a unique mapping that can be defined from the following equation:
	\begin{equation}
		\braket{Y, \Phi(X)} = \braket{\Phi^{\dag}(Y), X}.
	\end{equation}

	An operator $X \succeq 0$ satisfying $\Phi(X) \succeq B$ is called primal feasible, and an operator $Y \succeq 0$ satisfying $\Phi^{\dag}(Y) \preceq A$ is called dual feasible. We denote the sets of primal and dual feasible operators with $\mathcal{P}$ and $\mathcal{D}$, respectively.

	Weak duality holds for any SDP, that is the optimal value of the primal problem $p^*$ and the optimal value of the dual problem $d^*$ are always related by $p^* \geq d^*$. Strong duality holds if either of the following Slater's conditions are satisfied~\citep{Watrous2011a}:
	\begin{enumerate}
	 	\item If $\mathcal{P}$ is nonempty and there exists an operator $Y \succ 0$ such that $\Phi^{\dag}(Y) \prec A$, then there exists a primal feasible operator $X$ for which $\braket{A, X} = p^*$ and $p^* = d^*$,
	 	\item If $\mathcal{D}$ is nonempty and there exists an operator $X \succ 0$ such that $\Phi(X) \succ B$, then there exists a dual feasible operator $Y$ for which $\braket{B, Y} = d^*$ and $p^* = d^*$.
	 \end{enumerate} 
% section convex_optimization (end)

\section{General framework for postselection}
	\label{app:general_framework_for_postselection}

	Typically in a QKD protocol, Alice and Bob make public announcements during sifting in which they postselect for certain basis choices. During the quantum transmission stage, Alice and Bob measure their respective POVMs $\set{M^{(a_b, a_v)}_A}$ and $\set{M^{(b_b, b_v)}_B}$~\footnote{We have changed the labels $a$ and $b$ from the definition in the Results section into tuples $(a_b, a_v)$ and $(b_b, b_v)$, where the subscript $b$ denotes the basis (or any other possible announcement) of the measurement and the subscript $v$ denotes the classical outcome of the measurement.}. We follow the approach previously established in several papers (see e.g.~\cite{Winick2017}) and introduce extra classical registers $A_b$ and $A_v$ for Alice and $B_b$ and $B_v$ for Bob to store the basis and value information respectively. The idea is that Alice and Bob will keep most of the registers $A_v$ and $B_v$ to themselves (releasing some information for error correction), while they will eventually make public the registers $A_b$ and $B_b$. Alice's measurements and announcements can be described by a quantum channel with Kraus operators
	\begin{equation}
		K_A^{a_b} = \sum_{a_v} \sqrt{M^{(a_b, a_v)}_A} \otimes \ket{a_b}_{A_b} \ket{a_v}_{A_v},
	\end{equation}
	and, similarly, Bob's can be described by another set of Kraus operators
	\begin{equation}
		K_B^{b_b} = \sum_{b_v} \sqrt{M^{(b_b, b_v)}_B} \otimes \ket{b_b}_{B_b} \ket{b_v}_{B_v}.
	\end{equation}
	The quantum state after the announcement can be obtained through a completely positive trace-preserving (CPTP) map $\mathcal{A}$ involving the Kraus operators above, i.e.
	\begin{equation}
	\begin{aligned}
		\rho_{AA_vA_bBB_vB_b}^{\text{ann}} &= \mathcal{A}(\rho_{AB}) \\
		&= \sum_{a_b,b_b} (K_A^{a_b} \otimes K_B^{b_b}) \rho_{AB} (K_A^{a_b} \otimes K_B^{b_b})^{\dag}.
	\end{aligned}
	\end{equation}

	Next, Alice and Bob will postselect/sift to decide which parts of the data they will keep. Let $\mathcal{B}_{\text{keep}}$ be the set of basis measurements they will keep. For example, they may choose to keep only measurements in the same basis. Then, we can define a projector:
	\begin{equation}
	 	\Pi = \sum_{(a_b,b_b) \in \mathcal{B}_{\text{keep}} } \ket{a_b}\bra{a_b}_{A_b} \otimes \ket{b_b}\bra{b_b}_{B_b}.
	\end{equation} 
	The postselected state can then be modeled by using this projector:
	\begin{equation}
		\rho_{AA_vA_bBB_vB_b}^{\text{sift}} = \frac{\Pi \rho_{AA_vA_bBB_vB_b}^{\text{ann}} \Pi}{\ppass},
	\end{equation}
	with $\ppass = \Tr(\Pi \rho_{AA_vA_bBB_vB_b}^{\text{ann}})$ is the probability of passing the postselection filter~\footnote{$\ppass$ is not included in the optimization as it can be immediately measured from the experiment.}. We therefore can define a completely positive trace non-increasing map $\mathcal{S}$ for sifting, such that:
	\begin{equation}
		\mathcal{S}(\rho_{AB}) = \Pi \mathcal{A}(\rho_{AB}) \Pi = \ppass \; \rho_{AA_vA_bBB_vB_b}^{\text{sift}}.
	\end{equation}

	Following the derivation in~\cite{Coles2012}, we define another isometry $V_{Z_A} = \sum_j \ket{j}_{Z_A} \otimes Z_{A_v}^j$ to store the raw key information in the register $Z_A$. Applying this isometry to $\rho^{\text{sift}}$ gives us:
	\begin{equation}
		\tilde{\rho}_{Z_A AA_vA_bBB_vB_b}^{\text{sift}} = V_{Z_A} \rho_{AA_vA_bBB_vB_b}^{\text{sift}} V_{Z_A}^{\dag}.
	\end{equation}

We then take Eve's system to purify the state $\rho^{\text{sift}}$ (and $\tilde{\rho}^{\text{sift}}$) such that she's able to obtain the maximum amount of information from not only $A$ and $B$, but also $A_v$, $A_b$, $B_v$, and $B_b$. The key-rate problems that are solved by using von Neumann entropy are therefore modified from $\min_{\rho_{AB}} H(Z_A|E)_{\rho}$ to $\min_{\rho_{AB}} \left[ p_{\text{pass}} H(Z_A|E)_{\tilde{\rho}^{\text{sift}}} \right]$. Using similar arguments as outlined in the main manuscript, we obtain:
\begin{equation}
\begin{aligned}
	H(Z_A|E )_{\tilde{\rho}^{\text{sift}}} &= H(\tilde{\rho}^{\text{sift}}_{Z_A E}) - H(\rho^{\text{sift}}_{E}) \\
	&= H(\tilde{\rho}^{\text{sift}}_{A A_v A_b B B_v B_b}) - H(\rho^{\text{sift}}_{A A_v A_b B B_v B_b}) \\
	&= \infdiv{\rho^{\text{sift}}_{A A_v A_b B B_v B_b}}{\tilde{\rho}^{\text{sift}}_{A A_v A_b B B_v B_b}} \\
	&= \frac{1}{p_{\text{pass}}} \infdiv*{\mathcal{S}(\rho_{AB})}{\sum_j Z_{A_v}^j \mathcal{S}(\rho_{AB}) Z_{A_v}^j},
\end{aligned}	
\end{equation}
where the last line has been derived using the property that $\infdiv{c \rho}{c \sigma} = c \infdiv{\rho}{\sigma}$ for any constant $c > 0$. Furthermore, when linearizing this key rate problem to obtain a dual solution, we must update the gradient $\nabla f(\rho)^T$ defined in Eq.~\eqref{eq:grad_f_rho} to:
	\begin{equation}
	\label{eq:grad_f_rho_sifted}
	\begin{aligned}
		\nabla f(\rho_{AB})^T &= \mathcal{S}^{\dag} \left( \log_2 \mathcal{S} (\rho_{AB}) \right) \\ 
		& \qquad - \mathcal{S}^{\dag} \left(\log_2 \sum_j Z_{A_v}^j \mathcal{S}(\rho_{AB}) Z_{A_v}^j \right),
	\end{aligned}
	\end{equation}
	where $\mathcal{S}^{\dag}$ is the adjoint map of $\mathcal{S}$ that can be found from the fact that:
	\begin{equation}
		\Tr[\mathcal{S}(\rho) \sigma] = \Tr[\rho \mathcal{S}^{\dag}(\sigma)].
	\end{equation}
	Explicitly, since
	\begin{equation}
		\mathcal{S}(\rho) = \sum_{a_b,b_b} \Pi (K_A^{a_b} \otimes K_B^{b_b}) \rho (K_A^{a_b} \otimes K_B^{b_b})^{\dag} \Pi,
	\end{equation}
	then the adjoint map is
	\begin{equation}
		\mathcal{S}^{\dag}(\rho^{\text{sift}}) = \sum_{a_b,b_b} (K_A^{a_b} \otimes K_B^{b_b})^{\dag} \Pi \rho^{\text{sift}} \Pi (K_A^{a_b} \otimes K_B^{b_b}).
	\end{equation}

Similarly, we modify the key-rate problems that are solved using min-entropy from $\min_{\rho_{AB}} H_{\min}(Z_A|E)_{\rho}$ to $\min_{\rho_{AB}} \left[ p_{\text{pass}} H_{\min}(Z_A|E)_{\tilde{\rho}^{\text{sift}}} \right]$. Using similar arguments to the ones in the main manuscript, we obtain:
\begin{equation}
	\begin{aligned}
	&H_{\min}(Z_A|E)_{\tilde{\rho}^{\text{sift}}} \\
	&= -H_{\max}(Z_A|A A_v A_b B B_v B_b) \\
	&= -\log_2 \max_{\sigma_{AB}} F\left(\tilde{\rho}^{\text{sift}}_{Z_A A A_v A_b B B_v B_b}, \id_{Z_A} \otimes \; \sigma_{A A_v A_b B B_v B_b}\right) \\
	&= -\log_2 \max_{\sigma_{AB}} F\left(\mathcal{S}(\rho_{AB})/p_{\text{pass}}, \sum_j Z_{A_v}^j \sigma_{A A_v A_b B B_v B_b} Z_{A_v}^j\right) \\
	&= -\log_2 \max_{\sigma_{AB}} F\left(\mathcal{S}(\rho_{AB}), \sum_j Z_{A_v}^j \sigma_{A A_v A_b B B_v B_b} Z_{A_v}^j\right) \\
	& \qquad + \log_2 p_{\text{pass}},
	\end{aligned}
\end{equation}
	where the last line is found by noticing that $F(c \sigma, \rho) = F(\sigma, c \rho) = c F(\sigma, \rho)$ for any constant $c > 0$.

	When solving the numerical key rate problems, one can use the fine-grained constraints $\Gamma_i = M^{(a_b,a_v)}_A \otimes M^{(b_b,b_v)}_B$ for all values of $\{(a_b,a_v), (b_b, b_v)\}$. However, with measurement bases being well-defined in this framework, we can find general coarse-grained constraints where Alice and Bob obtain the same or different classical measurement values within each basis they postselect for. In other words, for $(a_b, b_b) \in \mathcal{B}_{\text{keep}}$, we can find such constraints:
	\begin{equation}
		\Gamma_{(a_b, b_b)}^{(=)} = \sum_{a_v = b_v} M^{(a_b,a_v)}_A \otimes M^{(b_b,b_v)}_B,
	\end{equation}
	and
	\begin{equation}
		\Gamma_{(a_b, b_b)}^{(\neq)} = \sum_{a_v \neq b_v} M^{(a_b,a_v)}_A \otimes M^{(b_b,b_v)}_B.
	\end{equation}
	Although generally using coarse-grained constraints leads to lower key rates, the key rates obtained in the more symmetric protocols we consider show no noticeable difference when compared to the key rates obtained using fine-grained constraints. In fact, when using only the coarse-grained constraints, the amount of information that must be communicated classically between Alice and Bob is reduced.

	Notice that this framework for postselection generally increases the size of the computation as it dilates the Hilbert space needed from just $AB$ to include extra registers $A_vA_bB_vB_b$. In particular, the SDP for solving the approximate problem involving the quantum relative entropy can become too large for a typical personal computer to handle. We are, however, able to simplify the postselection procedure for some protocols without needing to introduce many extra registers.

\section{Simplification to the postselection procedure}
\label{app:simplification_postselection}

Whenever postselection is performed---even for the simplest postselected BB84 protocol---direct calculation of the approximate SDP for quantum relative entropy can become a bottleneck (see SDP~\eqref{eq:key_rate_nonasymptotic_appx_primal}). For a density matrix $\rho$ of size $n \times n$, solving the approximate SDP problem at order $(m,k)$ involves solving for a total of $k$ blocks of $2n^2 \times 2n^2$ positive semidefinite matrices and $m$ blocks of $(n^2+1) \times (n^2+1)$ positive semidefinite matrices. For the postselected BB84 protocol, $n = (\dim A) \times (\dim A_v) \times (\dim A_b) \times (\dim B) \times (\dim B_v) \times (\dim B_b) = 64$ which results in an extremely large SDP to solve. We see a slowdown in the SDP for the fidelity function (see SDP~\eqref{eq:keyrate_Hmin_dual_problem}), but the problem is still small enough for our numerical solvers to find a solution within a reasonable amount of time.

We outline simplification steps that allows us to dramatically increase the calculation speed for the examples that we explore in this manuscript.

\subsection{BB84}
\label{sec:bb84}

The Kraus operators related to Alice and Bob's announcements are:
\begin{equation}
\label{eq:kraus_AB}
\begin{aligned}
	K^0_A &= \sqrt{p_Z} \left[ \kb{0}_A \otimes \ket{0}_{A_b} \otimes \ket{0}_{A_v} \right. \\
	& \qquad \left. + \kb{1}_A \otimes \ket{0}_{A_b} \otimes \ket{1}_{A_v} \right] \\
	K^1_A &= \sqrt{p_X} \left[ \kb{+}_A \otimes \ket{1}_{A_b} \otimes \ket{0}_{A_v} \right. \\
	& \qquad \left. + \kb{-}_A \otimes \ket{1}_{A_b} \otimes \ket{1}_{A_v} \right] \\
	K^0_B &= \sqrt{p_Z} \left[ \kb{0}_B \otimes \ket{0}_{B_b} + \kb{1}_B \otimes \ket{0}_{B_b}\right] \\
	&= \sqrt{p_Z} \id_{B} \otimes \ket{0}_{B_b} \\
	K^1_B &= \sqrt{p_X} \left[ \kb{+}_B \otimes \ket{1}_{B_b} + \kb{-}_B \otimes \ket{1}_{B_b}\right] \\
	&= \sqrt{p_X} \id_{B} \otimes \ket{1}_{B_b},
\end{aligned}
\end{equation}
where there is no need to keep track of $B_v$ because the key map is only applied to Alice's value register.

With the postselection operator:
\begin{equation}
	\Pi = \kb{0}_{A_b} \otimes \kb{0}_{B_b} + \kb{1}_{A_b} \otimes \kb{1}_{B_b},
\end{equation}
we can fully define the action of the sifting map:
\begin{equation}
\label{eq:sift_blocked}
\begin{aligned}
	\mathcal{S}(\rho_{AB}) &= \Pi \left[ \sum_{a_b, b_b \in \set{0,1}} (K_A^{a_b} K_B^{b_b}) \rho_{AB} (K_A^{a_b} K_B^{b_b})^{\dag} \right] \Pi \\
	&= \left( K_A^0 K_B^0 \right) \rho_{AB} \left( K_A^0 K_B^0 \right)^{\dag} \\
	& \qquad + \left( K_A^1 K_B^1 \right) \rho_{AB} \left( K_A^1 K_B^1 \right)^{\dag} \\
	&= p_Z^2 \kb{0}_{A_b} \otimes \kb{0}_{B_b} \otimes \rho_{z} \\
	& \qquad + p_X^2 \kb{1}_{A_b} \otimes \kb{1}_{B_b} \otimes \rho_{x}.
\end{aligned}
\end{equation}
Notice that the expression above is block diagonal so we can write:
\begin{equation}
	\mathcal{S}(\rho_{AB}) = p_Z^2 \rho_z \oplus p_X^2 \rho_x.
\end{equation}

\subsubsection{The key rate SDP}
\label{sub:the_key_rate_sdp_bb84}

The goal here is to simplify the key rate problem (in the von Neumann entropy formalism) by separating out the two blocks and thereby proving that
\begin{equation}
\label{eq:bb84_simplified}
\begin{aligned}
	\infdiv*{\mathcal{S}(\rho_{AB})}{\mathcal{Z}_{A_v}^{Z} (\mathcal{S}(\rho_{AB}))} &= p_Z^2 \infdiv{\rho_{AB}}{\mathcal{Z}_{A}^{Z} (\rho_{AB})} \\
	& + p_X^2 \infdiv{\rho_{AB}}{\mathcal{Z}_{A}^{X} (\rho_{AB})},
\end{aligned}
\end{equation}
where we have defined the notation $\mathcal{Z}_{\mathcal{H}}^{\mathcal{B}}$ which is the pinching channel in the $\mathcal{B}$-basis acting on Hilbert space $\mathcal{H}$. In particular:
\begin{equation}
	\mathcal{Z}_{A}^{Z} (\rho_{AB})) = \sum_{j\in\{0,1\}} \kb{j}_A \rho_{AB} \kb{j}_A,
\end{equation}
and
\begin{equation}
	\mathcal{Z}_{A}^{X} (\rho_{AB})) = \sum_{j\in\{+,-\}} \kb{j}_A \rho_{AB} \kb{j}_A.
\end{equation}

First, we state the following useful lemma:
\begin{lemma}
\label{lemma:block_diagonal}
Given $M = A \oplus B$ and $M' = A' \oplus B'$, where $\dim A = \dim A'$ and $\dim B = \dim B'$, we have
\begin{equation}
	\Tr[M \log_2 M'] = \Tr[A \log_2 A'] + \Tr[B \log_2 B'].
\end{equation}
\end{lemma}
\begin{proof}
The proof can be obtained by direct computation. Since we have
\begin{equation}
	M = \begin{pmatrix}
		A & 0\\
		0 & B
	\end{pmatrix}, \:
	M' = \begin{pmatrix}
		A' & 0 \\
		0 & B'
	\end{pmatrix},
\end{equation}
then the term
\begin{equation}
	\log_2 M' = 
	\begin{pmatrix}
		\log_2 A' & 0 \\
		0 & \log_2 B'
	\end{pmatrix}
\end{equation}.

Therefore,
\begin{equation}
	M \log_2 M' = \begin{pmatrix}
		A \log_2 A' & 0 \\
		0 & B \log_2 B'
	\end{pmatrix},
\end{equation}
and taking the trace of this matrix completes the proof.
\end{proof}
% subsection bb84 (end)

Applying Lemma~\ref{lemma:block_diagonal} to $\infdiv*{\mathcal{S}(\rho_{AB})}{\mathcal{Z}_{A_v}^{Z} (\mathcal{S}(\rho_{AB}))}$, with the identification $M = \mathcal{S}(\rho_{AB})$ and $M' = \mathcal{Z}_{A_v}^{Z} (\mathcal{S}(\rho_{AB}))$, gives us
\begin{equation}
\label{eq:bb84_simplified_intermediate}
\begin{aligned}
	\infdiv*{\mathcal{S}(\rho_{AB})}{\mathcal{Z}_{A_v}^{Z} (\mathcal{S}(\rho_{AB}))} &= p_Z^2  \infdiv*{\rho_z}{\mathcal{Z}_{A_v}^{Z} (\rho_z)} \\
	& + p_X^2 \infdiv*{\rho_x}{\mathcal{Z}_{A_v}^{Z} (\rho_x)}.
\end{aligned}
\end{equation}

To obtain the simplified Eq.~\eqref{eq:bb84_simplified}, we show the following:
\begin{equation}
\begin{aligned}
	\infdiv*{\rho_z}{\mathcal{Z}_{A_v}^{Z} (\rho_z)} = \infdiv*{\rho_{AB}}{\mathcal{Z}_{A}^{Z} (\rho_{AB})}, \\
	\infdiv*{\rho_x}{\mathcal{Z}_{A_v}^{Z} (\rho_x)} = \infdiv*{\rho_{AB}}{\mathcal{Z}_{A}^{X} (\rho_{AB})}.
\end{aligned}
\end{equation}
The identification is straightforward. Let us write down $\rho_z$ in the basis of $A_v$ and $A$:
\begin{equation}
	\rho_z = 
	\begin{blockarray}{ccc}
 & \bra{0}_{A_v} \bra{0}_A & \bra{1}_{A_v} \bra{1}_A \\
\begin{block}{c(cc)}
  \ket{0}_{A_v} \ket{0}_A & \braket{0|\rho_{AB}|0}_A & \braket{0|\rho_{AB}|1}_A \\
  \ket{1}_{A_v} \ket{1}_A &\braket{1|\rho_{AB}|0}_A  & \braket{1|\rho_{AB}|1}_A \\
\end{block}
\end{blockarray}\;\; ,
\end{equation}
which is equivalent to $\rho_{AB}$ in the $Z$-basis:
\begin{equation}
	\rho_{AB} = 
	\begin{blockarray}{ccc}
 & \bra{0}_A & \bra{1}_A \\
\begin{block}{c(cc)}
  \ket{0}_A & \braket{0|\rho_{AB}|0}_A & \braket{0|\rho_{AB}|1}_A \\
  \ket{1}_A &\braket{1|\rho_{AB}|0}_A  & \braket{1|\rho_{AB}|1}_A \\
\end{block}
\end{blockarray}\;\; .
\end{equation}
Now, let us write $\mathcal{Z}_{A_v}^{Z} (\rho_z)$ in the basis of $A_v$ and $A$:
\begin{equation}
	\mathcal{Z}_{A_v}^{Z} (\rho_z) = 
	\begin{blockarray}{ccc}
 & \bra{0}_{A_v} \bra{0}_A & \bra{1}_{A_v} \bra{1}_A \\
\begin{block}{c(cc)}
  \ket{0}_{A_v} \ket{0}_A & \braket{0|\rho_{AB}|0}_A & 0 \\
  \ket{1}_{A_v} \ket{1}_A & 0 & \braket{1|\rho_{AB}|1}_A \\
\end{block}
\end{blockarray}\;\; ,
\end{equation}
which is equivalent to $\mathcal{Z}_{A}^{Z} (\rho_{AB})$ in the $Z$-basis:
\begin{equation}
	\mathcal{Z}_{A}^{Z} (\rho_{AB}) = 
	\begin{blockarray}{ccc}
 & \bra{0}_A & \bra{1}_A \\
\begin{block}{c(cc)}
  \ket{0}_A & \braket{0|\rho_{AB}|0}_A & 0 \\
  \ket{1}_A & 0 & \braket{1|\rho_{AB}|1}_A \\
\end{block}
\end{blockarray}\;\; ,
\end{equation}
that shows $\infdiv*{\rho_z}{\mathcal{Z}_{A_v}^{Z} (\rho_z)} = \infdiv*{\rho_{AB}}{\mathcal{Z}_{A}^{Z} (\rho_{AB})}$.

Similarly, $\rho_x$ in the $Z$-basis of $A_v$ and the $X$-basis of $A$ is:
\begin{equation}
	\rho_x = 
	\begin{blockarray}{ccc}
 & \bra{0}_{A_v} \bra{+}_A & \bra{1}_{A_v} \bra{-}_A \\
\begin{block}{c(cc)}
  \ket{0}_{A_v} \ket{+}_A & \braket{+|\rho_{AB}|+}_A & \braket{+|\rho_{AB}|-}_A \\
  \ket{1}_{A_v} \ket{-}_A &\braket{-|\rho_{AB}|+}_A  & \braket{-|\rho_{AB}|-}_A \\
\end{block}
\end{blockarray}\;\; ,
\end{equation}
which is equivalent to $\rho_{AB}$ in the $X$-basis:
\begin{equation}
	\rho_{AB} = 
	\begin{blockarray}{ccc}
 & \bra{+}_A & \bra{-}_A \\
\begin{block}{c(cc)}
  \ket{+}_A & \braket{+|\rho_{AB}|+}_A & \braket{+|\rho_{AB}|-}_A \\
  \ket{-}_A &\braket{-|\rho_{AB}|+}_A  & \braket{-|\rho_{AB}|-}_A \\
\end{block}
\end{blockarray}\;\; .
\end{equation}
Furthermore, $\mathcal{Z}_{A_v}^{Z} (\rho_x)$ in the $Z$-basis of $A_v$ and the $X$-basis of $A$ is:
\begin{equation}
	\mathcal{Z}_{A_v}^{Z} (\rho_x) = 
	\begin{blockarray}{ccc}
 & \bra{0}_{A_v} \bra{+}_A & \bra{1}_{A_v} \bra{-}_A \\
\begin{block}{c(cc)}
  \ket{0}_{A_v} \ket{+}_A & \braket{+|\rho_{AB}|+}_A & 0 \\
  \ket{1}_{A_v} \ket{-}_A & 0  & \braket{-|\rho_{AB}|-}_A \\
\end{block}
\end{blockarray}\;\; ,
\end{equation}
which is equivalent to $\mathcal{Z}_{A}^{X} (\rho_{AB})$ in the $X$-basis:
\begin{equation}
	\mathcal{Z}_{A}^{X} (\rho_{AB}) = 
	\begin{blockarray}{ccc}
 & \bra{+}_A &  \bra{-}_A \\
\begin{block}{c(cc)}
  \ket{+}_A & \braket{+|\rho_{AB}|+}_A & 0 \\
  \ket{-}_A & 0  & \braket{-|\rho_{AB}|-}_A \\
\end{block}
\end{blockarray}\;\; ,
\end{equation}
that gives us $\infdiv*{\rho_x}{\mathcal{Z}_{A_v}^{Z} (\rho_x)} = \infdiv*{\rho_{AB}}{\mathcal{Z}_{A}^{X} (\rho_{AB})}$.

Solving the SDP problem:
\begin{equation}
	\min_{\rho_{AB}} \left[ \infdiv*{\mathcal{S}(\rho_{AB})}{\mathcal{Z}_{A_v}^{Z} (\mathcal{S}(\rho_{AB}))} \right]
\end{equation}
is therefore equivalent to solving:
\begin{equation}
\begin{aligned}
	\min_{\rho_{AB}} \left[p_Z^2 \infdiv*{\rho_{AB}}{\mathcal{Z}_{A}^{Z} (\rho_{AB})} + p_X^2 \infdiv*{\rho_{AB}}{\mathcal{Z}_{A}^{X} (\rho_{AB})} \right],
\end{aligned}
\end{equation}
which requires no dilation in the Hilbert space at all.

\subsubsection{The linearized dual SDP}
\label{sub:the_linearized_dual_sdp_1}

We also wish to show that the gradient $\nabla f(\rho)^T$, defined in Eq.~\eqref{eq:grad_f_rho_sifted}, has a similar simple structure:
\begin{equation}
\label{eq:grad_rho_sifted_app}
\begin{aligned}
\nabla f(\rho_{AB})^T &= \mathcal{S}^{\dag} \left( \log_2 \mathcal{S} (\rho_{AB}) \right) - \mathcal{S}^{\dag} \left(\log_2 \mathcal{Z}^Z_{A_v}(\mathcal{S}(\rho_{AB}))\right) \\
		&= p_Z^2 \left[\log_2 \rho_{AB} - \log_2\left(\mathcal{Z}_A^Z(\rho_{AB})\right) \right] \\
		& \quad + p_X^2 \left[\log_2 \rho_{AB} - \log_2\left(\mathcal{Z}_A^X(\rho_{AB})\right) \right].
\end{aligned}
\end{equation}

We start from the result of Eq.~\eqref{eq:sift_blocked}:
\begin{equation}
	\mathcal{S}(\rho_{AB}) = \kb{00}_{A_b B_b} \otimes p_Z^2 \rho_{z} + \kb{11}_{A_b B_b} \otimes p_X^2 \rho_{x},
\end{equation}
thus
\begin{equation}
\label{eq:first_term}
\begin{aligned}
	\log_2\left( \mathcal{S}(\rho_{AB}) \right) &= \kb{00}_{A_b B_b} \otimes \log_2 (p_Z^2 \rho_{z}) \\
	& \quad + \kb{11}_{A_b B_b} \otimes \log_2(p_X^2 \rho_{x}),
\end{aligned}
\end{equation}
which gives us the first term before applying the adjoint map $\mathcal{S}^{\dag}$.
We can also obtain the second term (before the adjoint map):
\begin{equation}
\label{eq:second_term}
\begin{aligned}
	& \log_2\left( \mathcal{Z}^Z_{A_v}(\mathcal{S}(\rho_{AB})) \right) = \kb{00}_{A_b B_b}  \otimes \log_2 \left(p_Z^2 \mathcal{Z}^Z_{A_v}(\rho_{z})\right) \\
	& + \kb{11}_{A_b B_b} \otimes \log_2 \left(p_X^2 \mathcal{Z}^Z_{A_v}(\rho_{x})\right).
\end{aligned}
\end{equation}

Now, let us apply the adjoint map to the first term (Eq.~\eqref{eq:first_term}):
\begin{equation}
\begin{aligned}
	\mathcal{S}^{\dag} (\log_2\left( \mathcal{S}(\rho_{AB}) \right)) &= (\bar{K}^0_A \bar{K}^0_B)^{\dag} \log_2 (p_Z^2 \rho_{z}) (\bar{K}^0_A \bar{K}^0_B) \\
	& + (\bar{K}^1_A \bar{K}^1_B)^{\dag} \log_2(p_X^2 \rho_{x}) (\bar{K}^1_A \bar{K}^1_B),
\end{aligned}
\end{equation}
where we have defined $K^i_X = \bar{K}^i_X \otimes \ket{i}_{X_b}$ for $i \in \set{0,1}$ and $X \in \set{A,B}$. Explicitly,
\begin{equation}
\label{eq:K0_mapping}
	\bar{K}^0_A \bar{K}^0_A = p_Z \sum_{i,j \in \set{0,1}} \ket{i_Z}_{A_v} \kb{i_Z j_Z}_{AB},
\end{equation}
and
\begin{equation}
\label{eq:K1_mapping}
	\bar{K}^1_A \bar{K}^1_A = p_X \sum_{i,j \in \set{0,1}} \ket{i_Z}_{A_v} \kb{i_X j_X}_{AB},
\end{equation}
where we have used the shorthand: $\ket{0_Z} = \ket{0}$, $\ket{1_Z} = \ket{1}$, $\ket{0_X} = \ket{+}$, $\ket{1_X} = \ket{-}$. In the previous section, we have also shown that:
\begin{enumerate}
	\item  $\rho_z$ in the $Z$-basis of $A_v$ and $A$ is equivalent to $\rho_{AB}$ in the $Z$-basis, and
	\item  $\rho_x$ in the $Z$-basis of $A_v$ and the $X$-basis of $A$ is equivalent to $\rho_{AB}$ in the $X$-basis,
\end{enumerate}
which are exactly the mapping described by the Kraus operators in Eqs.~\eqref{eq:K0_mapping} and~\eqref{eq:K1_mapping}. We therefore have
\begin{equation}
\label{eq:first_term_after}
	\mathcal{S}^{\dag} (\log_2\left( \mathcal{S}(\rho_{AB}) \right)) = p_Z^2 \log_2 (p_Z^2 \rho_{AB}) + p_X^2 \log_2 (p_X^2 \rho_{AB}).
\end{equation}

Now, let us apply the adjoint map to the second term (Eq.~\eqref{eq:second_term}):
\begin{equation}
\begin{aligned}
	&\mathcal{S}^{\dag} \left(\log_2\left( \mathcal{Z}^Z_{A_v}(\mathcal{S}(\rho_{AB})) \right)\right) \\
	&= (\bar{K}^0_A \bar{K}^0_B)^{\dag} \log_2 \left(p_Z^2 \mathcal{Z}^Z_{A_v}(\rho_{z})\right) (\bar{K}^0_A \bar{K}^0_B) \\
	& \quad + (\bar{K}^1_A \bar{K}^1_B)^{\dag} \log_2\left(p_X^2 \mathcal{Z}^Z_{A_v}(\rho_{x})\right) (\bar{K}^1_A \bar{K}^1_B).
\end{aligned}
\end{equation}
Similarly, we have previously shown that:
\begin{enumerate}
	\item  $\mathcal{Z}^Z_{A_v}(\rho_{z})$ in the $Z$-basis of $A_v$ and $A$ is equivalent to $\mathcal{Z}^Z_{A}(\rho_{AB})$ in the $Z$-basis, and
	\item  $\mathcal{Z}^Z_{A_v}(\rho_{x})$ in the $Z$-basis of $A_v$ and the $X$-basis of $A$ is equivalent to $\mathcal{Z}^X_{A}(\rho_{AB})$ in the $X$-basis.
\end{enumerate}
Therefore, 
\begin{equation}
\label{eq:second_term_after}
\begin{aligned}
	\mathcal{S}^{\dag} \left(\log_2\left( \mathcal{Z}^Z_{A_v}(\mathcal{S}(\rho_{AB})) \right)\right) &= p_Z^2 \log_2 \left(p_Z^2 \mathcal{Z}^Z_{A}(\rho_{AB})\right) \\
	&\quad + p_X^2 \log_2 \left(p_X^2 \mathcal{Z}^X_{A}(\rho_{AB})\right).
\end{aligned}
\end{equation}

Combining Eqs.~\eqref{eq:first_term_after} and~\eqref{eq:second_term_after} give us:
\begin{equation}
\begin{aligned}
\nabla f(\rho_{AB})^T &= \mathcal{S}^{\dag} \left( \log_2 \mathcal{S} (\rho_{AB}) \right) - \mathcal{S}^{\dag} \left(\log_2 \mathcal{Z}^Z_{A_v}(\mathcal{S}(\rho_{AB}))\right) \\
		&= p_Z^2 \left[\log_2 (p_Z^2 \rho_{AB}) - \log_2\left(p_Z^2 \mathcal{Z}_A^Z(\rho_{AB})\right) \right] \\
		& \quad + p_X^2 \left[\log_2 (p_X^2 \rho_{AB}) - \log_2\left(p_X^2 \mathcal{Z}_A^X(\rho_{AB})\right) \right] \\
		&= p_Z^2 \left[\log_2 \rho_{AB} - \log_2\left(\mathcal{Z}_A^Z(\rho_{AB})\right) \right] \\
		& \quad + p_X^2 \left[\log_2 \rho_{AB} - \log_2\left(\mathcal{Z}_A^X(\rho_{AB})\right) \right].
\end{aligned}
\end{equation}
To eliminate the factors of $p_Z^2$ and $p_X^2$ inside the logarithms we made use of a couple of properties of matrix logarithm for a density matrix $\rho_{AB}$, which is a positive definite operator: $\rho_{AB} \succ 0$~\footnote{More accurately, a density matrix $\rho$ is a positive semidefinite operator. We can consider the density matrix after a depolarizing channel with small depolarizing error/probability $\epsilon \geq 0$, then the resulting density operator $\rho \succ 0$. Calculating the key rate using this density matrix incurs some errors that can be quantified, see~\citep{Winick2017} for the complete details.}. First, the fact that $\log(A^{-1}) = - \log A$ for any $A \succ 0$. Secondly, we use the fact that $\rho_{AB}$ commutes with $\mathcal{Z}_A^Z(\rho_{AB})^{-1}$ and $\mathcal{Z}_A^X(\rho_{AB})^{-1}$, as the pinched density matrices are diagonal in the basis they are pinched on. This concludes our proof of the simplified form of $\nabla f(\rho)^T$.

% subsection the_linearized_dual_sdp (end)

\subsection{Prepare-and-measure BB84} % (fold)
\label{sec:prepare_and_measure_bb84}

The simplification described here can be used for the generic prepare-and-measure BB84 QKD protocol and for the case with Trojan horse attack. The general procedure is similar to the procedure described in Sec.~\ref{sec:bb84}, so we will only show the important steps.

The Kraus operators related to Alice and Bob's announcements are:
\begin{equation}
\begin{aligned}
	K^0_A &= \kb{0}_A \otimes \ket{0}_{A_b} \otimes \ket{0}_{A_v} + \kb{1}_A \otimes \ket{0}_{A_b} \otimes \ket{1}_{A_v}  \\
	K^1_A &=\kb{2}_A \otimes \ket{1}_{A_b} \otimes \ket{0}_{A_v} + \kb{3}_A \otimes \ket{1}_{A_b} \otimes \ket{1}_{A_v}  \\
	K^0_B &= \sqrt{p_Z} \left[ \kb{0}_B \otimes \ket{0}_{B_b} + \kb{1}_B \otimes \ket{0}_{B_b}\right] \\
	&= \sqrt{p_Z} \id_{B} \otimes \ket{0}_{B_b} \\
	K^1_B &= \sqrt{p_X} \left[ \kb{+}_B \otimes \ket{1}_{B_b} + \kb{-}_B \otimes \ket{1}_{B_b}\right] \\
	&= \sqrt{p_X} \id_{B} \otimes \ket{1}_{B_b},
\end{aligned}
\end{equation}
where there is we did not keep track $B_v$. With the postselection operator $\Pi = \kb{0}_{A_b} \otimes \kb{0}_{B_b} + \kb{1}_{A_b} \otimes \kb{1}_{B_b}$, we obtain:
\begin{equation}
	\mathcal{S}(\rho_{AB}) = \kb{00}_{A_b B_b} p_Z \rho_z + \kb{11}_{A_b B_b} p_X \rho_x.
\end{equation}

\subsubsection{The key rate SDP}
\label{sub:the_key_rate_sdp}

We will show that the key rate problem simplifies as follows:
\begin{equation}
\label{eq:pm_bb84_simplified}
\begin{aligned}
	\infdiv*{\mathcal{S}(\rho_{AB})}{\mathcal{Z}_{A_v}^{Z} (\mathcal{S}(\rho_{AB}))} &= p_Z \infdiv*{\rho_{AB}^{01}}{\mathcal{Z}_{A}^{Z} (\rho_{AB}^{01})} \\
	&+ p_X \infdiv*{\rho_{AB}^{23}}{\mathcal{Z}_{A}^{Z} (\rho_{AB}^{23})},
\end{aligned}
\end{equation}
where
\begin{equation}
	\rho_{AB}^{01} = \sum_{i,j=\{0,1\}} \ket{i}\bra{j}_A \braket{i|\rho_{AB}|j}_A,
\end{equation}
describing the upper-left block of $\rho_{AB}$ in the standard basis of $A$, i.e. in the $\left\{ \ket{0}_A, \ket{1}_A \right\}$ subspace. And,
\begin{equation}
	\rho_{AB}^{23} = \sum_{i,j=\{2,3\}} \ket{i}\bra{j}_A \braket{i|\rho_{AB}|j}_A,
\end{equation}
describing the lower-right block of $\rho_{AB}$ in the standard basis of $A$, i.e. in the $\left\{ \ket{2}_A, \ket{3}_A \right\}$ subspace.

Applying Lemma~\ref{lemma:block_diagonal} to $\infdiv*{\mathcal{S}(\rho_{AB})}{\mathcal{Z}_{A_v}^{Z} (\mathcal{S}(\rho_{AB})}$ gives us:
\begin{equation}
\begin{aligned}
	\infdiv*{\mathcal{S}(\rho_{AB})}{\mathcal{Z}_{A_v}^{Z} (\mathcal{S}(\rho_{AB}))} &= p_Z \infdiv*{\rho_z}{\mathcal{Z}_{A_v}^{Z} (\rho_z)} \\
	& + p_X \infdiv*{\rho_x}{\mathcal{Z}_{A_v}^{Z} (\rho_x)}.
\end{aligned}
\end{equation}
We then write down $\rho_z$ in the basis of $A_v$ and $A$:
\begin{equation}
	\rho_z = \sum_{i,j \in \set{0,1}} \ket{i}_A\ket{i}_{A_v} \bra{j}_A \bra{j}_{A_v} \braket{i| \rho_{AB}|j}_{A},
\end{equation}
which is equivalent to the upper-left block of $\rho_{AB}$. Similarly, applying the pinching channel on $\rho_z$ gives us:
\begin{equation}
	\mathcal{Z}_{A_v}^{Z} (\rho_z) = \sum_{i \in \set{0,1}} \ket{i}_A\ket{i}_{A_v} \bra{i}_A \bra{i}_{A_v} \braket{i| \rho_{AB}|i}_{A}.
\end{equation}
Thus, we conclude that
\begin{equation}
	\infdiv*{\rho_z}{\mathcal{Z}_{A_v}^{Z} (\rho_z)} = \infdiv*{\rho_{AB}^{01}}{\mathcal{Z}_{A}^{Z} (\rho_{AB}^{01})}.
\end{equation}
Next, we can also write $\rho_x$:
\begin{equation}
	\rho_x = \sum_{i,j \in \set{2,3}} \ket{i}_A\ket{i}_{A_v} \bra{j}_A \bra{j}_{A_v} \braket{i| \rho_{AB}|j}_{A},
\end{equation}
which is the lower-right block of $\rho_{AB}$. Applying the pinching channel on $\rho_x $ gives us:
\begin{equation}
	\mathcal{Z}_{A_v}^{Z} (\rho_x) = \sum_{i\in\set{2,3}} \ket{i}_A\ket{i}_{A_v} \bra{i}_A \bra{i}_{A_v} \braket{i| \rho_{AB}|i}_{A}.
\end{equation}
Finally,
\begin{equation}
	\infdiv*{\rho_x}{\mathcal{Z}_{A_v}^{Z} (\rho_x)} = \infdiv*{\rho_{AB}^{23}}{\mathcal{Z}_{A}^{Z} (\rho_{AB}^{23})}.
\end{equation}

\subsubsection{The linearized dual SDP}
\label{sub:the_linearized_dual_sdp_2}

Here we will show that
\begin{equation}
\label{eq:grad_rho_sifted_pm_bb84}
\begin{aligned}
\nabla f(\rho_{AB})^T &= \mathcal{S}^{\dag} \left( \log_2 \mathcal{S} (\rho_{AB}) \right) - \mathcal{S}^{\dag} \left(\log_2 \mathcal{Z}^Z_{A_v}(\mathcal{S}(\rho_{AB}))\right) \\
		&= p_Z \left[\log_2 \rho_{AB}^{01} - \log_2\left(\mathcal{Z}_A^Z(\rho_{AB}^{01})\right) \right] \\
		& \quad + p_X \left[\log_2 \rho_{AB}^{23} - \log_2\left(\mathcal{Z}_A^Z(\rho_{AB}^{23})\right) \right].
\end{aligned}
\end{equation}

The first term before the adjoint map $\mathcal{S}^{\dag}$ is:
\begin{equation}
\begin{aligned}
	\log_2\left( \mathcal{S}(\rho_{AB}) \right) &= \kb{00}_{A_b B_b} \otimes \log_2 (p_Z \rho_{z}) \\
	& + \kb{11}_{A_b B_b} \otimes \log_2(p_X \rho_{x}).
\end{aligned}
\end{equation}
Applying the map:
\begin{equation}
\label{eq:first_term_pm_bb84}
\begin{aligned}
	\mathcal{S}^{\dag} (\log_2\left( \mathcal{S}(\rho_{AB}) \right)) &= (\bar{K}^0_A \bar{K}^0_B)^{\dag} \log_2 (p_Z \rho_{z}) (\bar{K}^0_A \bar{K}^0_B) \\
	& \quad + (\bar{K}^1_A \bar{K}^1_B)^{\dag} \log_2(p_X \rho_{x}) (\bar{K}^1_A \bar{K}^1_B) \\
	&= p_Z \log_2 (p_Z \rho_{AB}^{01}) + p_X \log_2 (p_X \rho_{AB}^{23}),
\end{aligned}
\end{equation}
as the Kraus operators are
\begin{equation}
\label{eq:K0_mapping_pm_bb84}
	\bar{K}^0_A \bar{K}^0_A = \sqrt{p_Z} \sum_{i,j \in \set{0,1}} \ket{i_Z}_{A_v} \kb{i_Z j_Z}_{AB},
\end{equation}
and
\begin{equation}
\label{eq:K1_mapping_pm_bb84}
	\bar{K}^1_A \bar{K}^1_A = \sqrt{p_X} \sum_{i\in \set{2,3}} \sum_{j \in \set{0,1}}\ket{i_Z}_{A_v} \kb{i_Z j_X}_{AB}.
\end{equation}

Similarly, the second term before the adjoint map is:
\begin{equation}
\begin{aligned}
	\log_2\left( \mathcal{Z}^Z_{A_v}(\mathcal{S}(\rho_{AB})) \right) &= \kb{00}_{A_b B_b} \otimes \log_2\left(p_Z \mathcal{Z}^Z_{A_v}(\rho_{z}) \right) \\
	& + \kb{11}_{A_b B_b} \otimes \log_2\left(p_X \mathcal{Z}^Z_{A_v}(\rho_{x}) \right),
\end{aligned}
\end{equation}
and applying the map gives us
\begin{equation}
\label{eq:second_term_pm_bb84}
\begin{aligned}
	&\mathcal{S}^{\dag} \left(\log_2\left( \mathcal{Z}^Z_{A_v}(\mathcal{S}(\rho_{AB})) \right)\right) \\
	&= (\bar{K}^0_A \bar{K}^0_B)^{\dag} \log_2 \left(p_Z \mathcal{Z}^Z_{A_v}(\rho_{z})\right) (\bar{K}^0_A \bar{K}^0_B) \\
	& \quad + (\bar{K}^1_A \bar{K}^1_B)^{\dag} \log_2\left(p_X \mathcal{Z}^Z_{A_v}(\rho_{x})\right) (\bar{K}^1_A \bar{K}^1_B) \\
	&= p_Z \log_2 \left(p_Z \mathcal{Z}^Z_{A}(\rho_{AB}^{01})\right) + p_X \log_2\left(p_X \mathcal{Z}^Z_{A}(\rho_{AB}^{23})\right).
\end{aligned}
\end{equation}

Combining Eqs.~\eqref{eq:first_term_pm_bb84} and~\eqref{eq:second_term_pm_bb84}, and then eliminating the probability terms inside the logarithms, we obtain Eq.~\eqref{eq:grad_rho_sifted_pm_bb84}.

\subsection{B92}
In analyzing the B92 protocol, we consider a simplified postselection framework introduced in~\citep{Coles2016}. Let $\mathcal{S}$ be the completely positive linear map corresponding to postselection that is only given by a single Kraus operator $S$, i.e. $\mathcal{S}(\rho) = S \rho S^{\dag}$ such that $S^{\dag} S = \id$. In this case, we can consider a new key rate problem:
\begin{equation}
	H(Z_A | E) = \frac{1}{\ppass} \infdiv*{\mathcal{S}(\rho_{AB})}{\sum_j Z_A^j \mathcal{S}(\rho_{AB}) Z_A^j}.
\end{equation}
In this case, let us define $\tilde{\rho}_{AB} \equiv \mathcal{S}(\rho_{AB})$ such that we can translate the optimization problem to one that involves $\tilde{\rho}_{AB}$ instead of $\rho_{AB}$, i.e.
\begin{equation}
\begin{aligned}
	\min_{\rho_{AB}} H(Z_A | E) &= \min_{\rho_{AB}} \frac{1}{\ppass} \infdiv*{\mathcal{S}(\rho_{AB})}{\sum_j Z_A^j \mathcal{S}(\rho_{AB}) Z_A^j} \\
	&\equiv \min_{\tilde{\rho}_{AB}} \frac{1}{\ppass} \infdiv*{\tilde{\rho}_{AB}}{\sum_j Z_A^j \tilde{\rho}_{AB} Z_A^j},
\end{aligned}
\end{equation}
which can be solved using the SDP formulation without postselection. Similarly, in the min-entropy formulation, we have
\begin{equation}
\begin{aligned}
	&\min_{\rho_{AB}} H_{\min}(Z_A | E) \\
	&\qquad =-\log_2 \max_{\rho_{AB}} \max_{\sigma_{AB}} F\left(\mathcal{S}(\rho_{AB}), \sum_j Z^j_A \sigma_{AB} Z^j_A \right) \\
	&\qquad\qquad + \log_2 \ppass \\
	&\qquad \equiv -\log_2 \max_{\tilde{\rho}_{AB}} \max_{\sigma_{AB}} F\left(\tilde{\rho}_{AB}, \sum_j Z^j_A \sigma_{AB} Z^j_A \right) \\
	&\qquad\qquad + \log_2 \ppass.
\end{aligned}
\end{equation}
The constraints for $\tilde{\rho}_{AB}$ can be computed from:
\begin{equation}
\begin{aligned}
	\Tr(\rho_{AB} \Gamma_i) &= \Tr(S^{\dag} S \rho_{AB} S^{\dag} S \Gamma_i) \\
		&= \Tr(\tilde{\rho}_{AB} S \Gamma_i S^{\dag}) \\
		&= \Tr(\tilde{\rho}_{AB} \tilde{\Gamma}_i),
\end{aligned}
\end{equation}
where $\tilde{\Gamma}_i \equiv S \Gamma_i S^{\dag}$.

The Kraus operator for the postselection procedure is:
\begin{equation}
	S = \id_A \otimes \sqrt{\frac{1}{2} \left( \kb{\bar{\phi_0}}_B + \kb{\bar{\phi_1}}_B \right)},
\end{equation}
and the key maps are:
\begin{equation}
	Z_A = \{\kb{0}_A, \kb{1}_A \}.
\end{equation}

\section{Other Examples}
\label{sec:other_examples}

\subsection{BB84 with Efficiency Mismatch}
\label{sub:bb84_with_efficiency_mismatch}

We again consider the entanglement-based BB84 protocol. However, we now quantify the efficiency of each of Bob's detectors. In the $Z$-basis, his detector efficiencies are $\eta_{Z_0}$ and $\eta_{Z_1}$, and in the $X$-basis, his detector efficiencies are $\eta_{X_0}$ and $\eta_{X_1}$. This is an asymmetric scenario that one often encounters with practical QKD systems: no two detectors have the same exact detection efficiency. As has been suggested in Ref.~\cite{Winick2017}, Bob's measurement operators are the same as in the entanglement-based picture, except we now model his system as a qutrit, where the single photon exists in a single qubit subspace, and the third dimension describes the vacuum which will contribute to the no-click event. Bob's $Z$-basis states are $\{\ket{0}, \ket{1}, \ket{\emptyset} \}$ and his $X$-basis states are $\{\ket{+}, \ket{-}, \ket{\emptyset} \}$, with $\braket{\emptyset|0} = \braket{\emptyset|1} = 0$.

We can model the problem with the following POVMs for Alice:
\begin{center}
\begin{tabular}{c c c}
	\hline
	 $M_A^{(a_b, a_v)}$ & $a_b$ & $a_v$ \\
	 \hline	
	 $p_Z \ket{0}\bra{0}$ & 0 & 0 \\
	 $p_Z \ket{1}\bra{1}$ & 0 & 1 \\
	 $p_X \ket{+}\bra{+}$ & 1 & 0 \\
	 $p_X \ket{-}\bra{-}$ & 1 & 1 \\
   	\hline
\end{tabular}
\end{center}

Bob's POVMs are
\begin{center}
\begin{tabular}{c c c}
	\hline
	 $M_B^{(b_b, b_v)}$ & $b_b$ & $b_v$ \\
	  \hline	
	 $p_Z \eta_{Z_0} \ket{0}\bra{0}$ & 0 & 0 \\
	 $p_Z \eta_{Z_1} \ket{1}\bra{1}$ & 0 & 1 \\
	 $p_X \eta_{X_0} \ket{+}\bra{+}$ & 1 & 0 \\
	 $p_X \eta_{X_1} \ket{-}\bra{-}$ & 1 & 1 \\
	 $M_B^{\emptyset}$ & 2 & 0 \\
   	\hline
\end{tabular}
\end{center}
Here, 
\begin{equation}
	M_B^{\emptyset} = \id - \sum_{i,j=0}^{1} M_B^{(i,j)}.
\end{equation}
We generate keys from both the $Z$ and $X$-bases using the key-map POVM:
\begin{equation}
	Z_{A_v} = \{\kb{0}_{A_v}, \kb{1}_{A_v} \}.
\end{equation}

Similar to the previous example, we model the statistics with a depolarizing channel:
\begin{equation}
	\rho_{AB}' = (\id_A \otimes \mathcal{E}_{B}^{\text{dep}}(p)) (\kb{\psi}_{AB}).
\end{equation}
Due to the asymmetry in Bob's measurements, in addition to the error constraints, we must also include the constraints in which both parties measure the same values in the same basis. Thus, we have a total of four coarse-grained constraints: 
\begin{equation}
\begin{aligned}
	\Gamma^{(=)}_{Z} &= p_Z \left(\eta_{Z_0} \ket{0}\bra{0}_A \otimes \ket{0}\bra{0}_B \right. \\
	& \quad \left. + \eta_{Z_1} \ket{1}\bra{1}_A \otimes \ket{1}\bra{1}_B \right), \\
	\Gamma^{(=)}_{X} &= p_X \left(\eta_{X_0} \ket{2}\bra{2}_A \otimes \ket{+}\bra{+}_B \right. \\
	& \quad \left. + \eta_{X_1} \ket{3}\bra{3}_A \otimes \ket{-}\bra{-}_B \right), \\
	E_Z &= p_Z \left( \eta_{Z_1} \ket{0}\bra{0}_A \otimes \ket{1}\bra{1}_B \right. \\
	& \quad \left.+ \eta_{Z_0} \ket{1}\bra{1}_A \otimes \ket{0}\bra{0}_B \right),  \\
	E_X &= p_X \left( \eta_{X_1} \ket{+}\bra{+}_A \otimes \ket{-}\bra{-}_B \right. \\
	&\quad \left. + \eta_{X_0} \ket{-}\bra{-}_A \otimes \ket{+}\bra{+}_B \right).
\end{aligned}
\end{equation}

The key rate in the asymptotic limit is simple if we assume all detectors have the same efficiency. The situation is more interesting if we consider $\eta_{Z_0} = \eta_{X_0} = \eta_0$ and $\eta_{Z_1} = \eta_{X_1} = \eta_1$, in which there is detector efficiency mismatch within a single basis. This QKD configuration has been treated analytically by~\cite{Fung2009} in the asymptotic regime, and the key rate analytical formula derived is
\begin{equation}
\label{eq:eff_mismatch_analytic}
	r_1 = \min(\eta_0, \eta_1) \left[1-h_2(Q) \right] - h_2(Q).
\end{equation}
The numerical analysis in the asymptotic limit has been considered in Ref.~\cite{Winick2017}.
%In Figure~\ref{fig:bb84_mismatch_asymptotic}, we show that the secret key rate guaranteed by the numerical method outperforms the analytical method from~\citep{Fung2009} at most points. For this numerical simulation, we consider the highly asymmetric protocol by choosing $\eta_Z = 99\%$. For all the cases we considered, we see that our method provides an equivalent or better key rates at all points.

% \begin{figure}[H]
%     \centering
%     \includegraphics[width=\columnwidth]{bb84_mismatch_asymptotic.eps}
%     \caption{Asymptotic secret key rate per pulse for the BB84 protocol with detector efficiency mismatch. We fix detector ``0'' to have an efficiency $\eta_0$ of either $50\%$ or $100\%$, and we sweep over the efficiency of the other detector $\eta_1$. The key rates are evaluated for two QBER values $Q=0\%$ and $1\%$. The secret key rates calculated using our numerical method are shown in connected dots. The dashed and dotted lines with the same color show the corresponding key rates calculated using the analytical formula Eq.~\eqref{eq:eff_mismatch_analytic}. A similar plot is shown in Ref.~\cite{Winick2017}}
%     \label{fig:bb84_mismatch_asymptotic}
% \end{figure}

Here, we consider this problem in the nonasymptotic regime---taking into account the effects of statistical fluctuations during the QKD operations. Similar to the case of entanglement-based BB84, we assume equal security parameter $\varepsilon'$, and the value of each relevant parameter is tabulated in Tab.~\ref{tab:security_params}. %Combining all the failure probabilities, we have $\epssec = 11 \varepsilon'$ for calculation with von Neumann entropy (Eq.~\eqref{eq:keyrate_eq_vnm} with SDPs~\eqref{eq:key_rate_nonasymptotic_appx_primal} and~\eqref{eq:lower_bound_nonasymptotic}) and $\epssec = 10 \varepsilon'$ for calculation with min-entropy (Eq.~\eqref{eq:keyrate_eq_min_entr} with SDP~\eqref{eq:keyrate_Hmin_dual_problem}).
Fig.~\ref{fig:bb84_mismatch_finite} plots the secret key rate per pulse in terms of the number of pulses are generated by Alice, assuming $\epssec = 10^{-10}$ and $\epscor = 10^{-15}$. We consider the case where $\eta_0 = 1$ and $\eta_1 = 25\%$ or $75\%$. We see that, in this case, the bound from von Neumann entropy consistently outperforms the bound from min-entropy.

\begin{figure}[H]
    \centering
    \begin{subfigure}[b]{\columnwidth}
	\includegraphics[width=\columnwidth]{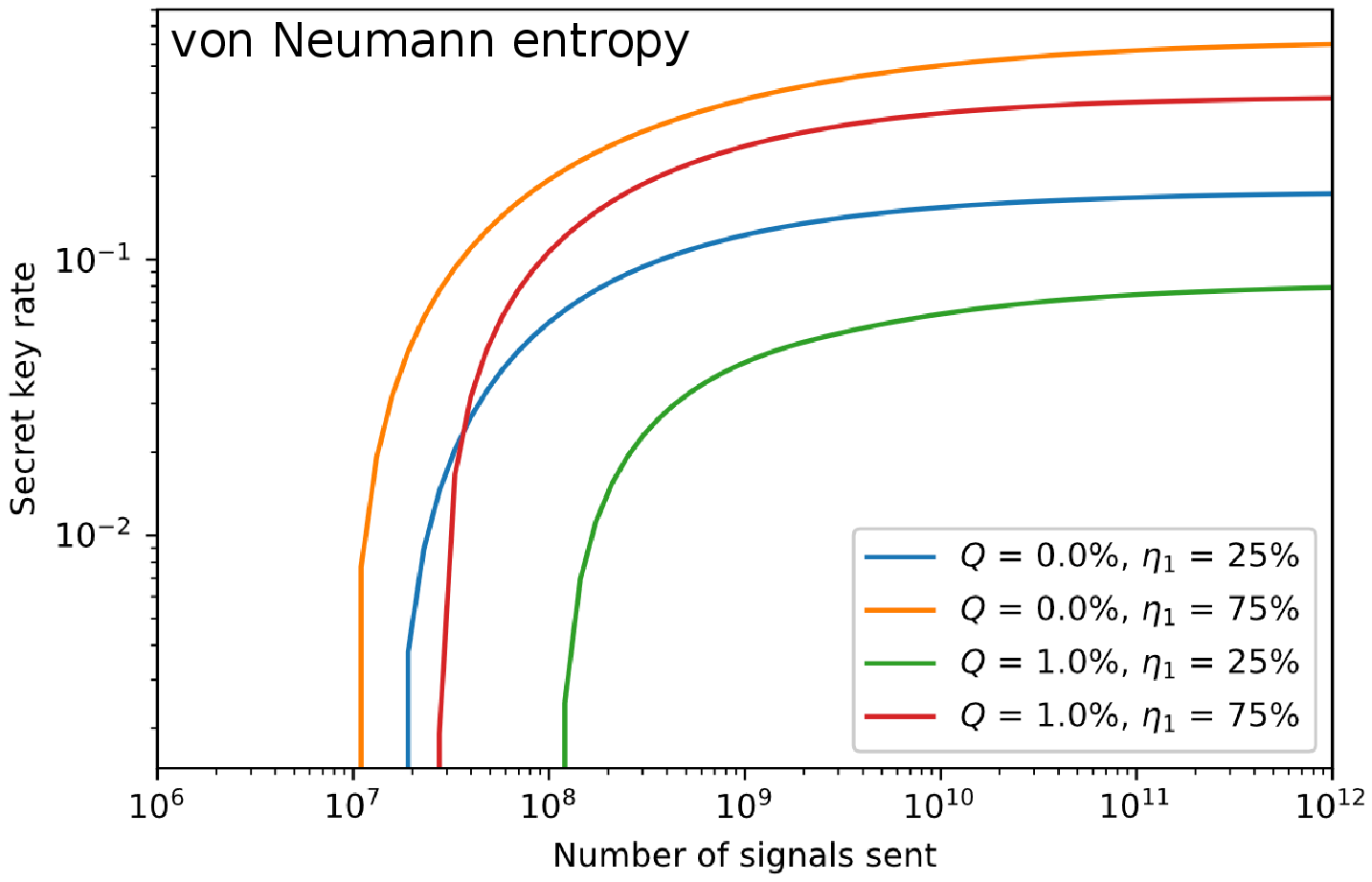}
	\end{subfigure}
	\begin{subfigure}[b]{\columnwidth}
	\includegraphics[width=\columnwidth]{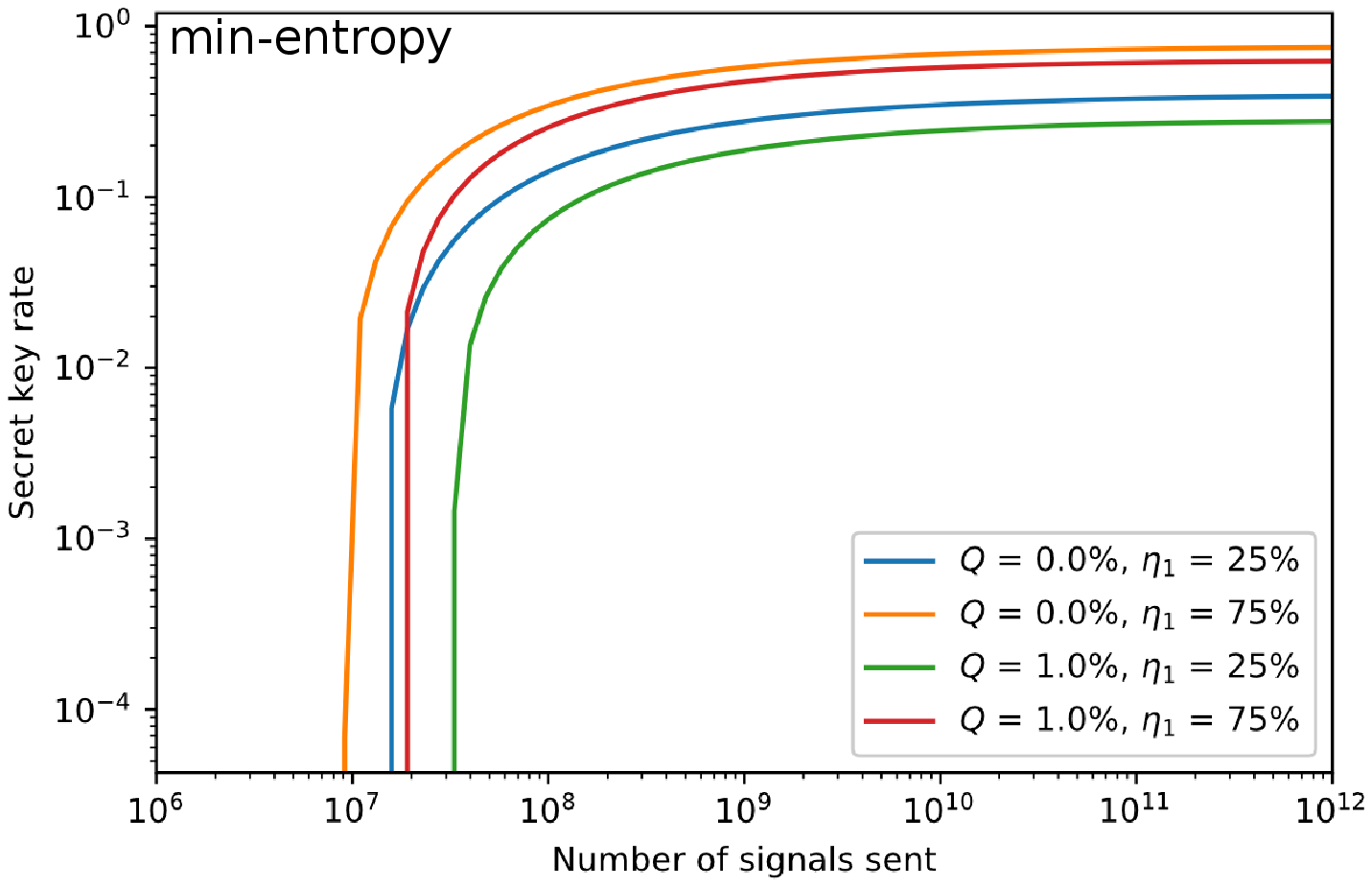}
	\end{subfigure}
    \caption{Secret key rate per pulse for the BB84 protocol calculated using the von Neumann entropy (top) and the min-entropy (bottom). We consider different values of error rate $Q$ and efficiency of the ``1'' detector $\eta_1$. We assume that the ``0'' detector has unit efficiency, i.e. $\eta_0 = 1$. For all plots, the lines are the results of our numerical method. }
    \label{fig:bb84_mismatch_finite}
\end{figure}

\subsection{Trojan-horse attack} % (fold)
\label{sub:trojan_horse_attack}

A common assumption for QKD analyses is that Eve cannot access Alice's laboratory. As its name suggests, the Trojan-horse attack is a side-channel attack where Eve tries to infiltrate Alice's laboratory to obtain information about the state Alice has sent towards Bob. In particular, Eve uses the optical link between Alice and Bob to launch a bright light pulse into Alice's supposedly secure module. The light pulse will reach Alice's encoding device and is encoded with the same information, e.g. the phase value $\varphi$, as the signal prepared by Alice. Some of these Trojan photons are reflected back to Eve. Although the information $\varphi$ is meant to be kept private by Alice, Eve can perform measurements on these back-reflected photons that may allow her to unambiguously learn about the value of $\varphi$. At the end of the QKD session, Eve can in principle obtain the same key as Alice and Bob---without her presence being detected by either Alice or Bob. The security of a QKD protocol can be seriously compromised if components are not installed to prevent these possible back-reflected lights. It has in fact been shown that the phase values $\varphi$ from an encoding device can be discriminated with higher than 90\% success probability using only three photons~\citep{Jain2014}.

Different solutions have been proposed to counteract the Trojan-horse attack. Alice could install an active phase randomizer~\citep{Gisin2006,Zhao2007} to remove the phase reference from Eve's hands or she could install a watchdog detector~\citep{Muller1997,Stucki2002} that alerts her when a bright pulse is injected into her setup. Countermeasures can also be realized with only passive components, e.g. optical fiber loops, filters, and isolators, which are simple to implement and to characterize experimentally~\citep{Lucamarini2015b}.

The security against the Trojan-horse attack using these passive countermeasures is based on the laser induced damage threshold (LIDT) of the optical components. Let us assume that Eve injects a coherent state $\ket{\sqrt{\mu_{\tin}}}$ with an average photon number $\mu_{\tin}$ into Alice's system. The Trojan photons will then acquire a phase modulation information $\varphi$ that will return to Eve as $\ket{e^{i\varphi} \sqrt{\mu_{\tout}}}$, where $\mu_{\tout} = \gamma \mu_{\tin}$, with $\gamma \ll 1$ describing the isolation factor of Alice's devices. If the value of $\mu_{\tin}$ is unbounded, the QKD protocol is insecure against a Trojan-horse attack. Fortunately, the value of $\mu_{\tin}$ is bounded by the LIDT of Alice's components.

Let us call $\dot{N}$ the maximum number of photons per second Eve is allowed to inject into Alice's lab without burning any of Alice's components. We assume that Alice has characterized this value of $\dot{N}$ very well. For a QKD system with a clock rate of $f$, Alice can assume the worst case scenario in which Trojan photons with a mean photon number $\mu_{\tout} = \dot{N} \gamma / f$ are emitted back to Eve at each transmission. Here, Alice has assumed that Eve distributes all her Trojan photons evenly in each round of transmission. The validity of this assumption relies on the convexity of the key rate as a function of $\mu_{\tout}$, which was shown in Ref.~\cite{Lucamarini2015b}.

Let us assume that Alice prepares the states she sends to Bob using a single-photon source in the prepare-and-measure scheme. She encodes her information in the phase difference $\varphi$ between the leading and the trailing single-photon pulse. To model the Trojan-Horse Attack in the BB84 protocol, we use the approach outlined in Ref.~\citep{Winick2017}. In this case, the state she prepares can be written as
\begin{equation}
\begin{aligned}
	\ket{\psi}_{ABE} &= \sqrt{\frac{p_Z}{2}} \left[ \ket{0}_A \ket{\phi_{z_+}}_{BE} + \ket{1}_A \ket{\phi_{z_-}}_{BE} \right] \\
	&+ \sqrt{\frac{p_X}{2}} \left[ \ket{2}_A \ket{\phi_{x_+}}_{BE} + \ket{3}_A \ket{\phi_{x_-}}_{BE} \right],
\end{aligned}
\end{equation}
where
\begin{equation}
\begin{aligned}
\ket{\phi_{z_{\pm}}}_{BE} &= \ket{z_{\pm}}_B \ket{\pm\sqrt{\mu_{\tout}}}_E \\
\ket{\phi_{x_{\pm}}}_{BE} &= \ket{x_{\pm}}_B \ket{\pm i \sqrt{\mu_{\tout}}}_E.
\end{aligned}
\end{equation}
Here, the states $\ket{z_{\pm}}$ and $\ket{x_{\pm}}$ are defined as
\begin{equation}
\begin{aligned}
\ket{z_{\pm}} &\equiv \frac{1}{\sqrt{2}} \left[ \ket{1}_l \ket{0}_t \pm \ket{0}_l \ket{1}_t\right], \\
\ket{x_{\pm}} &\equiv \frac{1}{\sqrt{2}} \left[ \ket{1}_l \ket{0}_t \pm i\ket{0}_l \ket{1}_t\right],
\end{aligned}
\end{equation}
where $\ket{n}_l$ and $\ket{n}_t$ denote an $n$-photon state in the leading and trailing pulse, respectively. Here, $p_Z$ denotes the probability she prepares a state in the $Z$-basis and $p_X = 1-p_Z$ denotes the probability she prepares a state in the $X$-basis. She sends the $B$ system to Bob, and Bob measures his state either in the $Z = \{\ket{0}, \ket{1}\}$ with probability $p_Z$ or in the $X = \{\ket{+}, \ket{-}\}$ basis with probability $p_X$. For simplicity, we again assume Alice and Bob make the their basis choices with the same probabilities $p_Z$ and $p_X$.

The protocol under consideration here is the prepare-and-measure BB84, thus the measurement POVMs for Alice are the standard basis:
\begin{center}
\begin{tabular}{c c c}
	\hline
	 $M_A^{(a_b, a_v)}$ & $a_b$ & $a_v$ \\
	  \hline	
	 $\ket{0}\bra{0}$ & 0 & 0 \\
	 $\ket{1}\bra{1}$ & 0 & 1 \\
	 $\ket{2}\bra{2}$ & 1 & 0 \\
	 $\ket{3}\bra{3}$ & 1 & 1 \\
   	\hline
\end{tabular}
\end{center}
and Bob's POVMs are
\begin{center}
\begin{tabular}{c c c}
	\hline
	 $M_B^{(b_b, b_v)}$ & $b_b$ & $b_v$ \\
	  \hline	
	 $p_Z \ket{z_+}\bra{z_+}$ & 0 & 0 \\
	 $p_Z \ket{z_-}\bra{z_-}$ & 0 & 1 \\
	 $p_X \ket{x_+}\bra{x_+}$ & 1 & 0 \\
	 $p_X \ket{x_-}\bra{x_-}$ & 1 & 1 \\
   	\hline
\end{tabular}
\end{center}

We again consider the depolarizing channel to simulate the statistics in our calculations:
\begin{equation}
	\rho'_{AB} = (\id_A \otimes \mathcal{E}^{\text{dep}}_B (p)) \Tr_E(\kb{\psi}_{ABE}),
\end{equation}
and we postselect for those cases where Alice and Bob choose the same basis:
\begin{equation}
	\Pi = \ket{0}\bra{0}_{A_b} \otimes \ket{0}\bra{0}_{B_b} + \ket{1}\bra{1}_{A_b} \otimes \ket{1}\bra{1}_{B_b}.
\end{equation}

Due to the asymmetry of this problem, we must include four constraints: two for the case when Alice and Bob obtain the same values and two for the case when they do not. They are
\begin{equation}
\begin{aligned}
	\Gamma^{(=)}_{Z} &= p_Z \left(\ket{0}\bra{0}_A \otimes \ket{0}\bra{0}_B + \ket{1}\bra{1}_A \otimes \ket{1}\bra{1}_B \right), \\
	\Gamma^{(=)}_{X} &= p_X \left(\ket{2}\bra{2}_A \otimes \ket{+}\bra{+}_B + \ket{3}\bra{3}_A \otimes \ket{-}\bra{-}_B \right), \\
	\Gamma^{(\neq)}_{Z} &= p_Z \left(\ket{0}\bra{0}_A \otimes \ket{1}\bra{1}_B + \ket{1}\bra{1}_A \otimes \ket{0}\bra{0}_B \right), \\
	\Gamma^{(\neq)}_{X} &= p_X \left(\ket{2}\bra{2}_A \otimes \ket{-}\bra{-}_B + \ket{3}\bra{3}_A \otimes \ket{+}\bra{+}_B \right).
\end{aligned}
\end{equation}

In addition to these constraints, we assume that Alice has characterized her source well. In other words, the state $\rho_{A} = \Tr_{BE}(\rho_{ABE}) = \Tr_B(\rho'_{AB})$ is known exactly to her, such that we can add a set of tomographically complete observables $\{\Omega_A^j\}$ on system $A$, and also add the calculated corresponding expectation values $\{\omega_j\}$ into the set of constraints. We add the constraints:
\begin{equation}
	\Tr\left[\rho_{AB}' (\Omega^j_A \otimes \id_B) \right] = \omega_j,
\end{equation}
which are known exactly with a failure probability of zero even in the nonasymptotic regime. Since $\rho_A$ is a valid normalized density operator, we find the values of $\{\Omega_A^j\}$ and $\{\omega_j\}$ in our simulations by computing the spectral decomposition of $\rho_A = \sum_j p_j \kb{\psi_j}_A$, such that $\Omega_A^j \equiv \kb{\psi_j}_A$ and $\omega_j \equiv p_j$. Under the Trojan-horse attack, Alice's state $\rho_{A}$ is constrained to be the following:
\begin{equation}
	\rho_{A} = \begin{pmatrix}
		p_Z/2 & 0 & \kappa &  \kappa^{*} \\
		0 & p_Z/2 & \kappa^{*} &  \kappa \\
		\kappa^{*} & \kappa & p_X/2 & 0 \\
		\kappa & \kappa^{*} & 0 & p_X/2 
	\end{pmatrix}, 	
\end{equation}
where $\kappa = \frac{1}{4}(1-i)\sqrt{p_Z p_X}\exp(-\mu_{\tout}[1+i])$ and $\kappa^*$ is the complex conjugate of $\kappa$.

% We compare the key rates generated by our numerical methods in the asymptotic limit compared to the key rate derived by~\cite{Lucamarini2015b}. Fig.~\ref{fig:tha_asymptotic} shows the comparison between the key rates computed with our numerics and the key rates computed with the analytic bound---as has also been shown in Ref.~\cite{Winick2017}. In the numerics, we take the asymmetric basis choice of $p_Z = 99\%$ to push it close to 1. It is clear that the key rates derived from the numerical approach are consistently higher than that of the analytical bound.

% \begin{figure}[H]
%     % \centering
%     \includegraphics[width=\columnwidth]{tha_asymptotic.eps}
%     \caption{Asymptotic secret key rate per pulse for the BB84 protocol under a Trojan-horse attack. Solid lines show the rate computed with our numerical method which improves on previous analytical results by~\cite{Lucamarini2015b} plotted in dashed lines. The rate is plotted for different quantum bit error rates (QBERs) against different values of $\mu_{\tout}$. A similar plot is shown in Ref.~\cite{Winick2017} }
%     \label{fig:tha_asymptotic}
% \end{figure}

\begin{figure}[H]
    % \centering
    \begin{subfigure}[b]{\columnwidth}
	\includegraphics[width=\columnwidth]{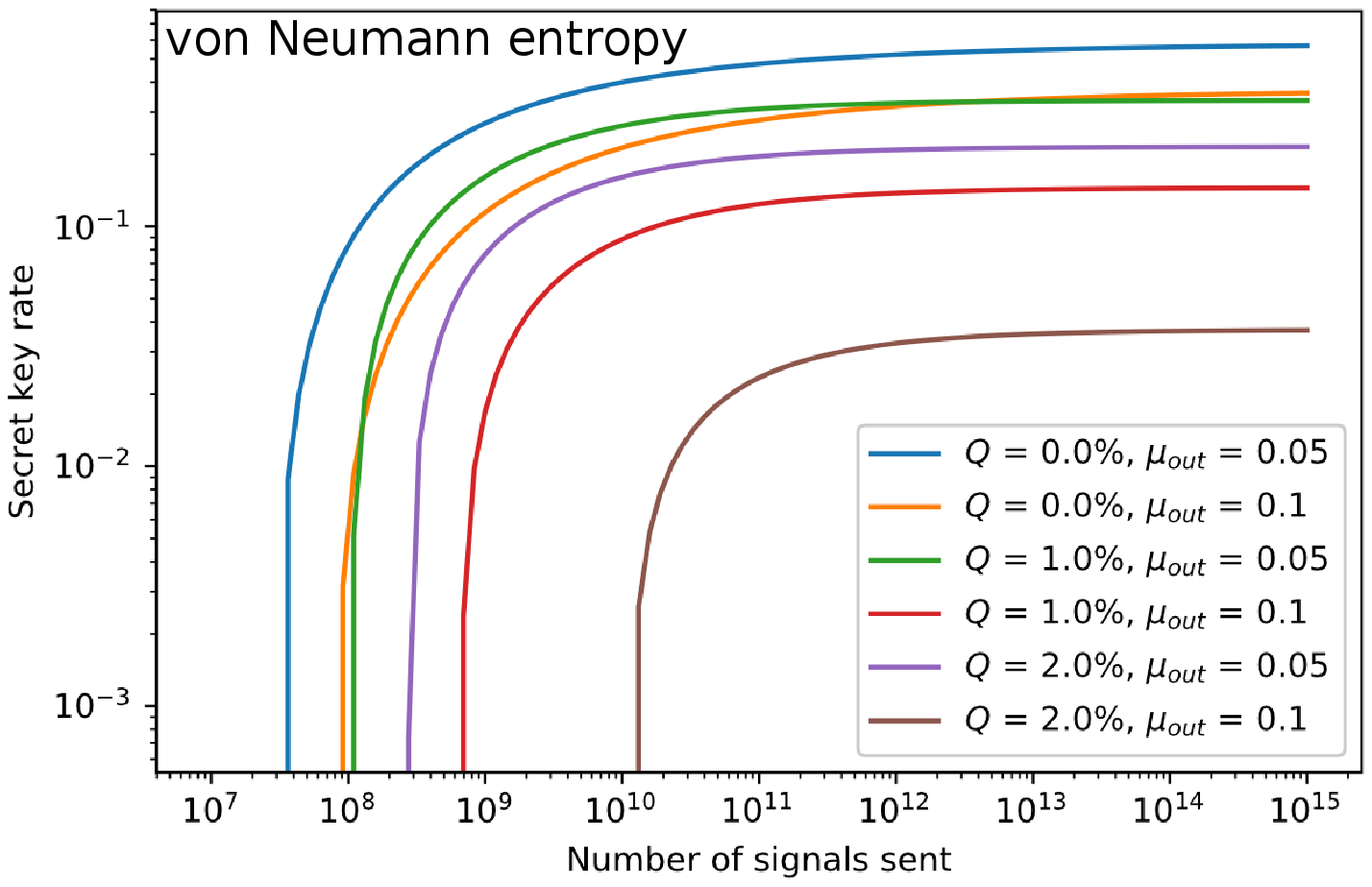}
	\end{subfigure}
	\begin{subfigure}[b]{\columnwidth}
	\includegraphics[width=\columnwidth]{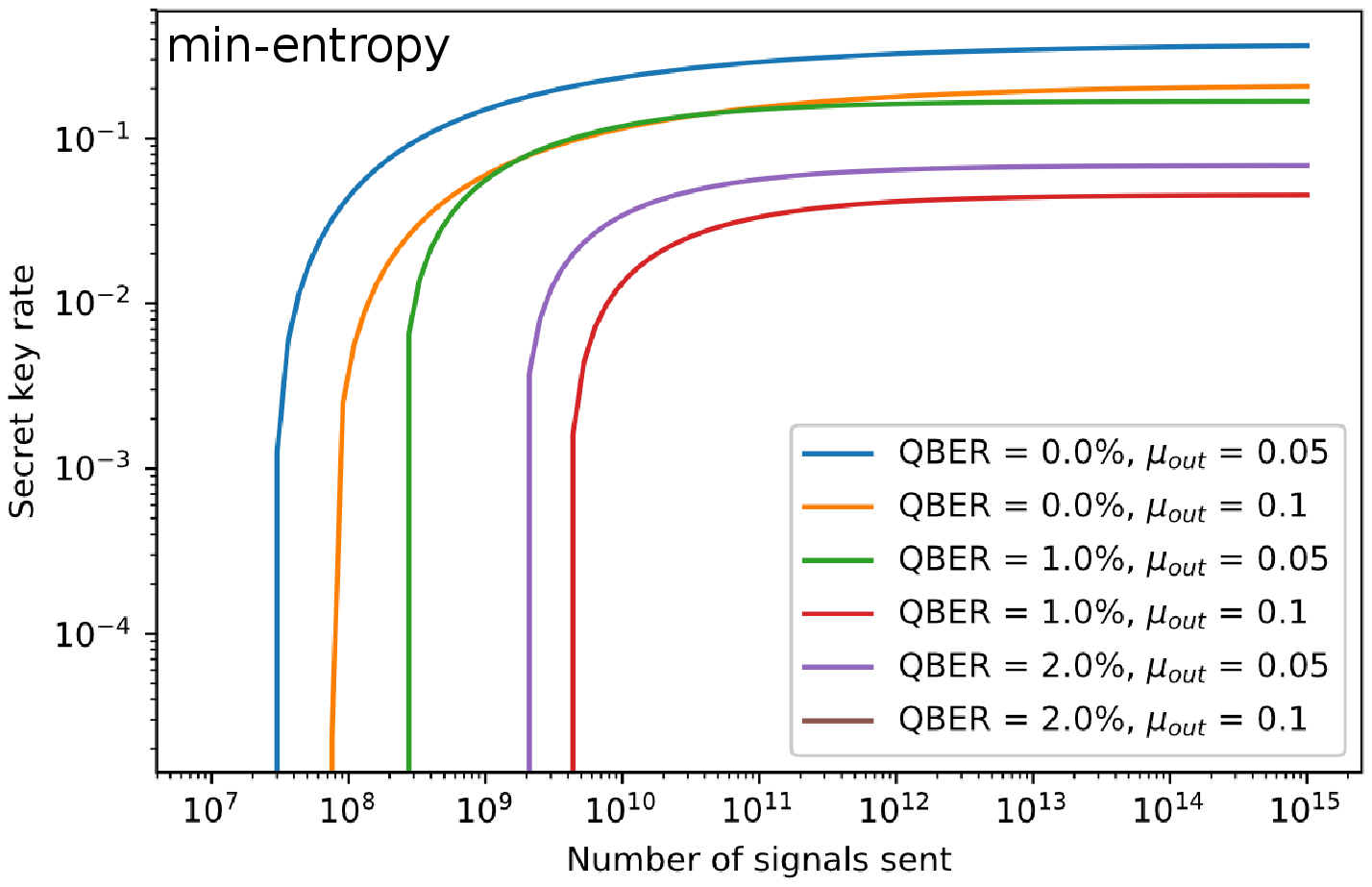}
	\end{subfigure}
    \caption{Nonasymptotic secret key rate per pulse for the BB84 protocol under a Trojan-horse attack for different values of error rate $Q$ and mean number of reflected Trojan photons $\mu_{\tout}$. The key rates are calculated using the von Neumann entropy (top) and the min-entropy (bottom). All lines are calculated using our numerical methods.}
    \label{fig:tha_finite}
\end{figure}

The key rate calculation for the asymptotic regime was considered in Ref.~\cite{Winick2017}. We consider the nonasymptotic regime, where we evaluate the security of an $\epssec$-secret and $\epscor$-correct QKD protocol with $\epssec = 10^{-10}$ and $\epscor = 10^{-15}$. The breakdown of the security parameters are tabulated in Tab.~\ref{tab:security_params}. 
% Combining all the parameters, we have $\epssec = 11 \varepsilon'$ for calculation with von Neumann entropy (Eq.~\eqref{eq:keyrate_eq_vnm} with SDPs~\eqref{eq:key_rate_nonasymptotic_appx_primal} and~\eqref{eq:lower_bound_nonasymptotic}) and $\epssec = 10 \varepsilon'$ for calculation with min-entropy (Eq.~\eqref{eq:keyrate_eq_min_entr} with SDP~\eqref{eq:keyrate_Hmin_dual_problem}). 
Fig.~\ref{fig:tha_finite} shows the secret key rate as a function of the number of pulses sent by Alice. The bounds from the von Neumann entropy calculations outperform the bounds from the min-entropy except for the case of zero QBER. We note this is the first time nonasymptotic security of a QKD protocol under Trojan-horse attack has ever been studied.

\end{document}